\documentclass[manuscript,screen,acmlarge]{acmart}

\AtBeginDocument{%
  \providecommand\BibTeX{{%
    \normalfont B\kern-0.5em{\scshape i\kern-0.25em b}\kern-0.8em\TeX}}}

\setcopyright{acmcopyright}
\copyrightyear{2022}
\acmYear{2022}
\acmDOI{TODO}

\renewcommand{\vec}[1]{\mathbf{#1}}

\usepackage{physics}

\newcommand{\R}{\mathbb{R}}
\newcommand{\N}{\mathbb{N}}
\newcommand{\Q}{\mathbb{Q}}
\newcommand{\Z}{\mathbb{Z}}

\newcommand{\dom}{{\rm dom} \;}

\newcommand{\calC}{\mathcal{C}}

\newcommand{\pair}[2]{\left\langle #1 , #2 \right\rangle}

\newcommand{\setr}[2]{\left\{\ #1 \ \left|\ #2 \right. \ \right\}}
\newcommand{\setl}[2]{\left\{\ \left. #1 \ \right|\ #2 \ \right\}}

\newcommand{\set}[1]{\left\{#1\right\}}
\newcommand{\seq}[1]{\left(#1\right)}

\newcommand{\calS}{\mathcal{S}}

\newcommand{\ignore}[1]{}

\usepackage[normalem]{ulem}
\usepackage{cancel}
\usepackage{xcolor}

\DeclareMathOperator{\sgn}{sgn}

\newcommand{\calM}{\mathcal{M}}

\newcommand{\hf}{\hat{f}}

\renewcommand{\va}{\vec{a}}
\newcommand{\vc}{\vec{c}}
\newcommand{\vd}{\vec{d}}
\newcommand{\ve}{\vec{e}}
\newcommand{\vr}{\vec{r}}

\newcommand{\vq}{\vec{q}}

\newcommand{\vo}{\vec{o}}
\newcommand{\vM}{\vec{M}}
\newcommand{\vA}{\vec{A}}
\newcommand{\vF}{\vec{F}}
\newcommand{\vp}{\vec{p}}
\newcommand{\vi}{\vec{i}}
\newcommand{\vx}{\vec{x}}
\newcommand{\vy}{\vec{y}}

\newcommand{\vz}{\vec{z}}
\newcommand{\vf}{\vec{f}}
\newcommand{\vg}{\vec{g}}
\newcommand{\vh}{\vec{h}}
\newcommand{\vH}{\vec{H}}
\newcommand{\vecv}{\vec{v}}
\renewcommand{\vu}{\vec{u}}
\newcommand{\vrho}{{\boldsymbol\rho}}
\newcommand{\vgamma}{{\boldsymbol\gamma}}
\newcommand{\vsigma}{{\boldsymbol\sigma}}

\renewcommand{\vb}{\vec{b}}
\newcommand{\bfr}{\mathbf{r}}
\newcommand{\bfp}{\mathbf{p}}

\newcommand{\Rp}{\mathbb{R}_{\geq 0}}

\newcommand{\reduce}{\ensuremath{\mathrm{red}}}

\newcommand{\producible}{\ensuremath{\mathsf{P}}}

\def\longrightharpoonup{\relbar\joinrel\rightharpoonup}
\def\longleftharpoondown{\leftharpoondown\joinrel\relbar}
\def\longrightleftharpoons{\mathop{\vcenter{\hbox{\ooalign{\raise1pt\hbox{$\longrightharpoonup\joinrel$}\crcr\lower1pt\hbox{$\longleftharpoondown\joinrel$}}}}}}
\def\rxn{\mathop{\rightarrow}\limits}  

\def\revrxn{\mathop{\rightleftharpoons}\limits}

\newcommand{\slto}{\to^1}      
\newcommand{\segto}{\rightsquigarrow}      

\usepackage[textsize=tiny]{todonotes}
\setuptodonotes{fancyline}

\usepackage{amsmath}
\usepackage{wasysym}   
\usepackage{ifpdf}
\usepackage{comment}
\usepackage{import}
\usepackage{caption}
\usepackage{subcaption}
\usepackage{wrapfig}
\usepackage{color}
\usepackage{url}
\usepackage{amsfonts}

\usepackage[capitalize]{cleveref}
\crefname{claim}{claim}{claims}

\crefname{lem}{Lemma}{Lemmas}
\crefname{defn}{Definition}{Definitions}
\crefname{thm}{Theorem}{Theorems}

\newtheorem{theorem}{Theorem}
\numberwithin{equation}{section}
\numberwithin{theorem}{section}
\newtheorem{definition}[theorem]{Definition}{\bfseries}{\itshape}
\newtheorem{corollary}[theorem]{Corollary}{\bfseries}{\itshape}
\newtheorem{lemma}[theorem]{Lemma}{\bfseries}{\itshape}
\newtheorem{proposition}[theorem]{Proposition}{\bfseries}{\itshape}
\newtheorem{observation}[theorem]{Observation}{\bfseries}{\itshape}
\newtheorem{thm}[theorem]{Theorem}{\bfseries}{\itshape}
\newtheorem{defn}[theorem]{Definition}{\bfseries}{\itshape}
\newtheorem{cor}[theorem]{Corollary}{\bfseries}{\itshape}
\newtheorem{lem}[theorem]{Lemma}{\bfseries}{\itshape}
{\bfseries}{\itshape}
\newtheorem{obs}[theorem]{Observation}{\bfseries}{\itshape}



\begin{document}

\title{Rate-Independent Computation in Continuous Chemical Reaction Networks}

\author{Ho-Lin Chen}
\email{holinchen@ntu.edu.tw}
\affiliation{%
  \institution{National Taiwan University}
  \streetaddress{MD-718}
  \city{Taipei}
  \country{Taiwan}
}

\author{David Doty}
\email{doty@ucdavis.edu}
\affiliation{%
  \institution{University of California, Davis}
  \streetaddress{One Shields Ave}
  \city{Davis}
  \country{USA}
}

\author{Wyatt Reeves}
\email{wreeves@math.harvard.edu}
\affiliation{%
  \institution{Harvard University}
  \streetaddress{1 Oxford St}
  \city{Cambridge}
  \country{USA}
}

\author{David Soloveichik}
\email{david.soloveichik@utexas.edu}
\affiliation{%
  \institution{University of Texas at Austin}
  \streetaddress{2501 Speedway}
  \city{Austin}
  \country{USA}
}


\begin{abstract}
  Understanding the algorithmic behaviors that are \emph{in principle} realizable in a chemical system is necessary for a rigorous understanding of the design principles of biological regulatory networks.
  Further, advances in synthetic biology herald the time when we will be able to rationally engineer complex chemical systems, and when idealized formal models will become blueprints for engineering.

  Coupled chemical interactions in a well-mixed solution are commonly formalized as chemical reaction networks (CRNs).
  However, despite the widespread use of CRNs in the natural sciences,
the range of computational behaviors exhibited by CRNs is not well understood.
  Here we study the following problem: what functions $f:\R^k \to \R$ can be computed by a 
  CRN, in which the CRN eventually produces the correct amount of the ``output'' molecule, no matter the rate at which reactions proceed?
    This captures a previously unexplored, but very natural class of computations: for example, the reaction $X_1 + X_2 \to Y$ can be thought to compute the function $y = \min(x_1, x_2)$.
    Such a CRN is robust in the sense that it is correct whether its evolution is governed by the standard model of mass-action kinetics, alternatives such as Hill-function or Michaelis-Menten kinetics, or other arbitrary models of chemistry that respect the (fundamentally digital) stoichiometric constraints (what are the reactants and products?).

We develop a reachability relation based on a broad notion of ``what could happen'' if reaction rates can vary arbitrarily over time.
Using reachability, we define \emph{stable computation} analogously to probability 1 computation in distributed computing,
and connect it with a seemingly stronger notion of rate-independent computation based on convergence in the limit $t \to \infty$ under a wide class of generalized rate laws.
Besides the direct mapping of a concentration to a nonnegative analog value,
we also consider the ``dual-rail representation'' that can represent negative values as the difference of two concentrations and allows the composition of CRN modules. 
We prove that a function is rate-independently computable if and only if it is piecewise linear (with rational coefficients) and continuous (dual-rail representation), or non-negative with discontinuities occurring only when some inputs switch from zero to positive (direct representation). 
The many contexts where continuous piecewise linear functions are powerful targets for implementation, combined with the systematic construction we develop for computing these functions,
demonstrate the potential of rate-independent chemical computation.
\end{abstract}

\maketitle


\section{Introduction}
\label{sec-intro}

Understanding the dynamic behaviors that are, in principle, achievable with chemical species interacting over time is crucial for engineering of complex molecular systems capable of diverse and robust behaviors.
The exploration of this space also helps to elucidate the constraints imposed upon biology by the laws of chemistry.
The natural language for describing the interactions of molecular species in a well-mixed solution is that of chemical reaction networks (CRNs), i.e., finite sets of chemical reactions such as $A + B \to A + C$.
The intuitive meaning of this expression is that a unit of chemical species $A$ reacts with a unit of chemical species $B$, producing a unit of a new chemical species $C$ and regenerating a unit of $A$ back.
Typically (in mass-action kinetics) the rate with which this occurs is proportional to the product of the amounts of the reactants $A$ and $B$.

Informally speaking we can identify two sources of computational power in CRNs.
First, the reaction \emph{stoichiometry} transforms some specific ratios of reactants to products.
For example, $X \rxn 2Y$ makes two units of $Y$ for every unit of $X$.
Second, in mass-action kinetics the reaction \emph{rate laws} effectively perform multiplication of the reactant concentrations.
In this work, we seek to disentangle the contributions of these two computational ingredients by focusing on the computational power of stoichiometry alone.
Besides fundamental scientific interest,
such rate-independent computation may be significantly easier to engineer than computation relying on rates (see \cref{sec:chemicalmotivation}).
Importantly, stoichiometry is robust---not requiring the tuning of reaction conditions, nor even the assumption that the solution is well-mixed.

In the \emph{discrete} model of chemical kinetics (see \Cref{sec:chemicalmotivation} for the distinction between the discrete and continuous models),
rate-independence is formally related to probability $1$ computation with passively mobile (i.e., interacting randomly) agents in distributed computing (the ``population protocols'' model~\cite{angluin2006passivelymobile,aspnes2007introduction}, see \Cref{sec:prev}). 
However, the continuous model of chemistry is most widely used, 
and is more applicable for engineering chemical computation where working with bulk concentrations remains the state of the art (see \Cref{sec:chemicalmotivation}).
This paper
formally articulates rate-independence in continuous CRNs and characterizes the computational power of stoichiometry in this model.

In the continuous setting
the amount of a species is a nonnegative real number representing its concentration (average count per unit volume).\footnote{Although the finite density of matter physically restricts what the largest concentration of any species could realistically be,
standard models of chemical kinetics focus on systems that are far from this bound,
mathematically allowing concentrations to be arbitrarily large.}
We characterize the class of real-valued functions computable by CRNs when reaction rates are permitted to vary arbitrarily 
(possibly adversarially)
over time.
Any computation in this setting must rely on stoichiometry alone.
How can rate laws ``preserve stoichiometry'' while varying ``arbitrarily over time''?
Formally, preserving stoichiometry means that if we reach state $\vd$ from state $\vc$, then $\vd = \vc + \vM \vu$ for 
some non-negative vector $\vu$ of reaction fluxes,
where the CRN's stoichiometry matrix $\vM$ maps those fluxes to the changes in species concentrations they cause.
(For example, flux $0.5$ of reaction $C+X \to C+3Y$ changes the concentrations of $C,X,Y$ respectively by $0,-0.5,+1.5.$)
Subject to this constraint, the widest class of trajectories that still satisfies the intuitive meaning of the reaction semantics can be described informally as follows:
(1) concentrations cannot become negative;
(2) all reactants must be present when a reaction occurs (e.g.,~if a reaction uses a catalyst\footnote{A species acts catalytically in a reaction if it is both a reactant and product: e.g.~$C$ in reaction $A + C \to B + C$. Note that executing this reaction without $C$ does not by itself violate condition (1).}, then the catalyst must be present);
(3) the causal relationships between the production of species is respected (e.g., if producing $A$ requires $B$ and producing $B$ requires $A$, then neither can ever be produced if both are absent)\footnote{See~\Cref{sec:mass-action-reachability} for examples showing that in the continuous setting conditions (2) and (3) are not mutually redundant.}.
This notion of ``allowed trajectories'' is formalized as Definition~\ref{defn:valid-rate-schedule}.

\begin{figure}
\centering 
\includegraphics[width=0.7\textwidth]{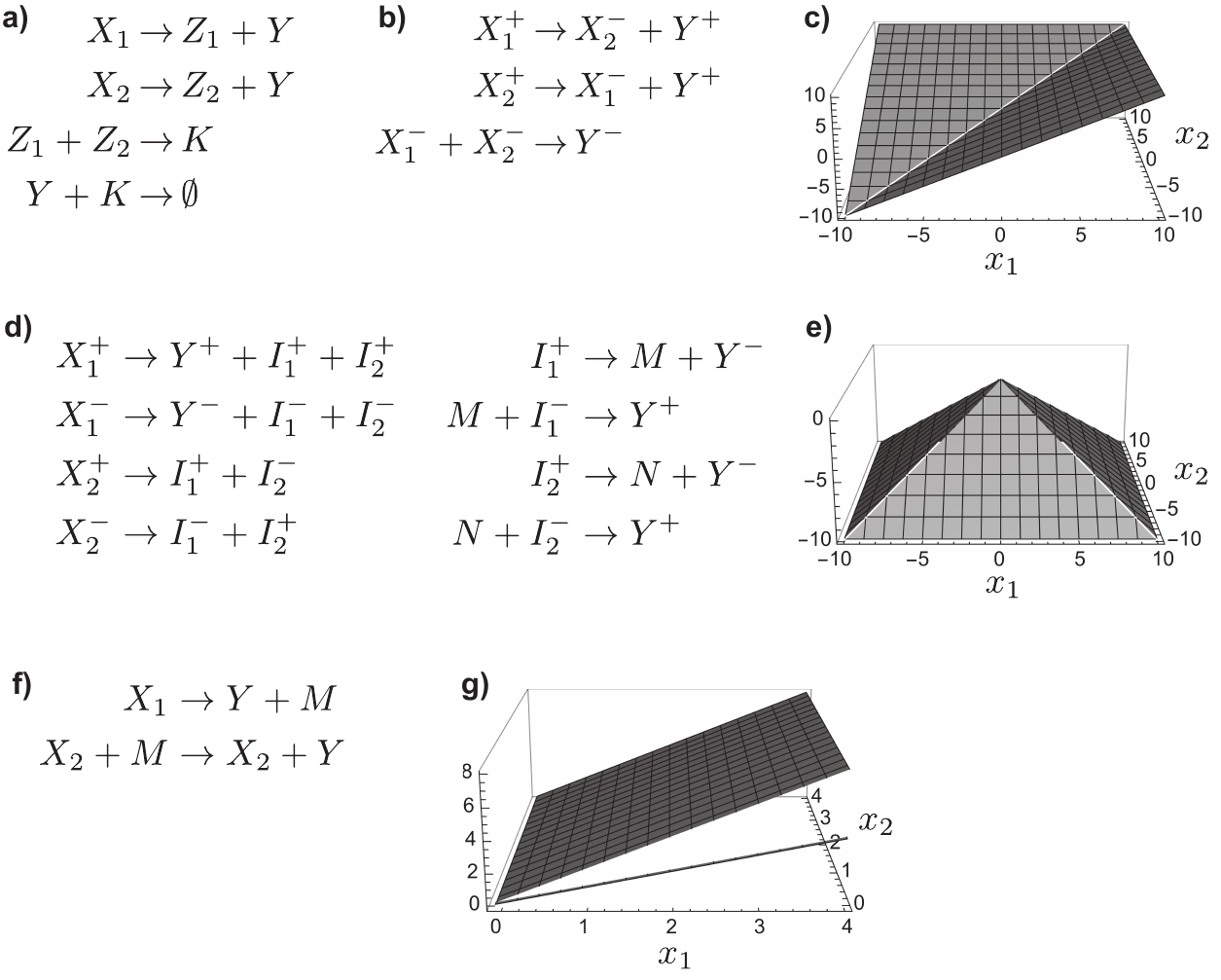}
\caption{Examples of rate-independent computation with chemical reaction networks. 
(a) Direct and (b) dual-rail CRNs computing the function $f(x_1, x_2) = \max(x_1,x_2)$ plotted in (c).
(d) Dual-rail CRN computing the function
$f(x_1, x_2) = x_1+\min (-x_1,x_2)-\max (x_1,x_2)$ plotted in (e).
(f) Direct CRN computing the (discontinuous but still positive-continuous, see \Cref{def:positive-continuous}) function $f(x_1, x_2) = x_1$ \emph{if $x_2 = 0$ and $2 x_1$ if $x_2 > 0$} plotted in (g).
}
\label{fig:ex} 
\end{figure}

The example shown in \cref{fig:ex}(a) illustrates the style of computation studied here.
Let $f: \Rp^2 \to \Rp$ be the max function $f(x_1, x_2) = \max(x_1,x_2)$ restricted to non-negative $x_1$ and $x_2$.
The CRN of \cref{fig:ex}(a) computes this function in the following sense.
Inputs $x_1$  and $x_2$ are given as initial concentrations of input species $X_1$ and $X_2$.
Then the CRN converges to $f$'s output value $\max(x_1,x_2)$ of species $Y$, under a very wide interpretation of rate laws.
Intuitively, the first two reactions must eventually produce $x_1 + x_2$ of $Y$, and $x_1$, $x_2$ of $Z_1$ and $Z_2$, respectively.
This is enforced by the stoichiometric constraint that the amount of $Z_1$ and $Y$ produced is equal to the amount of $X_1$ consumed (and analogously for the second reaction).
Stoichiometric constraints require the third reaction to produce the amount of $K$ that is the minimum of the amount of $Z_1$ and $Z_2$ eventually produced in the first two reactions.
Thus $\min(x_1,x_2)$ of $K$ is eventually produced.
Therefore, the fourth reaction eventually consumes $\min(x_1,x_2)$ molecules of $Y$ leaving $x_1 + x_2 - \min(x_1 , x_2) = \max(x_1 , x_2)$ of $Y$ behind.
We can imagine an adversary pushing flux through these four reactions in any devious stratagem (i.e., arbitrary rates),
yet unable to prevent the CRN from converging to the correct output, so long as applicable reactions must eventually occur.

We further consider the natural extension of such computation to handle negative real values.
The example shown in \cref{fig:ex}(b) computes $f(x_1, x_2) = \max(x_1,x_2)$ ($f: \R^2 \to \R$), graphed in (c).
In order to handle negative input and output values, we represent the value of each input and output by a pair of species using the so-called ``dual-rail'' representation.
For example, in state $\vc$, $x_1 = \vc(X_1^+)-\vc(X_1^-)$---i.e.~the difference between the concentrations of species $X_1^+$ and $X_1^-$.
Note that when $X_1^-$ and $X_2^-$ are initially absent, the CRN becomes equivalent to the first three reactions of Fig.~\ref{fig:ex}(a) under relabeling of species.
We do not need the last reaction of (a) because the output is represented as the difference of $Y^+$ and $Y^-$ by our convention.
For the argument that the computation is correct even if $X_1^-$ and $X_2^-$ are initially present, we refer the reader to 
the proof of \cref{cor:dual-rail-max-computable} in 
\cref{negative-piecewise-linear-implies-computable}. 

In addition to handling negative values, the dual-rail representation has the benefit of allowing composition.
Specifically, the dual-rail representation allows a CRN to never consume its output species (e.g.~rather than consuming $Y^+$, it can produce $Y^-$). 
This monotonicity in the production of output allows directly composing CRN computations simply by concatenating CRNs and relabeling species (e.g.~to make the output of one be input to the other).
Since the upstream CRN never consumes its output species, the downstream CRN is free to consume them without interfering with the upstream computation. 
Since the class of functions computable by dual-rail CRNs ends up being invariant to whether or not they are allowed to consume their output, our results imply that dual-rail computation is composable without sacrificing computational power (see \cref{subsec:dual-rail-defn}).

\subsection{Summary of Main Results}

Our first contribution is to define a reachability relation that captures the broadest reasonable notion of ``what could happen''  and is of independent interest.
Although concentration trajectories of mass-action kinetics (and other standard rate laws) follow smooth curves, we base the reachability relation on taking simple-to-analyze straight-line paths.
Theorem~\ref{thm:valid-reachable-implies-segment-reachable} shows that this notion of reachability is exactly equivalent to satisfying the three intuitive properties described above,
(1) nonnegative concentrations,
(2) reactions require their reactants present, and
(3) respecting causal relationships between the production of species.
Thus our reachability relation has all reasonable rate laws as special cases; i.e., if any of them can reach a state, so can our reachability relation.

The reachability relation allows us to formally define \emph{stable computation} in \Cref{def:direct-computation},
analogously to similar definitions of probability 1 computation in discrete systems~\cite{angluin2006passivelymobile,CheDotSolNaCo14}.
Stable computation allows us to delineate when a function \emph{cannot} be computed rate-independently. 
Roughly, unless the system is stably computing, then an adversary can always push it ``far'' away from the correct output (\cref{lem:not-stably-compute-implies-can-reach-far-from-correct}),
precluding the system from being reasonably rate-independent.

For the positive direction, the CRN should converge to the correct output no matter the reaction rates.
While a CRN that does not stably compute is not rate-independent (which is sufficient for negative results),
the positive direction does not directly follow from stable computation for continuous systems.
Indeed we show examples of CRNs that stably compute a function by our definition, 
yet under standard mass-action kinetics fail to converge to the correct output; 
see \Cref{sec:fair-computation}.
Instead, we capture a very strong notion of ``convergence despite perturbations'' in \emph{fair computation}
(\Cref{defn:fair-computation}),
based on generalized rate laws (so-called fair rate schedules, \Cref{defn:fair-rate-schedule}).
A CRN that fairly computes converges to the correct output as time $t \to \infty$ under any trajectory satisfying the three intuitive conditions above,
plus an additional requirement that reactions \emph{do} occur when applicable.
(In particular, mass-action (\Cref{cor:mass-action-convergence-stable-computation}) satisfies these conditions, but the range of rate laws satisfying the conditions is much broader.)
Luckily,  stable computation and fair computation can be tightly connected, 
and we show that for a special class of CRNs we call \emph{feedforward} 
(\Cref{defn:feedforward})
the two definitions coincide.
In other words, a feedforward CRN stably computes a function if and only if it fairly computes the function 
(\cref{lem:fair-computation-implies-stable-computation,lem:stable-computation-implies-fair-computation-feedforward}).
We show that all functions stably computable by CRNs are computable by feedforward CRNs 
(\Cref{lem-nonnegative-piecewise-linear-implies-computable,lem-piecewise-linear-implies-computable}),
implying that the class of functions computable by CRNs under either definition---stable computation or fair computation---is identical.
In other words, we can freely work with the simpler definition of stable computation,
knowing that we are actually reasoning about a very general notion of rate-independence.

The above line of reasoning leads us to conclude that exactly the functions that are positive-continuous, piecewise linear (direct) or continuous, piecewise linear (dual-rail) can be rate-independently computed 
(\Cref{thm:computable-characterization,thm:computable-characterization-dual-rail}).
Positive-continuous means that the only discontinuities occur on a ``face'' of $\Rp^k$---i.e., the function may discontinuously jump only at a point where some input goes from $0$ to positive.
We already saw a simple example of a continuous, piecewise linear function (max function, \Cref{fig:ex}(a,b,c)).
\Cref{fig:ex}(d,e) shows a more complex example and a CRN that computes it.
\Cref{fig:ex}(f,g) shows a discontinuous but positive-continuous function and a CRN that computes it.
Although our work shows that the computational power of rate-independent CRNs is limited, the power of the computable class of functions should not be underestimated.
For example, allowing a fixed non-zero initial concentration of non-input species (see \Cref{sec:initial-context}), 
such CRNs are equivalent to ReLU neural networks---arguably the most widely used type of neural networks in machine learning~\cite{vasic2022programming}.

\subsection{Chemical Motivation}
\label{sec:chemicalmotivation}
Traditionally CRNs have been used as a descriptive language to analyze naturally occurring chemical reactions,
as well as various other systems with a large number of interacting components such as gene regulatory networks and animal populations.
However, CRNs also constitute a natural choice of programming language for engineering artificial systems.
For example, nucleic-acid networks can be rationally designed to implement arbitrary chemical reaction networks~\cite{SolSeeWin10,cardelli2011strand,chen2013programmable,srinivas2017enzyme}.
Thus, since in principle any CRN can be physically built, hypothetical CRNs with interesting behaviors are becoming of more than theoretical interest.
One day artificial CRNs may underlie embedded controllers for biochemical, nanotechnological, or medical applications, where environments are inherently incompatible with traditional electronic controllers.
However, to effectively program chemistry, we must understand the computational power at our disposal.
In turn, the computer science approach to CRNs is also beginning to generate novel insights regarding natural cellular regulatory networks~\cite{cardelli2012cell}.

At the fine-grained level of detail, chemistry is discrete and stochastic.
This level is typically modeled by discrete CRNs, where the state is a vector of nonnegative \emph{integers} representing the counts of each species in the given reaction vessel,
and reactions are modeled by a Markov jump process~\cite{Gillespie77}.
The continuous model is governed by a system of mass-action ordinary differential equations, which can be derived as a limiting case of the discrete model when volume and counts are large~\cite{kurtz1972relationship}.%
\footnote{
    The exact statement of Kurtz's convergence result~\cite{kurtz1972relationship} is beyond the scope of this paper.
    It considers taking a discrete CRN with initial integer molecular counts given by vector $\vc \in \N^k$ in volume $V > 0$,
    then ``scaling up'' by factor $n\in\N$, i.e., considering the discrete CRN with initial state $n \cdot \vc$ in volume $n \cdot V$.
    The result, stated very roughly, is that with high probability the $n$-scaled CRN has a trajectory (dividing discrete counts by $n \cdot V$ to convert to units of concentration) that stays close to the real-valued mass-action concentration trajectory, but only for time $O(\log n)$.
    An example CRN where this time bound is tight is $A+B \to A+B+Y,\quad A \to 2A,\quad B \to \emptyset$, starting with $1A, 1B$.
    In mass-action, 
    the concentrations of $A$ and $B$ at time $t$ are respectively $e^t$ and $e^{-t}$,
    whose product is the constant 1, so the first reaction produces $Y$ at a unit rate forever.
    However, scaling up to $n A, n B$, the discrete CRN consumes all $B$ in $O(\log n)$ time, at which point production of $Y$ halts.
    See~\cite{lathrop2020population, ppsim} for other example CRNs 
    for which the two models diverge after sufficient time.
}
While of physical primacy, the discrete model can be less suitable for reasoning about feasible chemical algorithms.
Many algorithms in the discrete model rely on a single molecule (called a ``leader'') to coordinate computation~\cite{angluin2006fast}.
Whether the initial state is assumed to have a leader, or the CRN is designed to eliminate all but one copy of the leader species (``leader election''), such algorithms relying on single-molecule behavior are currently infeasible since any single molecule can get damaged or become effectively lost.

An important reason for our focus on \emph{stoichiometric computation} is that algorithms relying only on stoichiometry make easier design targets.
The rates of reactions are real-valued quantities that can fluctuate with reaction conditions 
such as 
temperature, while the 
stoichiometric coefficients
are immutable whole numbers set by the nature of the reaction.
Methods for physically implementing CRNs naturally yield systems with digital stoichiometry that can be set exactly~\cite{SolSeeWin10,cardelli2011strand},
whereas these methods often suffer from imprecise control over reaction rates~\cite{chen2013programmable,srinivas2017enzyme}.
Further, relying on specific rate laws can be problematic:
many systems do not apparently follow mass-action rate laws and chemists have developed an array of alternative rate laws such as Michaelis-Menten 
(modeling enzymes) and Hill-function kinetics (widely used for gene regulation).\footnote{It is generally supposed that chemical reactions would follow mass-action if properly decomposed into truly elementary reactions and the solution is well-mixed.
For example, Michaelis-Menten and Hill-function kinetics can be derived as a limiting case of mass-action when the reaction is initiated and completed at vastly different time scales.}
It is well-known that cells are not well-mixed, and many models have been developed to take space into account (e.g., reaction-diffusion~\cite{kondo2010reaction}).
Moreover, robustness of rate laws is a recurring motif in systems biology due to much evidence that biological regulatory networks tend to be robust to the form of the rate laws and the rate parameters~\cite{barkal1997robustness}.
Thus we are interested in what computations can be understood or engineered without regard for the reaction rate laws.

\subsection{Related Works} \label{sec:prev}

An earlier conference version of this paper appeared as~\cite{CheDotSol14}.
Besides replacing a number of informal arguments with rigorous proofs, this journal version also expands and generalizes the results of the conference version.
For example, we introduce new machinery for representing and manipulating trajectories as linear objects (piecewise linear paths).
We also define a broad class of rate laws, formalized by Definition~\ref{defn:valid-rate-schedule}, which captures mass-action kinetics and all other known rate laws such as Michaelis-Menten and Hill-function kinetics, and prove that our definition of reachability is as general as any in this class.
For the constructive part, this version also generalizes Lemma 3.4 of \cite{CheDotSol14} (in addition to correcting its proof) by introducing feedforward CRNs and proving that correct computation in our setting implies 
convergence 
under any ``reasonable'' rate law (one that produces a \emph{fair} schedule of rates; \cref{defn:fair-rate-schedule})
for any 
feedforward CRN
(\cref{lem:fair-rate-schedule-converges}).

The relationship between the discrete and continuous CRN models is a complex and much studied one in the natural sciences~\cite{samoilov2006deviant}.
The computational abilities of discrete CRNs have been investigated more thoroughly than of continuous CRNs, and have been shown to have a surprisingly rich computational structure.
Of most relevance here is the work in the discrete setting showing that the class of functions that can be computed depends strongly on whether the computation must be correct, or just likely to be correct (under the usual stochastic kinetics)---which is the discrete version of the distinction between rate-independent and rate-dependent computation.
While Turing universal computation is possible with an arbitrarily small, non-zero probability of error over all time~\cite{SolCooWinBru08},
forbidding error altogether limits the computational power:
Error-free computation by stochastic CRNs is limited to semilinear predicates and functions~\cite{angluin2006passivelymobile,CheDotSolNaCo14}.
(Intuitively, semilinear functions are expressible as a finite union of affine functions, with ``simple, periodic'' domains of each affine function~\cite{CheDotSolNaCo14}.)
The study of error-free computation in discrete CRNs is heavily based on the results first developed for a model of distributed computing called population protocols~\cite{angluin2006passivelymobile,aspnes2007introduction}.
We formally refer to our notion of rate-independent computation as \emph{stable computation} in direct reference to the analogous notion in population protocols.

While our notion of rate-independent computation is the natural extension of deterministic computation in the discrete model, there are many differences between the two settings.
As mentioned above, many discrete algorithms such as those that rely on a single ``leader'' molecule fail to work in the continuous setting,
and some functions like distinguishing between even and odd molecular counts do not make sense.
Broadly speaking, the proof techniques appear to require very different machinery,
and the importance of stable computation itself needs substantial justification in the continuous model (as the examples shown at the beginning of  \cref{sec:fair-computation} demonstrate).

Continuous CRNs have been proven to be Turing universal under mass-action rate laws~\cite{fages2017strong},
a consequence of the surprising computational power of polynomial ODEs~\cite{bournez2017odes}.
In ODEs without the CRN semantics, there is no natural notion of stoichiometry and thus no notion of rate-independence analogous to ours.
In chemistry, the same physical process (a reaction) is responsible for multiple monomials across multiple ODEs,
which justifies these monomials being exactly the same or in fixed ratios (corresponding to obeying reaction stoichiometry).
Such forced relationships do not seem natural for more general polynomial ODEs that do not correspond to chemical reactions.
\footnote{For example, consider the reaction $A \to 2B$, with ODEs $\dot{a} = -a$ and $\dot{b} = 2a$.
One can imagine a ``chemical'' adversary adjusting the rate of the reaction to speed it up or slow it down, but what the adversary \emph{cannot} control is that to consume $x$ amount of $A$ requires producing exactly $2x$ amount of $B$, and vice versa.
This connection between the rates of consumption of $A$ and production of $B$ does not have an obvious counterpart in more general polynomial ODEs and analog computational models.}

Our notion of reachability 
(\cref{defn-reachable-segment})
is intended to capture a wide diversity of possible rate laws.
Generalized rate laws (extending mass-action, Michaelis-Menten, etc) have been previously studied, although not in a computational setting. 
For example, certain conditions were identified on global convergence to equilibrium based on properties intuitively similar to ours~\cite{angeli2006structural}.
A related idea in the literature, generalizing mass-action, is differential inclusion~\cite{gopalkrishnan2013projection}. 
In that model, the mass-action rate constants are not fixed to be particular real numbers constant over time, but instead can vary over time within some bounded interval $[l,u]$ fixed in advance, with $0 < l \leq u < \infty$.
Another related idea is the notion of a reaction system~\cite{fages2015inferring},
which generalizes even beyond mass-action, allowing reaction rates to be an (almost) arbitrary function of species concentrations.
\footnote{
    Our notion of valid rate schedules in \cref{defn:valid-rate-schedule} is even more general than a reaction system  in that a valid rate schedule does not require a reaction's rate to be a function of species concentrations, for instance allowing an adversary to visit the same state twice but apply different reaction rates each time.
}
Other generalized rate laws have been defined as well~\cite{angeli2007petri,degrand2020graphical}.

Since the original publication of the conference version of this paper~\cite{CheDotSol14}, a number of works have used our framework.
A key concept in capturing rate-independent computation is the reachability relation (segment-reachability, \cref{defn-reachable-segment}).
Reference~\cite{case2018reachability} showed that, given two states, deciding whether one is reachable from the other is solvable in polynomial time.
This contrasts sharply with the hardness of the reachability problem for discrete CRNs
which, although computable~\cite{mayr1984algorithm}, 
is not even primitive recursive~\cite{leroux2021reachability,czerwinski2021reachability}.
(These results were proven using the terminology of the equivalent models of Petri nets/vector addition systems.)

The question of deciding whether a given CRN is rate-independent was studied in~\cite{degrand2020graphical}.
The work provides sufficient graphical conditions on the structure of the CRN that ensure rate-independence for the whole CRN or only for certain output species. 
Interestingly, the authors of~\cite{degrand2020graphical} applied this method to the Biomodels repository of curated CRNs of biological origin and found a number of CRNs that satisfy the rate-independence conditions.

An important motivation for the dual-rail representation in this work is to allow composition of rate-independent CRN modules (\cref{subsec:dual-rail-defn}).
Such rate-independent modules can be composed into overall rate-independent computation simply by 
concatenating their chemical reactions and relabeling species (such that the output species of the first is the input species of the second, and all other species are distinct).
In contrast, rate-independent composition with the direct (non-dual rail) representation, 
introduces an additional ``superadditivity'' constraint that for all input vectors $\vx$ and $\vx'$, $f(\vx) + f(\vx') \leq f(\vx + \vx')$~\cite{chalk2021composable}.
Thus, for example, the non-superadditive max function (Figure~\ref{fig:ex}) provably cannot be composably computed with a rate-independent CRN in the direct representation.
Composable computation has also been characterized in the discrete model~\cite{severson2021composable,hashemi2020composable}.

Other input encodings have been considered besides direct and dual-rail.
For example, the so-called ``fractional encoding'' encodes a real number between 0 and 1 as a ratio $\frac{x_1}{x_0+x_1}$ where $x_0,x_1$ are concentrations of two input species~\cite{salehi2017chemical}.
Other notions of chemical ``rate-independence'' include CRNs that work independently of the rate law as long as there is a separation into fast and slow reactions~\cite{senum2011rate}.
For a detailed survey on computation with CRNs (both continuous and discrete),
see~\cite{brijder2019computing}.


\section{Defining Reachability in Chemical Reaction Networks}
\label{sec-prelim}

\subsection{Chemical Reaction Networks}
\label{subsec-prelim-defs}

We first explain our notation for vectors of concentrations of chemical species, and then formally define chemical reaction networks.

Given a finite set $F$,
let $\R^F$ denote the set of functions $\vc: F \to \R$.
We view $\vc$ equivalently as a vector of real numbers indexed by elements of $F$.
Given $x \in F$, we write $\vc(x)$, or sometimes $\vc_x$, to denote the real number indexed by $x$.
The notation $\Rp^F$ is defined similarly for nonnegative real vectors.
Throughout this paper, let $\Lambda$ be a finite set of chemical \emph{species}.
Given $S\in \Lambda$ and $\vc \in \Rp^\Lambda$, we refer to $\vc(S)$ as the \emph{concentration of $S$ in $\vc$}.
For any $\vc\in \Rp^\Lambda$, let $[\vc] = \{S \in \Lambda \ |\ \vc(S) > 0 \}$, the set of species \emph{present} in $\vc$
(a.k.a., the \emph{support} of $\vc$).
We write $\vc \leq \vc'$ to denote that $\vc(S) \leq \vc'(S)$ for all $S \in \Lambda$.
Given $\vc,\vc' \in \Rp^\Lambda$, we define the vector component-wise operations of addition $\vc+\vc'$, subtraction $\vc-\vc'$, and scalar multiplication $x \vc$ for $x \in \R$.
If $\Delta \subset \Lambda$, we view a vector $\vc \in \Rp^\Delta$ equivalently as a vector $\vc \in \Rp^\Lambda$ by assuming $\vc(S)=0$ for all $S \in \Lambda \setminus \Delta.$
For $\Delta \subset \Lambda$, we write $\vc  \upharpoonright \Delta$ to denote $\vc$ \emph{restricted to} $\Delta$; in particular, $\vc  \upharpoonright \Delta = \vec{0} \iff (\forall S\in\Delta)\ \vc(S)=0.$
(We use the convention that $\vc \upharpoonright \emptyset = \vec{0}$ for all states $\vc$.)

A \emph{reaction} over $\Lambda$ is a pair $\alpha = \langle \bfr,\bfp \rangle \in \N^\Lambda \times \N^\Lambda$, 
such that $\vr \neq \vp$,
specifying the stoichiometry of the reactants and products, respectively.\footnote{It is customary to define, for each reaction, a \emph{rate constant} $k \in \R_{>0}$ specifying a constant multiplier on the mass-action rate (i.e., the product of the reactant concentrations), but as we are studying CRNs whose output is independent of the reaction rates, we leave the rate constants out of the definition.} 
For instance, given $\Lambda=\{A,B,C\}$, the reaction $A+2B \to A+3C$ is the pair $\pair{(1,2,0)}{(1,0,3)}.$
We represent reversible reactions such as $A \revrxn B$ as two irreversible reactions $A \to B$ and $B \to A$.
In this paper, we assume that $\bfr \neq \vec{0}$, i.e., we have no reactions of the form $\emptyset \to \ldots$.\footnote{We allow high order reactions; i.e., those that have more than two reactants.
Such higher order reactions could be eliminated from our constructions using the transformation that replaces $S_1 + S_2 + \ldots + S_n \to P_1 + \ldots + P_m$ with bimolecular reactions $S_1 + S_2 \revrxn S_{12}, S_{12} + S_3 \revrxn S_{123}, S_{123} + S_4 \revrxn S_{1234}, \ldots, S_n + S_{12 \ldots n-1} \to P_1 + \ldots + P_m$.
}
A \emph{(finite) chemical reaction network (CRN)} is a pair $\calC=(\Lambda,R)$, where $\Lambda$ is a finite set of chemical \emph{species},
and $R$ is a finite set of reactions over $\Lambda$.
A \emph{state} of a CRN $\calC=(\Lambda,R)$ is a vector $\vc \in \Rp^\Lambda$.
Given a state $\vc$ and reaction $\alpha=\pair{\bfr}{\bfp}$, we say that $\alpha$ is \emph{applicable} in $\vc$ if $[\bfr] \subseteq [\vc]$ (i.e., $\vc$ contains positive concentration of all of the reactants).
If no reaction is applicable in state $\vc$, we say $\vc$ is \emph{static}.
We say a species $S$ is \emph{produced} in reaction $\langle  \vr,\vp\rangle$ if $\vr(S) < \vp(S)$,
and
\emph{consumed} if $\vr(S) > \vp(S)$.
(Note that a catalyst, such as $C$ in the reaction $C+X \to C+Y$, is neither produced nor consumed.)

\subsection{Segment Reachability}
\label{subsec-reachability}

In the previous section we defined the syntax of CRNs.
Toward studying rate-independent computation,
we now want to define the semantics of what ``could happen'' if reaction rates can vary arbitrarily over time.
This is captured by a notion of \emph{reachability}, which is the focus of this section.
Intuitively, $\vd$ is reachable from $\vc$ if applying some amount of reactions to $\vc$ results in $\vd$, such that no reaction is ever applied when any of its reactants are concentration 0.
Formalizing this concept is a bit tricky and constitutes one of the contributions of this paper.
Intuitively, we'll think of reachability via straight line segments.
This may appear overly limiting; after all mass-action and other rate laws trace out smooth curves.
However in this and subsequent sections we show a number of properties of our definition that support its reasonableness.

Throughout this section, fix a CRN $\calC = (\Lambda,R)$.
All states $\vc$, etc., are assumed to be states of $\calC$.
We define the 
$|\Lambda| \times |R|$ \emph{reaction stoichiometry matrix} $\vM$ such that,
for species $S \in \Lambda$ and reaction $\alpha = \langle \vr,\vp \rangle \in R$,
$\vM(S,\alpha) = \vp(S) - \vr(S)$
is the net amount of 
$S$ produced by $\alpha$ (negative if $S$ is consumed).\footnote{Note that $\vM$ does not fully specify $\calC$, since catalysts are not modeled: reactions $Z + X \to Z + Y$ and $X \to Y$ both correspond to the column vector $(-1,1,0)^\top$.}
For example, if we have the reactions $X \to Y$ and $X + A \to 2X + 3Y$, and if the three rows correspond to $X$, $A$, and $Y$, in that order, then
  $$
    \vM =
    \left(
      \begin{array}{cc}
        -1 & 1 \\
        0 & -1 \\
        1 & 3 \\
      \end{array}
    \right)
  $$

\begin{definition}\label{defn-reachable-line}
State $\vd$ is \emph{straight-line reachable (aka $1$-segment reachable)} from state $\vc$, written $\vc \slto \vd$, if $(\exists \vu \in \Rp^R)\ \vc + \vM \vu = \vd$
and $\vu(\alpha) > 0$ only if reaction $\alpha$ is applicable at $\vc$.
In this case write $\vc \slto_\vu \vd$.
\end{definition}

Intuitively, by a single segment we mean running the reactions applicable at $\vc$ at a constant (possibly 0) rate to get from $\vc$ to $\vd$.
In the definition, $\vu(\alpha)$ represents the flux of reaction $\alpha \in R$.

The next definition is used in our main notion of reachability,
which uses either a finite number of straight lines,
or infinitely many so long as they converge to a single state.

\begin{definition}
\label{defn-reachable-lines}
Let $k \in \N \cup \{\infty\}$.
State $\vd$ is \emph{$k$-segment reachable} from state $\vc$, written $\vc \segto^k \vd$, 
if $(\exists \vb_0, \dots, \vb_{k})\ \vc  = \vb_0 \slto \vb_1 \slto \vb_2 \slto \dots  \slto \vb_{k}$,
with $\vb_k = \vd$ 
if $k \in \N$,
or $\lim\limits_{i \to \infty} \vb_i = \vd$ if $k = \infty$.
\end{definition}

\begin{definition}
\label{defn-reachable-segment}
State $\vd$ is \emph{segment-reachable} 
(or simply \emph{reachable}) 
from state $\vc$, written $\vc \segto \vd$, if $(\exists k\in\N \cup \{\infty\})\ \vc \segto^k \vd$.
\end{definition}

For example, suppose the reactions are $X \to C$ and $C + Y \to C + Z$, and we are in state $\{0 C, 1 X, 1 Y, 0 Z\}$.
With straight-line segments, any state with a positive amount of $Z$ must be reached in at least two segments: first to produce $C$, which allows the second reaction to occur, and then any combination of the first and second reactions.
For example,
$\{0 C$, $1 X$, $1 Y$, $0 Z\}$  $\slto$  $\{0.1 C$, $0.9 X$, $1 Y$, $0 Z\}$ $\slto$  $\{1 C$, $0 X$, $0 Y$, $1 Z\}$.
This is a simple example showing that more states are reachable with $\segto$ than $\slto$.
Often Definition~\ref{defn-reachable-segment} is used implicitly, when we make statements such as, ``Run reaction 1 until $X$ is gone, then run reaction 2 until $Y$ is gone'', which implicitly defines two straight lines in concentration space.

Although more effort will be needed to justify its reasonableness (see \Cref{sec:mass-action-reachability}), segment-reachability will serve as the main notion of reachability in this paper.

\subsection{Bound on Number of Required Line Segments in Segment Reachability}

It may seem that we can never achieve the ``full diversity'' of states reachable with 
an infinite number of line segments
if we use only a bounded number of line segments.
However, Theorem~\ref{thm-reachable-segment-bound} shows that increasing the number of straight-line segments beyond a certain point does not make any additional states reachable.
Thus using a few line segments captures all the states reachable with arbitrarily many line segments, 
and in fact even in the limit of infinitely many line segments.

In order to prove Theorem~\ref{thm-reachable-segment-bound}, 
we first develop important machinery for representing and manipulating paths under $\segto$.
Note that reachability is closed under addition and scaling in the sense that if $\vc \segto \vd$ and $\vc' \segto \vd'$ then $\alpha \vc + \beta \vc' \segto \alpha \vd + \beta \vd'$ for all $\alpha, \beta \in \R_{\ge 0}$. 
The following definition captures this property by defining a linear space of all paths. 
This machinery will also be key to proving the piecewise linearity of the computed function in Section~\ref{subsec-negative-computable-implies-piecewise-linear}.

\newcommand{\prepathsset}{\ensuremath{\Psi}}

\begin{defn}\label{def:prepaths}

Let $\calC = (\Lambda, R)$ be a CRN with species $\Lambda$ and reactions $R$.
For $n \in \N$,
we define a linear map $\vx_n: (\R^\Lambda \times \bigtimes_{i = 1}^\infty \R^R) \to \R^\Lambda$, which takes $\vgamma = (\vx_0, \vu_1, \vu_2, \ldots)$
representing an initial state $\vx_0$ and reaction flux vectors $\vu_1, \vu_2, \ldots$,
and produces
\[\vx_n(\vgamma) = \vx_0 + \sum_{i = 1}^n \vM \vu_i,\]
which intuitively is the state reached after traversing the first $n$ line segments.
Let 
$\prepathsset$
be the set of $\vgamma$ for which $\lim\limits_{n \to \infty} \vx_n(\vgamma)$ converges.
We call elements of 
$\prepathsset$
\emph{prepaths}. 
\end{defn}

Definition~\ref{def:prepaths} allows a prepath to be essentially any sequence of vectors in the linear subspace spanned by reaction vectors. The next definition restricts the vectors with three physical constraints: species concentrations are nonnegative, reaction fluxes are nonnegative (i.e., reactions can only go one way, turning reactants into products), and reactions cannot occur if any reactant is 0.

\begin{defn}\label{def:paths}
Let $\Gamma_\infty$ be the subset of 
$\prepathsset$
consisting of vectors $\vgamma = (\vec{x}_0, \vec{u}_1, \vec{u}_2, \ldots)$ satisfying the following conditions for all $n\in\N$:

\begin{enumerate}
    \item 
    $\vec{x}_n(\vgamma) \in \Rp^\Lambda$.

    \item
    $\vu_n \in \Rp^R$.
    
    \item 
    every reaction with positive flux in $\vec{u}_{n+1}$ is applicable at $\vec{x}_n$.
\end{enumerate}
We call an element of $\Gamma_\infty$ a \emph{piecewise linear path} or sometimes just a \emph{path}.
\end{defn}

Definitions~\ref{def:prepaths} and~\ref{def:paths} allow an infinite sequence of reaction flux vectors (each corresponding to a straight line in the definition of $1$-segment reachability).
A finite number of straight lines can be specified by letting all but finitely many $\vu_i = \vec{0}$.
The next definition bounds how many can be nonzero.

\begin{defn}
For $k\in \N$, define $\Gamma_k$ to be the subset of $\Gamma_\infty$ consisting of all paths $\vgamma = (\vx_0, \vu_1, \vu_2, \ldots)$ such that $\vu_i = \vec{0}$ for all $i > k$. 
Say that a path is \emph{finite} if it is contained in $\Gamma_k$ for some $k \in \N$.
\end{defn}

Intuitively, $\Gamma_\infty$ is the space of all of the valid piecewise linear paths that the system can take starting from any given initial state and $\Gamma_k$ ($k \in \N$) is the set of all such paths that have length at most $k$; 
thus $\Gamma_0 \subseteq \Gamma_1 \subseteq \ldots \subseteq \Gamma_\infty$.

\begin{lem}\label{lem:Pathspace-convex}
{For $k \in \N \cup \{\infty\}$, $\Gamma_k$ is convex}.
\end{lem}

\begin{proof}
Let $\vgamma_0, \vgamma_1 \in \Gamma_\infty$ be two paths and consider $\lambda \in (0,1)$. 
We need to show that $\vgamma_\lambda = (1 - \lambda)\vgamma_0 + \lambda \vgamma_1$ is in $\Gamma_\infty$.
Recall that $\vx_n(\vgamma)$ is the state reached after the first $n$ segments of path $\vgamma$.
Note that for any $n \in \N$,
\[\vec{x}_n(\vgamma_\lambda) = (1- \lambda)\vec{x}_n(\vgamma_0) + \lambda \vec{x}_n(\vgamma_1).\]
Since $\Rp^\Lambda$ is convex and both $\vec{x}_n(\vgamma_1)$ and $\vec{x}_n(\vgamma_2)$ are in $\Rp^\Lambda$, we conclude that $\vec{x}_n(\vgamma_\lambda)$ is in $\Rp^\Lambda$, too. 
Moreover, because both $\lim_{n \to \infty} \vx_n(\vgamma_0)$ and $\lim_{n \to \infty} \vx_n(\vgamma_1)$ converge, we see that
\[\lim_{n \to \infty}\vx_n(\vgamma_\lambda) = (1-\lambda)\lim_{n \to \infty}\vx_n(\vgamma_0) + \lambda \lim_{n \to \infty}\vx_n(\vgamma_1)\]
also converges.

Below, for a path $\vgamma = (\vx_0, \vu_1, \vu_2, \ldots)$, we use the notation $\vu_n(\vgamma)$ to represent the $n$'th flux vector $\vu_n$ in $\vgamma.$
Since 
\[\vec{u}_n(\vgamma_\lambda) = (1 - \lambda)\vec{u}_n(\vgamma_0) + \lambda \vec{u}_n(\vgamma_1),\]
any reaction $\alpha$ occurs with positive flux in $\vec{u}_n(\vgamma_\lambda)$ only if 
$\alpha$ occurs with positive flux in $\vec{u}_n(\vgamma_i)$
for $i = 0$ or 1. 
Without loss of generality, suppose that 
{$\alpha$ occurs with positive flux in $\vec{u}_n(\vgamma_0)$}. 
Then reaction 
$\alpha$ is applicable at $\vec{x}_{n-1}(\vgamma_0)$, so the reactants are all present in positive concentrations in $\vec{x}_{n - 1}(\vgamma_0)$. This implies that they are present with positive concentrations in $\vec{x}_{n - 1}(\vgamma_\lambda)$ (note that we have excluded the case $\lambda = 1$ from the outset). Therefore reaction 
$\alpha$ is applicable at $\vec{x}_{n-1}(\vgamma_\lambda)$.
We conclude that every reaction occurring with positive flux in $\vec{u}_n(\vgamma_\lambda)$ is applicable at $\vec{x}_{n-1}(\vgamma_\lambda)$. This shows that $\Gamma_\infty$ is convex. 

To see that $\Gamma_k$ for $k \in \N$ is also convex, note that if $\vu_n(\vgamma_0) = 0$ and $\vu_n(\vgamma_1) = 0$ then $\vu_n(\vgamma_\lambda)$ will also be zero.
\end{proof}

The next lemma shows that if it is possible to reach from a state $\vc$ to several other states, each containing some species possibly distinct from each other, then it is possible to reach from $\vc$ to a state with \emph{all} of those species present at once.

\begin{lem}  \label{lem:produce-all}
    Let $l \in \N$, $k \in \N \cup \{\infty\}$, and let $\vc,\vd_1,\ldots,\vd_l$ be states such that 
    $\vc \segto^k \vd_1$, 
    $\vc \segto^k \vd_2$, 
    $\ldots$, 
    and
    $\vc \segto^k \vd_l$. 
    Then there exists $\vd$ such that $\vc \segto^k \vd$ and $[\vd] = \bigcup_{i=1}^l [\vd_i]$.
\end{lem}

\begin{proof}
    Write $\vgamma_i$ for the path from $\vc$ to $\vd_i$; the convexity of $\Gamma_k$ shows that the convex combination 
    \[\vgamma = \frac{1}{l}\sum_{i=1}^l \vgamma_i\] 
    is a valid path in $\Gamma_k$. 
    Letting 
    \[\vd = \frac{1}{l}\sum_{i=1}^l \vd_i,\]
    $\vgamma$ exhibits a $k$-segment path from $\vc$ to $\vd$.  If $S$ is a species that is present at $\vd_i$ for any $i$ then $S$ is also present at $\vd$. On the other hand, if $S$ is present in none of the $\vd_i$ then $S$ is not present in $\vd$. As a result, $[\vd] = \bigcup_{i=1}^l [\vd_i]$.
\end{proof}

\begin{definition}
Given a state $\vc$, let $\producible(\vc)$ be the set of all species that are \emph{producible} from $\vc$---i.e., present in some state that is segment-reachable from $\vc$. 
\end{definition}

The next lemma shows that with at most a constant number of straight line segments,
it is possible to reach from any state $\vc$ to a state containing \emph{all} species possible to produce from $\vc$.

\begin{lemma}\label{lem-all-species-present}
Let $m$ be the minimum of $|R|$ and $|\Lambda|$ and let $\vc$ be any state. Then there is a state $\vd$ such that $\vc \segto^{m} \vd$ and $[\vd] = \producible(\vc)$. 
\end{lemma}

\begin{proof}
    Given a state $\vc$, let $\producible_{< \infty}(\vc)$
    be the set of all species that are present in some state that is $k$-segment-reachable from $\vc$ for some $k < \infty$. 
    
    We first argue that $\producible(\vc) = \producible_{< \infty}(\vc)$.
    Since clearly $\producible_{< \infty}(\vc) \subseteq \producible(\vc)$, it remains to show that $\producible(\vc) \subseteq \producible_{< \infty}(\vc)$. 
    Let $S$ be a species that is present in the state $\vd$ such that $\vc \segto^\infty \vd$. 
    Then there exists a sequence $(\vb_i)_i$ such that $\vc \to^1 \vb_1 \to^1 \vb_2 \to^1 \ldots$ with $\vd = \lim_{i \to \infty} \vb_i$. Because $S \in [\vd]$ there must be some $i_0 \in \N$ where $S \in [\vb_{i_0}]$, and since $\vc \segto^{i_0} \vb_{i_0}$ we see that $S \in \producible_{< \infty}(\vc)$. 
    Thus $\producible(\vc) = \producible_{< \infty}(\vc)$.
    
    We show that the lemma holds for $\producible_{< \infty}(\vc)$;
    since $\producible(\vc) = \producible_{< \infty}(\vc)$ this establishes the full lemma.

    For all $i \in \N$, let $\Lambda_i$ be the set of species $S$ such that there exists a $\vd$ with $\vc \segto^i \vd$ and $S \in [\vd]$. 
    Similarly, let $R_i$ be the set of reactions $\alpha$ such that there exists a $\vd$ with $\vc \segto^i \vd$ and $\alpha$ is applicable in $\vd$.
    Note that $\Lambda_0 = [\vc]$ and $R_0$ is the set of reactions applicable in $\vc$. Also, since $\vc \segto^i \vd$ implies $\vc \segto^{i + 1} \vd$ we see that $\Lambda_i \subseteq \Lambda_{i + 1}$ and $R_i \subseteq R_{i + 1}$ for all $i$.
    
    Now we show that for all $i$ there exists some $\vx_i$ such that $[\vx_i] = \Lambda_i$ and $\vc \segto^i \vx_i$ (and therefore $R_i$ consists of the reactions applicable at $\vx_i$). To see this, for each $S \in \Lambda_i$ let $\vd_S$ be a state such that $\vc \segto^i \vd_S$ and $S \in [\vd_S]$. By applying Lemma~\ref{lem:produce-all} to the set of all $\vd_S$, there is some $\vd$ such that $\vc \segto^i \vd$ and $[\vd] = \Lambda_i$; this $\vd$ is our desired $\vx_i$. 
    
    Now we will show that if $\Lambda_i = \Lambda_{i + 1}$ then $R_i = R_{i + 1}$ and, independently, if $R_i = R_{i + 1}$ then $\Lambda_{i + 1} = \Lambda_{i + 2}$ for all $i$. First suppose that $\Lambda_i = \Lambda_{i + 1}$ and let $\alpha$ be a reaction in $R_{i + 1}$. Then there is some state $\vd$ such that $\vc \segto^{i+1} \vd$ and $\alpha$ is applicable at $\vd$. Since all of the reactants of $\alpha$ are present at $\vd$, they are a subset of $\Lambda_{i + 1} = \Lambda_i$. They are therefore present at $\vx_i$, so $\alpha$ is applicable at $\vx_i$. We conclude that $\alpha \in R_i$ so $R_{i + 1} = R_i$.
    
    Now suppose that $R_i = R_{i + 1}$ and let $S$ be a species in $\Lambda_{i + 2}$. Then there is some $\vd$ such that $\vc \segto^{i + 2} \vd$ and $S \in [\vd]$. 
    {If $S \in [\vc]$, then $S \in \Lambda_{i + 1}$.
    Otherwise,} $S$ must be produced by some reaction $\alpha$ in $R_{i + 1} = R_i$, and we can apply $\alpha$ to $\vx_i$ to obtain a state $\vd'$ such that $\vc \segto^i \vx_i \to^1 \vd'$ and $S \in [\vd']$. 
    Again, we conclude that $S \in \Lambda_{i + 1}$ so $\Lambda_{i + 2} = \Lambda_{i + 1}$. 
    
    Combining the two statements we just proved, we see that if $\Lambda_i = \Lambda_{i + 1}$, then $\Lambda_i = \Lambda_j$ for all $j \ge i$, so $\Lambda_i = \producible_{< \infty}(\vc)$. Similarly, if $R_i = R_{i + 1}$, then $\Lambda_{i + 1} = \producible_{< \infty}(\vc)$. 
    
    If $|\Lambda| \le |R|$, then since $\Lambda_0 \subseteq \Lambda_1 \subseteq \ldots$ is an increasing sequence of subsets of the finite set $\Lambda$, it must be the case that $\Lambda_j = \Lambda_{j + 1}$ for some $j \le |\Lambda|$, and in this case $\vx_j$ gives our desired $\vd$. If, on the other hand, $|R| \le |\Lambda|$, the proof is similar: first note that if $R_0 = \emptyset$ we're done. Otherwise $|R_0| \ge 1$ so since $R_i$ is an increasing sequence of subsets of $R$ there is some $j \le |R| - 1$ such that $R_j = R_{j + 1}$. Then $\Lambda_{j + 1} = \producible_{< \infty}(\vc)$ so $\vx_{j + 1}$ gives our desired $\vd$.
\end{proof}

Recall that a set is \emph{closed} if it contains all of its limit points.

\begin{lemma} \label{lem-one-seg-closed}
Let $\vc \in \Rp^\Lambda$ be any state and let $S_\vc \subseteq \Rp^\Lambda$ be the set of states that are straight-line reachable from $\vc$. Then $S_\vc$ is closed.
\end{lemma}

\begin{proof}
Let $R_\vc$ be the set of reactions that are applicable at $\vc$. 
Then $C = \setl{ \vu \in \R_{\ge 0}^R }{ \vu(\alpha) = 0 \text{ for } \alpha \notin R_\vc }$ is a polyhedron.
{Then} $\vc + \vM C$ is also a polyhedron 
(see~\cite{ziegler1995polytopes}), and is in particular closed. $S_c$ is just $\R_{\ge 0}^\Lambda \cap (\vc + \vM C)$, and is therefore also closed. 
\end{proof}

Note that Lemma~\ref{lem-one-seg-closed} is false if we replace ``straight-line reachable'' with ``segment-reachable''.
For example, consider $X \to C$ and $X + C \to Y + C$, where we take the initial state 
$\vc = \{1X, 0C, 0Y\}$.
Note that for any $\varepsilon > 0$ we can reach the state $\vd_\varepsilon = \{0 X, \varepsilon C, (1 - \varepsilon) Y \}$.
However, 
because producing $Y$ first requires consuming a positive amount of $X$ to create the catalyst $C$,
the state 
$\vd_0 = \{ 0X, 0C, 1 Y\}$
is not segment reachable from $\vc$, even though $\vd_0 = \lim_{\varepsilon \to 0} \vd_{\varepsilon}$.

If we have an infinite sequence of states such that $\vc = \vb_0 \segto \vb_1 \segto \vb_2 \dots$ and $\lim_{i \to \infty} \vb_i = \vd$, this does not immediately imply that $\vc \segto^\infty \vd$ by \cref{defn-reachable-lines}.
This is because although the endpoints of the paths $\vb_i \segto \vb_{i+1}$ converge to $\vd$, the intermediate states on the paths related by $\to^1$
(i.e., $\vb_i \to^1 \vb_i' \to^1 \vb_i'' \to^1 \dots \to^1 \vb_{i+1}$)
may not converge.
To capture this weaker notion of convergence we introduce the following definition,
which generalizes $\segto^\infty$ by requiring only that there be a converging \emph{subsequence} of states. 
(The weaker notion will be eventually needed to prove \cref{lem:not-stably-compute-implies-can-reach-far-from-correct}.)

\newcommand{\segtoio}{\segto^\infty_{\mathrm{ss}}}

\begin{defn}
\label{defn:io-segment-reachable}
State $\vd$ is \emph{s.s.\ segment reachable} from state $\vc$, written $\vc \segtoio \vd$, if 
$(\exists \vb'_0,\vb'_1,\dots)\ \vc = \vb'_0 \to^1 \vb'_1 \to^1 \dots$,
where for some subsequence 
$\vb_0, \vb_1, \dots$ of 
$\vb'_0,\vb'_1,\dots$,
$\lim\limits_{i\to\infty} \vb_i = \vd$.
\end{defn}

The main results of this section have this weaker notion of convergence as a precondition,
which will imply that, despite appearances, $\segto^\infty$ and $\segtoio$ are actually equivalent.

The next lemma shows that if no more species are producible from state $\vc$ than are already present in $\vc$,
then any state $\vd$ that is $\segtoio$ reachable from $\vc$ is reachable via a single straight line segment.

\begin{lem}
\label{lem-all-producible-present}
If $[\vc] = \producible(\vc)$, then every state $\vd$ such that $\vc \segtoio \vd$ is straight-line reachable from $\vc$ (i.e., $\vc \segtoio \vd$ implies $\vc \to^1 \vd$).
\end{lem}

\begin{proof}
First consider the finite case where $\vc \segto^k \vd$ for $k < \infty$. Let $\vu_1, \ldots, \vu_k$ be the flux vectors corresponding to the segments in the path from $\vc$ to $\vd$. Then $\vu = \vu_1 + \ldots + \vu_k$ is a vector in $\Rp^R$.
Every reaction that occurs with positive flux in $\vu$ has positive flux in one of the $\vu_i$,
and thus its reactants are in $\producible(\vc)$, 
so are present in $\vc$ by the assumption $[\vc] = \producible(\vc)$.
Thus 
every reaction that occurs with positive flux in $\vu$ is applicable at $\vc$. 
The straight-line from $\vc$ corresponding to $\vu$ takes $\vc$ to 
\[\vc + \vM\vu = \vc + \sum_{i = 1}^k \vM\vu_i = \vd,\]
so $\vc \to_\vu^1 \vd$.

Now suppose that $\vc \segtoio \vd$. Then there is a sequence $\vb'_1, \vb'_2, \dots$ of states such that $\vc \to^1 \vb'_1 \to^1 \vb'_2 \ldots$ and,
for some subsequence $\vb_1, \vb_2, \dots$ of $\vb'_1, \vb'_2, \dots$,
$\vd = \lim_{i \to \infty} \vb_i$. For each finite $i \geq 1$, for some finite $j \geq i$, $\vc \segto^j \vb_i$.
So by the finite case shown above, 
$\vc \to^1 \vb_i$. 
Thus $\{\vb_1, \vb_2, \ldots\} \subseteq S_\vc$,
where $S_\vc$ is the set of states straight-line reachable from $\vc$.
Since $S_\vc$ is closed by Lemma~\ref{lem-one-seg-closed},
it contains all its limit points,
so $\vd \in S_\vc$ as well,
i.e., $\vc \to^1 \vd$. 
\end{proof}

Finally, the previous lemmas can be combined to show that at most a constant number of straight line segments (depending on the CRN) are required to reach from any state to any other reachable state.
In fact, this holds even for states that are only $\segtoio$ reachable.

\begin{thm}
\label{thm:io-reachable-segment-flux-bound}
  If $\vc \segtoio \vd$, then $\vc \segto^{m+1} \vd$, where $m=\min\{|\Lambda|,|R|\}$. Additionally, there is a constant $K$, 
  depending only on the CRN, 
  so that the path from $\vc$ to $\vd$ can be chosen so that the total flux of all reactions along the path is less than $K\| \vd - \vc \|$.
\end{thm}

\begin{proof}
    First, without loss of generality we can consider the reduced CRN where we remove all of the reactions that are not used with positive flux in the given path from $\vc$ to $\vd$. By Lemma~\ref{lem-all-species-present}, we can find a state $\vc'$ such that $\vc \segto^m \vc'$ and $[\vc'] = \producible(\vc)$. 
    We now show that we can ``scale-down'' the path $\vc \segto^m \vc'$ such that no reaction occurs more than in the original path $\vc \segto \vd$, allowing us to complete the path to $\vd$ using Lemma~\ref{lem-all-producible-present}.
    
    First consider the finite case, where $\vc \segto^k \vd$ for $k < \infty$.
    {We make the following general observation about finite paths: Let}
    $\vrho \in \Gamma_l$ be any finite path 
    with segments given by the flux vectors $\vu_1 \ldots \vu_l$, and let $F_{\vrho, \alpha}$ be the total flux through the reaction $\alpha$ along $\vrho$, i.e.
    \[F_{\vrho, \alpha} = \sum_{i = 1}^{l} \vu_i(\alpha).\]
    Then if $\vrho = \lambda \vrho_1 + (1 - \lambda) \vrho_0$, then
    \[F_{\vrho, \alpha} = \lambda F_{\vrho_1, \alpha} + (1 - \lambda) F_{\vrho_0, \alpha}.\]

    Let $\vsigma \in \Gamma_m$ be the path from $\vc$ to $\vc'$, let $\vsigma_0 \in \Gamma_0 \subseteq \Gamma_m$ be the 
    {trivial path staying at $\vc$}, and let $\vgamma \in \Gamma_k$ be the given path from $\vc$ to $\vd$. We can find some small $\varepsilon > 0$ such that
    \[\varepsilon F_{\vsigma, \alpha} < F_{\vgamma, \alpha}\]
    for all reactions $\alpha \in R$. As a result, letting $\vgamma' = \varepsilon \vsigma + (1 - \varepsilon) \vsigma_0$, 
    we see that
    \[F_{\vgamma', \alpha} = \varepsilon F_{\vsigma, \alpha} + ( 1- \varepsilon) F_{\vsigma_0, \alpha}  = \varepsilon F_{\vsigma, \alpha} < F_{\vgamma, \alpha}.\]
    
    Let $\va$ be the state reached via $\vgamma'$ (in particular $\vc \segto^m \va$).
    Since $\va = \varepsilon \vc' + (1-\varepsilon) \vc$ and $[\vc'] = \producible(\vc)$ for the reduced CRN, 
    we also have that $[\va] = \producible(\vc)$.
    Thus all reactions $\alpha$ of the reduced CRN are applicable at $\va$,
    and by Lemma~\ref{lem-all-producible-present},
    the final straight line from $\va$ can be defined by the flux vector $\vu_\alpha = F_{\vgamma, \alpha} - F_{\vgamma', \alpha}$, so that $\va + \vM \vu = \vd$.
    This shows that 
    $\va \to^1 \vd$, so
    $\vc \segto^{m+1} \vd$,
    proving the theorem for the case of finitely many segments.
    
    Now suppose that $\vc \segtoio \vd$, and let $\vgamma \in \Gamma_\infty$ be the path $\vc \to^1_{\vu_1} \vc_1 \to^1_{\vu_2} \vc_2 \ldots$ that starts at $\vc$ and has an infinite subsequence of states $\vb_1,\vb_2,\dots$ converging to to $\vd$. Because we have assumed without loss of generality that every reaction $\alpha$ occurs with positive flux along $\vgamma$, we know that 
    \[\sum_{i = 1}^\infty \vu_i(\alpha)\]
    is positive (although it might be infinite). As a result, there is some finite $N_\alpha$ such that 
    \[\sum_{i = 1}^{N_\alpha} \vu_i(\alpha) > 0.\]
    Let $N = \max_{\alpha \in R} N_\alpha$ be the number of line segments required for each reaction to have had positive flux. 
    The truncation of $\vgamma$ to a path with $N$ segments from $\vc$ to $\vc_N$ is then a path $\vgamma_N$ where every reaction $\alpha$ occurs with positive flux. By applying the first part of the argument, we can find a state $\va$ with $[\va] = \producible(\vc)$ so that $\vc \segto^m \va$ and $\va \to^1 \vc_N$. But then because $\vc_N \segtoio \vd$ we see that $\va \segtoio \vd$, so by \Cref{lem-all-producible-present} we see that $\va \to^1 \vd$. As a result, we conclude that $\vc \segto^{m + 1} \vd$.
    
    To see that we can bound the total flux along the path from $\vc$ to $\vd$, first note that by taking $\varepsilon$ small enough, we can guarantee both that the total flux along $\vgamma'$ is bounded by $\| \vd - \vc \|$ and that 
    \[\| \va - \vc\| = \| \vM \vF_{\vgamma'} \| < \| \vd - \vc\|. \]
    By the triangle inequality, this implies that $\| \va - \vd \| < 2\| \vd - \vc \|$. Now by applying \Cref{lem:bound-reaction-fluxes-straight-line-reachability} we see that there's some constant $C$ depending only on the CRN so that the flux vector $\vu$ of the straight line from $\va$ to $\vd$ can chosen with $\| \vu \| < 2C \| \vd - \vc \|$. Taking $K = 2C + 1$ we see that the flux along the whole path $\vc \segto^m \va \to^1 \vd$ is bounded above by $K\| \vd - \vc \|$.
\end{proof}

Note that \cref{thm:io-reachable-segment-flux-bound} immediately implies that $\segto^\infty$ and $\segtoio$ are the same relation, since $\vc \segtoio \vd$ implies $\vc \segto^{m+1} \vd$ implies $\vc \segto^{\infty} \vd$.

Although the full power of \Cref{thm:io-reachable-segment-flux-bound} is useful later in \Cref{thm:partial-state-reachable}, the most important consequence of \Cref{thm:io-reachable-segment-flux-bound} is the following result, which we will use repeatedly.

\begin{corollary}
\label{thm-reachable-segment-bound}
  If $\vc \segto \vd$, then $\vc \segto^{m+1} \vd$, where $m=\min\{|\Lambda|,|R|\}$.
\end{corollary}

\begin{proof}
This follows from \Cref{thm:io-reachable-segment-flux-bound} and the fact that $\vc \segto \vd$ implies $\vc \segtoio \vd$.
\end{proof}

\Cref{thm-reachable-segment-bound} will allow us to assume without loss of generality that there are a constant number of line segments between any two states, simplifying many arguments.
For example, it is not obvious that the relation $\segto^\infty$ is transitive, since one cannot concatenate two infinite sequences.
However, since two finite sequences of segments can be concatenated,
the following corollary is immediate.

\begin{corollary}
\label{cor:segto-is-transitive}
The relation $\segto$ is transitive.
\end{corollary}

    The goal of the reachability relation is to capture ``what could happen'' in chemical reaction networks independently of rates. Thus it is natural to satisfy several properties: 
    the relation should be reflexive (true for $\segto$ since $\vx \segto^0 \vx$) and transitive (\Cref{cor:segto-is-transitive}).
    Further, the relation should be \emph{additive} in the intuitive sense that 
    the presence of additional molecules cannot entirely prevent reactions from happening (although in a kinetic model it could effectively slow down reactions due to competition);
    formally, if $\vx \segto \vy$, then 
    $\vx + \vc \segto \vy + \vc$
    for any state $\vc$.
    Additivity is a crucial property of the more standard notion of discrete CRN reachability, used for example in many cases to prove impossibility results for those systems~\cite{AngluinAE2006semilinear, LeaderElectionDIST, alistarh2017time, belleville2017hardness, alistarh2018space}.
    We also employ additivity of $\segto$ for impossibility results, 
    for example in the proofs of 
    Lemmas~\ref{lem-shifted-in-unshifted} and~\ref{lem-shifted-sets-cover}.%
    \footnote{
    Another property of $\segto$ that we use extensively is \emph{scale-invariance}:
    if $\vx \segto \vy$,
    then 
    $\lambda\vx \segto \lambda\vy$
    for any $\lambda \geq 0$,
    which is essentially responsible for the convexity of \Cref{lem:Pathspace-convex}.
    This does \emph{not} hold for discrete CRN reachability when $\lambda < 1$,
    even when the scaled discrete states are well-defined, 
    e.g., 
    the reaction $X+X \to Y$ is applicable in state $\vx = \{2X\}$ but not in state $0.5 \vx = \{1 X\}$ in the discrete model.
    }

    While satisfying these properties is a good start for justifying the reasonableness of $\segto$,
    it is natural to wonder whether there are some ``reasonable'' rate laws that segment-reachability fails to capture, i.e., perhaps some CRN rate law would take $\vx$ to $\vy$ even though $\vx \not\segto \vy$.
    In \Cref{sec:mass-action-reachability}, we define an apparently much more general notion of reachability (Definition~\ref{defn:valid-rate-schedule}) that captures all commonly-studied rate laws, 
    while still respecting the fundamental semantics of reactions.
    We prove that our reachability relation is in fact identical to it,
    i.e., $\vx$ can reach to $\vy$ under this notion if and only if $\vx \segto \vy$ (Theorem~\ref{thm:valid-reachable-implies-segment-reachable} and Lemma~\ref{lem:segment-reachable-implies-valid-rate-schedule}).

\newcommand{\bdd}{\mathrm{b}}
\newcommand{\unbdd}{\mathrm{u}}

\subsection{Generality of Segment Reachability}
\label{sec:mass-action-reachability}

In this section we justify that our notion of reachability via straight lines actually corresponds to the most general notion of ``being able to get from one state to another,'' restricted only by the non-negativity of concentrations, reaction stoichiometry and the need for catalysts---as long as we maintain the causal relationships between the production of species.
This notion of reachability admits ``time-varying'' rate laws where reactions occur according to some arbitrary schedule, which captures situations such as solutions that are not well-mixed, or where physical parameters, such as temperature, change in some arbitrary way.
The main result of this section,
~\Cref{thm:valid-reachable-implies-segment-reachable},
formalizes this idea and
justifies calling segment-reachability simply ``reachability'' in the rest of this paper.

We begin with a review of mass-action kinetics, 
the most commonly used rate law in chemistry, 
and show the (physically and intuitively obvious but mathematically subtle) features that make it consistent with segment reachability. 
We then generalize rate-law trajectories to arbitrary ``valid rate schedules'', and prove that these are exactly captured by segment-reachability.

A CRN with positive rate constants assigned to each reaction defines a mass-action ODE (ordinary differential equation) system with a variable for each species, 
which  represents the time-varying concentration of that species. 
We follow the convention of upper-case species names and lower-case concentration variables.
Each reaction contributes one term to the ODEs for each species 
produced or consumed in it.
The term from reaction $\alpha$ appearing in the ODE for $x$ is the product of:
the rate constant,
the reactant concentrations,
and the net stoichiometry of $X$ in $\alpha$ (i.e., the net amount of $X$ produced by $\alpha$, negative if consumed).
For example, the CRN
\begin{align*}
X + X &\rxn^{k_1} C \\
C + X &\rxn^{k_2} C + Y
\end{align*}
corresponds to ODEs:
\begin{align}
\label{eqns-mass-action-1}
dx/dt &= -2 k_1 x^2 - k_2 c x \\
dc/dt &= k_1 x^2 \\
dy/dt &=  k_2 c x
\label{eqns-mass-action-3}
\end{align}
where $k_1, k_2$ are the rate constants of the two reactions.

Given a CRN $C = (\Lambda, R)$,
let $\vec{A}(\vd): \Rp^\Lambda \to \Rp^R$ be the rates of all the reactions in state $\vd$ as given by the mass-action ODEs.
Given an assignment of (strictly) positive rate constants,
and an initial state $\vc$,
the \emph{mass-action trajectory} is a function $\vrho: [0,t_\mathrm{max}) \to \Rp^\Lambda$,
where $t_\mathrm{max} \in \Rp \cup \{\infty\}$,
such that $\vrho$ is the solution to $d \vrho/dt = \vM\cdot\vec{A}(\vrho(t))$ with $\vrho(0) = \vc$, where $t_\mathrm{max}$ is the maximum time, typically $\infty$, for which the solution is defined on all of $[0,t_\mathrm{max})$.
Although beyond the scope of this paper, 
mass-action ODEs are locally Lipschitz, so a CRN 
admits exactly one mass-action trajectory for a fixed collection of rate constants and initial state $\vc$.
Note that for some CRNs (e.g.~$2X \to 3X$), the solution of the ODEs goes to infinite concentration in finite time,\footnote{Indeed, the mass-action ODE corresponding to the CRN~$2X \to 3X$ is $\frac{dx}{dt} = x^2$, which is solved by $x(t) = \frac{1}{C - t}$, where $C = 1/x(0)$.
This goes to infinity as $t$ approaches $C$.}
and for such CRNs, $t_\mathrm{max}$ is finite.

\begin{defn}
\label{defn:mass-action-reachable}
Fix an assignment of positive mass-action rate constants. 
Let $\vc,\vd$ be two states.
We say $\vd$ is \emph{mass-action reachable (with respect to the rate constants)} from $\vc$ if the associated mass action trajectory $\vrho$ obeys $\vrho(0) = \vc$ and either $\vrho(t) = \vd$ for some finite $t \geq 0$ or $\lim_{t\to\infty} \vrho(t) = \vd$.\footnote{
    Note that a more general definition would say $\vd$ is mass-action reachable from $\vc$ if there exist positive rate constants such that the trajectory starting at $\vc$ passes through or approaches $\vd$. Note, however, that this relation is not transitive: for some CRNs, $\vc$ reaches to $\vd$ under one set of rate constants and $\vd$ reaches to $\vx$ under another set of rate constants, yet no \emph{single} assignment of rate constants takes the CRN from $\vc$ to $\vx$.
}
\end{defn}

In order to prove Theorem~\ref{thm:valid-reachable-implies-segment-reachable} we need to introduce the notion of a \emph{siphon} from the Petri net literature.
This notion will be used, as well, to prove negative results in \Cref{subsec-negative-computable-implies-piecewise-linear}.

\begin{defn}
\label{defn:siphon}
Let $\calC=(\Lambda,R)$ be a CRN.
A \emph{siphon} is a set of species $\Omega \subseteq \Lambda$ such that, for all reactions $\langle \vr,\vp \rangle \in R$, $[\vp] \cap \Omega \neq \emptyset \implies [\vr] \cap \Omega \neq \emptyset$, i.e., every reaction that has a product in $\Omega$ also has a reactant in $\Omega$.
\end{defn}

The following lemma, due to Angeli, De Leenheer, and Sontag~\cite{angeli2007petri}, shows that this is equivalent to the notion that ``the absence of $\Omega$ is forward-invariant'' under mass-action: if all species in $\Omega$ are absent, then they can never again be produced (under mass-action).
\footnote{It is obvious in the discrete CRN model, and in an intuitive physical sense, that 
if producing a species initially absent causally requires another species also initially absent and vice versa, then neither species can ever be produced.
However, 
it requires care to prove this for mass-action ODEs.
Consider the CRN $2X \to 3X$. The corresponding mass-action ODE is $dx/dt = x^{2}$, and has the property that starting with $x(0) = 0$, it cannot become positive, i.e., the only solution with $x(0)=0$ is $x(t) = 0$ for all $t \geq 0$.
However, the very similar non-mass-action ODE
$dx/dt = x^{1/2}$
has a perfectly valid solution $x(t) = t^2/4$, which starts at $0$ but becomes positive,
despite the fact that at $t=0$, $dx/dt = 0$.
(Though $x(t) = 0$ for all $t \geq 0$ is another valid solution.)
The difference is that mass-action polynomial rates are locally Lipschitz (have bounded rates of change, unlike $x^{1/2}$, whose derivative goes to $\infty$ as $x \to 0$) and so are guaranteed to have a unique solution by the Picard-Lindel\"of theorem.
}
For the sake of completeness, we give a self-contained proof in \Cref{app:siphons-absence-forward-invariant}.

\newcommand{\siphonsLemText}{
  Fix any assignment of positive mass-action rate constants.
  Let $\Omega \subseteq \Lambda$ be a set of species.
  Then $\Omega$ is a siphon if and only if, for any state $\vc$ such that $\Omega \cap [\vc] = \emptyset$ and any state $\vd$ that is mass-action reachable from $\vc$, $\Omega \cap [\vd] = \emptyset$.
}

\newtheorem*{siphonsLem}{\Cref{lem-mass-action-forward-invariant-siphon}}

\begin{lem}[\cite{angeli2007petri}, Proposition 5.5]
\label{lem-mass-action-forward-invariant-siphon}
\siphonsLemText
\end{lem}

We show that the same holds true for segment-reachability.
Due to the discrete nature of segment-reachability, the proof is more straightforward than that of \Cref{lem-mass-action-forward-invariant-siphon}. It follows the same essential structure one would use to prove this in the discrete CRN model: if the siphon $\Omega$ is absent, no reaction with a reactant in $\Omega$ can be the next reaction to fire, so by the siphon property, no species in $\Omega$ is produced in the next step.

\begin{lem}
\label{lem-seg-reachable-forward-invariant-siphon}
  Let $\Omega \subseteq \Lambda$ be a set of species.
  Then $\Omega$ is a siphon if and only if, for any state $\vc$ such that $\Omega \cap [\vc] = \emptyset$ and any state $\vd$ such that $\vc \segto \vd$, $\Omega \cap [\vd] = \emptyset$.
\end{lem}

\begin{proof}
  To see the forward direction,
  suppose $\Omega$ is a siphon, let $\vc$ be a state such that $[\vc] \cap \Omega = \emptyset$, and let $\vd$ be such that $\vc \segto \vd$.
  {By Theorem~\ref{thm-reachable-segment-bound}, there is a finite path $\vgamma$ such that 
  the $n$'th line segment is between states $\vx_{n-1}(\vgamma)$ and $\vx_{n}(\vgamma)$,
  with $\vc = \vx_{0}(\vgamma)$ and $\vd = \vx_{m+1}(\vgamma)$. 
  Assume inductively that $[\vx_{n-1}(\vgamma)] \cap \Omega = \emptyset$;
  then no reaction applicable at $\vx_{n-1}(\vgamma)$ has reactants in $\Omega$.
  So by definition of siphon, 
  no reaction applicable at $\vx_{n-1}(\vgamma)$ has products in $\Omega$,
  and  
  $[\vx_{n}(\vgamma)] \cap \Omega = \emptyset$ as well.
  }
  Therefore $\vd \cap \Omega = \emptyset$.
  This shows the forward direction.
  
  To show the reverse direction, suppose that $\Omega$ is not a siphon.
  Then there is a reaction $\alpha = \langle \vr,\vp \rangle$ such that $[\vp] \cap \Omega \neq \emptyset$, but $[\vr] \cap \Omega = \emptyset$.
  Then from any state $\vc$ such that $[\vc] = \Lambda \setminus \Omega$ (i.e., all species not in $\Omega$ are present), all reactants of $\alpha$ are present, so $\alpha$ is applicable.
  Running $\alpha$ produces $S$, hence results in a state $\vd$ such that $\vc \segto \vd$ with $\Omega \cap [\vd] \neq \emptyset$, since $S \in \Omega$.
\end{proof}

Recall that $\producible(\vc)$ represents the set of species producible from state $\vc$.
The next lemma shows that the set of species that cannot ever be produced from a given state is a siphon.

\begin{lem}
\label{lem:not-producible-is-siphon}
If $\vc$ is any state then $\Omega = \Lambda \setminus \producible(\vc)$ is a siphon.
\end{lem}

\begin{proof}
By Lemma~\ref{lem-all-species-present} there is a state $\vc'$ that is segment reachable from $\vc$ with all of the species in $\producible(\vc)$ present. If $\Omega$ were not a siphon, there would be a reaction $\alpha$ that produced a species $S$ of $\Omega$ such that all of the reactants of $\alpha$ would be contained in $\Lambda \setminus \Omega = \producible(\vc)$. This implies that $\alpha$ would be applicable at $\vc'$, so $S$ would be in $\producible(\vc)$, giving a contradiction.
\end{proof}

The main result of this section is Theorem~\ref{thm:valid-reachable-implies-segment-reachable},
which justifies that our (seemingly limited) notion of reachability via straight lines is actually quite general.
To state the theorem, we define a very general notion of ``reasonable rate laws'',
which are essentially schedules of rates to assign to reactions over time.
All known rate laws such as mass-action, Michaelis-Menten, Hill function kinetics, 
as well as our own nondeterministic notion of adversarial rates following straight lines (segment reachability, Definition~\ref{defn-reachable-segment}),
obey this definition.
(We justify this below explicitly for mass-action and Definition~\ref{defn-reachable-segment}, but it is straightforward to verify in the other cases.)

Recall that $R$ is the set of all reactions in some CRN, and $\Lambda$ is the set of its species.

\begin{defn}
\label{defn:valid-rate-schedule}
A \emph{rate schedule} is a function $\vf: \R_{\geq 0} \to \R_{\geq 0}^R$.
We interpret $\vf_\alpha(t)$ to be the rate, or instantaneous flux, at which reaction $\alpha$ occurs at time $t$.
Given a state $\vc \in \Rp^\Lambda$,
we say $\vf$ is a \emph{valid rate schedule starting at $\vc$} if:
\begin{enumerate}
    \item
    \label{defn:valid-rate-schedule:integrable}
    \textbf{(Total reaction fluxes are well-defined)}
    For each $\alpha \in R$,
    $\vf_\alpha$ is (locally Lebesgue) integrable:
    for each time $t \geq 0$,
    $\vF_\alpha(t) = \int_0^t \vf_\alpha(t) dt$ is well-defined and finite
    (although $\int_0^\infty \vf_\alpha(t) dt$ may be infinite).
    Let $\vF(t)$ be the vector in $\R_{\ge 0}^R$ whose $\alpha$ coordinate is $\vF_\alpha(t)$, which represents the total amount of each reaction flux that has happened by time $t$.\footnote{
        An alternative to Definition~\ref{defn:valid-rate-schedule} would start with a differentiable trajectory $\vrho$ and total flux $\vF$ 
        (related via $\vrho(t) = \vM \cdot \vF(t) + \vc$) and define $\vf = d\vF/dt.$
        However, requiring differentiable $\vrho$ and $\vF$ rules out many natural cases, such as the rate schedules implicit in segment-reachability (\Cref{defn-reachable-lines}),
        whose trajectories are not differentiable at cusp points $\vb_i$ in between straight lines and whose rate schedules are not even continuous.
    }
    
    Define the \emph{trajectory} $\vrho: \Rp \to \R_{\geq 0}^\Lambda$ of $\vf$ starting at $\vc$ for all $t \geq 0$ by 
    $\vrho(t) = \vM \cdot \vF(t) + \vc$,
    which represents the state of the CRN at time $t$.

    

    \item
    \label{defn:valid-rate-schedule:reactants-present}
    \textbf{(Positive-rate reactions require their reactants present)}
    For all times $t \geq 0$ and reactions $\alpha \in R$,
    if $\vf_\alpha(t) > 0$,
    then $\alpha$ is applicable in $\vrho(t)$.
    %
    
    \item 
    \label{defn:valid-rate-schedule:siphons}
    \textbf{(Absence of siphons is forward-invariant)}
    For every siphon $\Omega \subseteq \Lambda$, if $\Omega \cap [\vrho(t)] = \emptyset$ for some time $t$, then $\Omega \cap [\vrho(t')] = \emptyset$ for all times $t' \geq t$.
    
\end{enumerate}
\end{defn}

\begin{definition}
We say that a state $\vd$ is \emph{reachable from a state $\vc$ by a valid rate schedule} if there is a valid rate schedule $\vf$ starting at $\vc$, with trajectory $\vrho$, such that either $\vd = \vrho(t_f)$ for some $t_f < \infty$, 
or $\vd = \lim_{t \to \infty} \vrho(t)$. 
In the first case we say that $\vd$ is reached in finite time.
\end{definition}

We note that because $\vf$ is Lebesgue integrable, by \cite[Theorem~6.11]{royden1988real}, $\vF_\alpha(t)$ is locally absolutely continuous.

Although Definition~\ref{defn:valid-rate-schedule} explicitly constrains the states $\vrho(t)$ to be non-negative, the non-negativity of $\vrho(t)$ actually follows from condition~\eqref{defn:valid-rate-schedule:reactants-present}.\footnote{
Consider a rate schedule that takes concentrations negative: for instance starting with $1 X$ and applying reaction $\alpha: X \to Y$ with $\vF_\alpha(t) = 2$ for some $t > 0$.
To see that this contradicts condition~\eqref{defn:valid-rate-schedule:reactants-present}
when we naturally generalize notation $[\vc]$ to possibly negative $\vc$
(for any $\vc \in \R^\Lambda$, $[\vc] = \{S \in \Lambda \ |\ \vc(S) > 0 \}$),
suppose that $\vrho_X(t') < 0$ for some species $X$ at some time $t'$. 
Let $t_0$ be the supremum of all the times less than $t'$ where $\vrho_X(t) \ge 0$. 
Recall
$\vF_\alpha(t)$ is a locally absolutely continuous (and therefore continuous) function. 
Thus $\vrho(t)$ is also a continuous function so $\vrho_X(t_0) \ge 0$. Moreover, for all $t_0 < t < t'$ we know that $\vrho_X(t) < 0$ by our choice of $t_0$. So by condition~\eqref{defn:valid-rate-schedule:reactants-present} $\vf_\alpha(t) = 0$ for all $\alpha$ where $X$ is a reactant, and therefore (recall $\vM(X,\alpha)$ means the net consumption of $X$ in reaction $\alpha$)
\[\vrho_X(t') = \vrho_X(t_0) + \sum_{\alpha \in R} \vM(X,\alpha) \int_{t_0}^{t'} 
\vf_\alpha(t) dt \ge \vrho_X(t_0) \ge 0\]
a contradiction since $\vrho(t')$ was assumed to be negative. 
See also~\cite[Proposition 2.8]{fages2015inferring}, where the term \emph{strict} is equivalent to condition~\eqref{defn:valid-rate-schedule:reactants-present}.
}

Conditions~\eqref{defn:valid-rate-schedule:reactants-present} and~\eqref{defn:valid-rate-schedule:siphons} may appear redundant,
but in fact each can be obeyed while the other is violated.

For example, consider the reaction $\alpha: X \to 2X$,
starting in the state $\{0 X\}$,
with invalid rate schedule $\vf_\alpha(t) = t$,
with trajectory $\vrho_X(t) = t^2 / 2$.
Since the rate $\vf_\alpha(0)$ is 0, this vacuously satisfies~\eqref{defn:valid-rate-schedule:reactants-present} at time 0,
and since $\vrho_X(t) > 0$ for $t > 0$ ($X$ is present at all positive times), \eqref{defn:valid-rate-schedule:reactants-present} is also satisfied for positive times.
However, \eqref{defn:valid-rate-schedule:siphons} is violated, since $\{X\}$ is a siphon absent at time 0 but present at future times.
This example also demonstrates why condition~\eqref{defn:valid-rate-schedule:siphons} is required to satisfy our intuitive understanding of reasonable rate laws respecting ``causality of production'' among species:
with only the reaction $X \to 2X$,
the only way to produce more $X$ is already to have some $X$.

To see the other case,
take reactions $\alpha: X \to C$ and $\beta: C+X \to C+Y$,
starting in state $\{1 X, 1 Y, 0 C\}$.
Consider the invalid rate schedule $\vf_\alpha(t) = 0$ for all $t$, 
$\vf_\beta(t) = 1$ for $0 \leq t \leq 1/2$, 
and $\vf_\beta(t)=0$ for $t > 1/2$,
i.e., run only $\beta$, until $X$ is half gone.
This violates~\eqref{defn:valid-rate-schedule:reactants-present},
since $\beta$ occurs without its reactant $C$ present.
However, the only set of species absent along this trajectory is $\{C\}$,
which is not a siphon since reaction $\alpha$ has $C$ as a product but not a reactant,
so~\eqref{defn:valid-rate-schedule:siphons} is vacuously satisfied.

The next lemma shows that the most commonly-used rate law, mass-action,
gives a valid rate schedule and trajectory according to Definition~\ref{defn:valid-rate-schedule}.
Recall that $\vec{A}(\vrho(t)): \Rp^\Lambda \to \Rp^R$ represents the rates of all the reactions in state $\vrho(t)$ as given by the mass-action ODEs.
For instance, for our mass-action example at the beginning of this section, the function $\vf(t) = \vA(\vrho(t))$ corresponds to the ODEs of equations~\ref{eqns-mass-action-1}--\ref{eqns-mass-action-3} when written as $d\vrho/dt = \vM\cdot\vf(t)$.

\begin{lemma}
\label{lem:mass-action-gives-valid-rate-schedule}
Suppose we fix an assignment of positive mass-action rate constants for a given CRN as well as an initial state $\vc$. Suppose that the associated mass action trajectory $\vrho$ is defined for all time. Then $\vf: \Rp \to \Rp^R$ such that
$\vf(t) = \vec{A}(\vrho(t))$
is a valid rate schedule whose trajectory is $\vrho$.
\end{lemma}

\begin{proof}
First, note that because $\vrho(t)$ is a real analytic function, $\vec{A}(\vrho(t))$ is necessarily also real analytic, and therefore locally integrable, so condition~\eqref{defn:valid-rate-schedule:integrable} of Definition~\ref{defn:valid-rate-schedule} is satisfied. Let $\Tilde{\vrho}(t)$ be the trajectory associated with the rate schedule $\vf$. Then because $\vrho(t)$ is a solution to the mass-action ODEs with initial state $\vc$, 
\begin{align*}
\Tilde{\vrho}(t) &= \vc + \vM \int_0^t \vf(t) dt \\
&= \vc + \int_0^t \vM \cdot \vec{A}(\vrho(t)) dt \\
&= \vc + \int_0^t \frac{d}{dt} \vrho(t) dt \\
&= \vc + (\vrho(t) - \vc) \\
&= \vrho(t).
\end{align*}
Because $\vec{A}(\vrho(t))$ can only be positive when $\vrho_X(t) > 0$ for all reactants $X$ of $\alpha$, we see that $\vf_\alpha(t) > 0$ implies that $[\Tilde{\vrho}(t)] = [\vrho(t)]$ contains all of the reactants of $\alpha$. Therefore condition~\eqref{defn:valid-rate-schedule:reactants-present} of Definition~\ref{defn:valid-rate-schedule} is satisfied. Finally condition~\eqref{defn:valid-rate-schedule:siphons} of Definition~\ref{defn:valid-rate-schedule} is satisfied by Lemma~\ref{lem-mass-action-forward-invariant-siphon}.
\end{proof}

We say that a rate schedule $\vf: \Rp \to \Rp^R$ is \emph{finite} if there is $t_0 \geq 0$ such that $\vf(t) = \vec{0}$ for all $t \geq t_0$, i.e., reactions eventually stop occurring.
The following observation is straightforward to verify, showing that the concatenation of two valid rate schedules, with the first finite, is also a valid rate schedule.

\begin{observation}
\label{obs:concat-valid-rate-schedules}
If $\vf, \vg$ are valid rate schedules,
with $\vf$ finite such that $\vf(t) = \vec{0}$ for all $t \geq t_0$,
then $\vh$ defined by $\vh(t) = \vf(t)$ for $0 \leq t \leq t_0$ and $\vh(t) = \vg(t-t_0)$ for $t > t_0$, is a valid rate schedule.
\end{observation}

The next lemma shows essentially that our definition of segment-reachability creates a valid rate schedule.

\begin{lemma}
\label{lem:segment-reachable-implies-valid-rate-schedule}
If $\vc \segto \vd$ then $\vd$ is reachable from $\vc$ by a valid rate schedule.
\end{lemma}

\begin{proof}
Since $\vc \segto \vd$, by Theorem~\ref{thm-reachable-segment-bound} we know that $\vc \segto^{m + 1} \vd$.
Using induction and Observation~\ref{obs:concat-valid-rate-schedules},
it suffices to verify that the rates defined by the straight-line reachability relation $\vc \to^1_{\vu} \vd$ describe a valid rate schedule,
since the rate schedules given by $\segto^{m+1}$ are simply concatenations of these.
Define $\vf: \Rp \to \Rp^R$ by $\vf(t) = \vu(\alpha)$ for $0 \leq t < 1$ and $\vf(t) = \vec{0}$ for all $t \geq 1$.
(In other words, for one unit of time, run the reactions at constant rates described by $\vu$.)
Then $\vf$ is piecewise constant, and therefore integrable, so condition~\eqref{defn:valid-rate-schedule:integrable} of Definition~\ref{defn:valid-rate-schedule} is satisfied. 
Next note that for all $t \ge 1$, since $\vf_\alpha(t) = 0$ for all $\alpha$, 
condition~\eqref{defn:valid-rate-schedule:reactants-present} of Definition~\ref{defn:valid-rate-schedule} holds vacuously, and condition~\eqref{defn:valid-rate-schedule:siphons} holds because $\vrho(t) = \vrho(t')$ for any $t' > t$. 
Now let $t < 1$. 
Observe that
\begin{align*}
\vrho(t) &= \vc + \vM \int_0^{t} \vf(T) dT
\\&= \vc + \vM \int_0^{t} \vu(T) dT 
\\&= \vc + t \vM\vu
\\&= (1 - t)\vc + t(\vc + \vM\vu) 
\\&= (1 - t)\vc + t\vd.
\end{align*}
Since $\vd$ is a state (and thus non-negative on all species) and $t<1$, every species present with positive concentration in $\vc$ is present with positive concentration in $\vrho(t)$.
Thus all reactions applicable at $\vc$ are also applicable at $\vrho(t)$,
so condition~\eqref{defn:valid-rate-schedule:reactants-present} of Definition~\ref{defn:valid-rate-schedule} is also satisfied. 
Finally by Lemma~\ref{lem-seg-reachable-forward-invariant-siphon} we see that condition~\eqref{defn:valid-rate-schedule:siphons} of Definition~\ref{defn:valid-rate-schedule} is satisfied and therefore $\vf$ is a valid rate schedule starting at $\vc$. Since $\vrho(1) = \vd$, we see that $\vd$ is reachable from $\vc$ by a valid rate schedule.
\end{proof}


Finally, we have the main result of this section, which shows that segment reachability is as general as any valid rate schedule.

\begin{thm}  \label{thm:valid-reachable-implies-segment-reachable}
Given two states $\vc$ and $\vd$,  $\vd$ is reachable from $\vc$ by a valid rate schedule if and only if $\vc \segto \vd$.
\end{thm}

\begin{proof}
\Cref{lem:segment-reachable-implies-valid-rate-schedule} establishes the reverse direction.
To see the forward direction,
let $\vf$ be the valid rate schedule from $\vc$ to $\vd$, and define $\vF$ and $\vrho$ for $\vf$ as in Definition~\ref{defn:valid-rate-schedule}.

First, suppose $\vd$ is reached in finite time $t_f \in \Rp$,
i.e., $\vd = \vrho(t_f)$.

We say that a reaction $\alpha$ \emph{occurs with positive flux} if $\vF_\alpha(t_f) > 0$. 
Let $R_\vrho = \setr{\alpha\in R}{\vF_\alpha(t_f) > 0}$ be the reactions that occur with positive flux along the trajectory $\vrho$,
and let $\Lambda_\vrho = \{ S \in \Lambda \mid \vrho_S(t) > 0 \text{ for some } 0 \leq t \leq t_f \}$ be the species that are present with positive concentration at some point along the trajectory $\vrho$.

Consider removing species not in $\Lambda_\vrho$ and reactions not in $R_\vrho$.
We claim that the pair $(\Lambda_\vrho, R_\vrho)$ is a well-defined CRN, as defined in \Cref{subsec-prelim-defs}, because every reactant and product in $R_\rho$ is in $\Lambda_\rho$.
To see why, let $S \in \Lambda \setminus \Lambda_\vrho$ 
(i.e., $\vrho_S(t) = 0$ for all $0 \leq t \leq t_f$),
let $R_S$ 
be the reactions with $S$ as a reactant,
and let $P_S$ 
be the reactions with $S$ as a product;
we must show 
$R_S \cap R_\vrho = \emptyset$ and
$P_S \cap R_\vrho = \emptyset$.
By \Cref{defn:valid-rate-schedule} part \eqref{defn:valid-rate-schedule:reactants-present},
no reaction in $R_S$ has positive flux, 
so $R_S \cap R_\vrho = \emptyset$.
Since no reaction in $R_S$ has positive flux,
no reaction in $P_S$ can have positive flux or else $S$ would be produced with no reaction to consume it, contradicting $\vrho_S(t_f) = 0$,
so $P_S \cap R_\vrho = \emptyset$.

Now we claim that every species in this reduced CRN $(\Lambda_\vrho, R_\vrho)$ is segment-producible from $\vc$, i.e., $\producible(\vc) = \Lambda_\vrho$.
If not, then $\Omega = \Lambda_\vrho \setminus \producible(\vc)$ is non-empty. 
Letting $S^*$ be some element of $\Omega$, we know that $S^*$ has positive concentration along $\vrho$ by our construction of the reduced CRN. 
However, by Lemma~\ref{lem:not-producible-is-siphon}, $\Omega$ is a siphon. 
Since $\vc$ is zero on $\Omega$ and $S^* \in \Omega$, this violates 
\Cref{defn:valid-rate-schedule} part~\eqref{defn:valid-rate-schedule:siphons}.

Since every species in our reduced CRN is segment-producible from $\vc$, by Lemma~\ref{lem-all-species-present} we can construct a state $\vc'$ segment-reachable from $\vc$ where all of the species in the reduced CRN are present simultaneously.
Since every reaction in $R_\vrho$ is applicable at $\vc'$, the remainder of the proof is similar to the proof for the finite case of Theorem~\ref{thm-reachable-segment-bound}:
by ``scaling down'' the path from $\vc$ to $\vc'$,
there is a state $\va$ such that 
$\vc \segto^m \va$
(where $m = \min\{ |\Lambda_\vrho|, |R_\vrho| \}$),
$[\va] = \producible(\vc) = \Lambda_\vrho$,
and $\va \to^1 \vd$. 
Thus
$\vc \segto^{m+1} \vd$.
This handles the case that $\vd$ is reached in finite time.

On the other hand, suppose that $\vd$ is not reached from $\vc$ in finite time, but instead $\vd = \lim_{t \to \infty} \vrho(t)$.
This case is similar to the proof of the infinite case of Theorem~\ref{thm-reachable-segment-bound}.
By definition of $R_\vrho$, for each reaction $\alpha \in R_\vrho$, $\vF_\alpha(t) > 0$ for some $t > 0$. 
As a result, there is some finite $t_\alpha$ such that $\vF_\alpha(t_\alpha) > 0.$
Let $t_\mathrm{pos} = \max \{t_\alpha \mid \alpha \in \R_\vrho \}$, 
noting that $\vF(t_\mathrm{pos}) > 0$ for all $\alpha \in R_\vrho$,
i.e., each reaction has occurred by time $t_\mathrm{pos}$.
Let $\vb = \vrho(t_\mathrm{pos})$.

By applying the first part of the argument, we can find a state $\va$ with $[\va] = \producible(\vc)$ such that $\vc \segto^m \va$ and $\va \to^1 \vb$. 
Now let $\vd_t = \vrho(t)$ for every time $t > t_\mathrm{pos}$. Because $\vrho$ restricted to $[t_\mathrm{pos}, t]$ gives a finite trajectory from $\vb$ to $\vd_t$, we know by the first part of the argument that $\vb \segto \vd_t$, so $\va \segto \vd_t$. By Lemma~\ref{lem-all-producible-present} we see that $\va \to^1 \vd_t$. Since $\vd = \lim_{t \to \infty} \vd_t$ we see by Lemma~\ref{lem-one-seg-closed} that $\va \to^1 \vd$. 
As a result, we conclude that $\vc \segto \vd$.
\end{proof}

Recall mass-action trajectories correspond to valid rate schedules by Lemma~\ref{lem:mass-action-gives-valid-rate-schedule}.
Thus Theorem~\ref{thm:valid-reachable-implies-segment-reachable} implies the following corollary, which intuitively says that if a state is reachable via a mass-action trajectory 
(even in the limit of infinite time), 
then it is segment-reachable.\footnote{
Although Lemma~\ref{lem:mass-action-gives-valid-rate-schedule} has the precondition that the mass-action trajectory be defined for all time,
states reached in finite time by diverging mass-action CRNs can also be segment-reached.
For example, for the CRN $2X \to 3X$ (with rate constant 1) starting in $\{1 X\}$, 
which diverges as $t \to 1$,
all states on the trajectory prior to time $t = 1$ are segment reachable:
For each such state, we can construct a valid rate schedule that obeys mass-action until reaching that state, and then is constant for all later time.
}

\begin{cor}  
\label{cor:mass-action-reachable}
    Fix an assignment of positive mass-action rate constants for a given CRN.
    Let $\vc,\vd$ be two states
    such that $\vd$ is mass-action reachable from $\vc$.
    Then $\vc \segto \vd$.
\end{cor}

\section{Stable Computation}
\label{subsec-stable-crns}

We now use segment-reachability (\Cref{defn-reachable-segment}) to formalize
what it means for a CRN to \emph{stably compute} a function (\Cref{def:direct-computation}).
The notion of stable computation is motivated by, 
and is essentially identical to, the definition of stable computation for population protocols and discrete CRNs~\cite{AngluinAE2006semilinear, CheDotSolNaCo14}.

In this section, we justify stable computation by arguing for \emph{necessity}: CRCs that we can reasonably call ``rate-independent'' must obey stable computation.
Thus stable computation is immediately useful for negative (impossibility) results: 
showing a function cannot be stably computed implies it is not rate-independent in the desired intuitive sense.
In the next section (\cref{sec:fair-computation}),
we address the other (sufficiency) direction,
and connect stable computation to another notion of computation based on convergence in the limit as $t \to \infty$
that provides very strong guarantees for the desired rate-independent behavior of our constructions.
Based on this connection,
stable computation is taken as the primary definition of rate-independent computation in this work.

First, to formally define what it means for such a CRN to compute a function in any sense, we single out some aspects of the CRN as semantically meaningful.
Formally, a \emph{chemical reaction computer (CRC)} is a tuple $\calC = (\Lambda,R,\Sigma,\Gamma)$, where $(\Lambda,R)$ is a CRN, $\Sigma \subsetneq \Lambda$, written as $\Sigma = \{X_1,\ldots,X_k\}$,\footnote{We assume a canonical ordering of $\Sigma=\{X_1,\ldots,X_k\}$ so that a vector $\vx\in\Rp^k$ (i.e., an input to $f$) can be viewed equivalently as a state $\vx\in\Rp^\Sigma$ of $\calC$ (i.e., an input to $\calC$).
Note that we have defined valid initial states to contain only the input species $\Sigma$; other species must have initial concentration 0.
Our results would change slightly if we relaxed this assumption---see \Cref{sec:initial-context}.
} 
is the set of \emph{input species}, and $\Gamma \subsetneq \Lambda \setminus \Sigma$ is the set of \emph{output species}.
Input and output values can also be encoded indirectly via combinations of species.
An important encoding for the purposes of this paper will be the dual-rail representation (discussed in \cref{sec:results-statement}), which can handle both positive and negative quantities, and allows for easier composition of CRN ``modules''.
Since we focus on single-output functions, we will have either a single output species $\Gamma = \{Y\}$, or in the case of dual-rail computation two output species $\Gamma = \{Y^+, Y^-\}$.

We now define output stable states and stable computation.
Intuitively, output stable states are ``ideal'' output states for rate-independent computation: the output is correct and no rate law can change it. 
Stable computation is then defined with respect to output stable states,
by requiring that the correct output stable state remains reachable no matter what ``devious rate laws'' may do.
Although it is not obvious that the notion of output stable states remains pertinent when transferred from the discrete setting to the continuous one (see the discussion at the beginning of \cref{sec:fair-computation}),
\cref{lem:not-stably-compute-implies-can-reach-far-from-correct} below and the subsequent results of \cref{sec:fair-computation} show that output stability remains crucial.

\begin{definition}
\label{defn:output-stable}
A state $\vo \in \Rp^\Lambda$ is \emph{output stable} if, for all $\vo'$ such that $\vo \segto \vo'$, 
$\vo \upharpoonright \Gamma = \vo' \upharpoonright \Gamma$,  
i.e., once $\vo$ is reached, no reactions can change the concentration of any output species.
\end{definition}

Note that for a single output species $Y$, \Cref{defn:output-stable} says that $\vo(Y) = \vo'(Y)$ for all $\vo'$ such that $\vo \segto \vo'$.
For the sake of brevity and readability, 
subsequently we will state many definitions and formal theorem/lemma statements assuming there is only a single output species $Y$.
In each case, there is a straightforward modification of the definition or result so that it applies to CRCs with multiple output species as well.

\begin{definition} 
\label{def:direct-computation}
Let $f:\Rp^k \to \Rp$ be a function
and let $\calC$ be a CRC.
We say that $\calC$ \emph{stably computes} $f$ if, for all $\vx \in \Rp^k$, for all $\vc$ such that $\vx \segto \vc$, there exists an output stable state $\vo$ such that $\vc \segto \vo$ and $\vo(Y) = f(\vx)$.
\end{definition}

To extend our results to functions with $l$ outputs we can compute $l$ separate functions $f_j:\Rp^k\to\Rp$ for $1 \leq j \leq l$ with $l$ independent CRCs, and then combine them into a single CRC with $l$ output species.
In particular, we can use reactions like $X_i \to X_i^1 + \ldots + X_i^l$ to copy input $X_i$ to each of the $l$ CRCs.


We now capture in a theorem the intuition that for a CRC to compute a function rate-independently in any reasonable sense, 
it must stably compute the function.
The theorem says that if a CRC does not stably compute, 
then, no matter what you do, an adversary can ``fight back'' and make the output substantially ($\epsilon$) wrong.
The proof uses the definitions of partial states and partial reachability, as well as \Cref{thm:partial-state-reachable}, which are in \Cref{sec:partial-reachability}.

\begin{thm}
\label{lem:not-stably-compute-implies-can-reach-far-from-correct}
Suppose a CRC does not stably compute $f:\Rp^k \to \Rp$.
Then there is $\epsilon > 0$, input state $\vx$ and state $\vz$ reachable from $\vx$ such that for all $\vo$ reachable from $\vz$ there is $\vo'$ reachable from $\vo$ such that $\abs{\vo'(Y) - f(\vx)} > \epsilon.$ 
\end{thm}

\begin{proof}
We prove the contrapositive. Suppose that for all $\epsilon > 0$, for any given input state $\vx$ and state $\vz$ such that $\vx \segto \vz$, there exists a state $\vo$ such that $\vz \segto \vo$ and for all $\vo'$ such that $\vo \segto \vo'$, $\abs{\vo'(Y) - f(\vx)} \leq \epsilon$.

For any input state $\vx$ and any $\vz$ reachable from $\vx$, first we argue that there is an infinite sequence of states $\vx_{1/2}, \vx_{1/3}, \vx_{1/4}, \ldots$ such that
$\vx \segto \vz \segto \vx_{1/2} \segto \vx_{1/3} \segto \vx_{1/4} \segto \dots$,
and for all $n \geq 2$,
for all $\vo'$ such that $\vx_{1/n} \segto \vo'$,
$\abs{\vo'(Y) - f(\vx)} \leq 1/n$.
In other words, there is a sequence of states we can visit, where the adversary has less and less freedom to push the output away from the target value $f(\vx)$.
This is true by induction on $n$,
choosing $\epsilon = 1/(n+1)$, $\vz = \vx_{1/n}$ and $\vo = \vx_{1/(n+1)}$ in the above assumption.

By the definition of $\vx_{1/n}$, we see that $|\vx_{1/n}(Y) - f(\vx)| \le 1/n$, so $\vx_{1/n}(Y)$ converges to $f(\vx)$ as $n \to \infty$. Let $\Delta = \set{Y}$ and let $\vp \in \Rp^\Delta$ be the partial state with $\vp(Y) = f(\vx)$. Then we see that $\vz \segtoio \vp$ via $(\vx_{1/i})_{i=2}^\infty$, so by \cref{thm:partial-state-reachable}, we can find a partition of $\Lambda$ into $\Lambda_\bdd$ and $\Lambda_\unbdd$ with $Y \in \Lambda_\bdd$, a state $\vy$ ($\vd$ in \Cref{thm:partial-state-reachable}), 
and a subsequence $\seq{\vx'_i}_i$ of $\seq{\vx_{1/n}}_n$ so that $\vz \segto \vy$ and the subsequence $\seq{\vx'_i}$ has the property that $\vx'_i(S) \to \infty$ for all $S \in \Lambda_\unbdd$ and $\vx'_i(S) \to \vy(S)$ for all $S \in \Lambda_\bdd$. Note that because $Y \in \Lambda_\bdd$ we have $\vy(Y) = \lim_{i \to \infty} \vx'_i(Y) = f(\vx)$. 

We now claim that $\vy$ is also output stable, which is sufficient to prove the lemma as follows.
Since $\vy$ is correct ($\vy(Y)=f(\vx)$) and reachable from $\vz$, an arbitrary state reachable from input state $\vx$,
this establishes that the CRC in fact stably computes $f$.

Suppose for the sake of contradiction that $\vy$ is not output stable.
Then there is some $\vy'$ with $\vy \segto \vy'$ and $|\vy'(Y) - \vy(Y)| > \epsilon$. Let $\vz_i = \vx'_i - \vy/2 + \vy'/2$. Because $\vx'_i(S) \to \vy(S)$ for all $S \in \Lambda_\bdd$ and $\vx'_i(S) \to \infty$ for all $S \in \Lambda_\unbdd$, there is some $N$ so that for all $i > N$ we have $\vx'_i \ge \vy/2$. 
Then by additivity of $\segto$,
for all $i > N$,
\[\vx'_i = (\vx'_i - \vy/2) + \vy/2 \segto (\vx'_i - \vy/2) + \vy'/2  = \vz_i\]
where $|\vz_i(Y) - \vx'_i(Y)| = |\vy'(Y) - \vy(Y)|/2 > \epsilon/2$. Since $\seq{\vx'_i}$ is a subsequence of $\seq{\vx_{1/n}}$, if we take $i$ large enough so that $\vx'_i$ is $\vx_{1/n_i}$ with $1/n_i < \epsilon/4$, we see that
\[|\vz_i(Y) - f(\vx)| \ge |\vz_i(Y) - \vx'_i(Y)| - |\vx'_i(Y) - f(\vx)| \ge \epsilon/2 - 1/{n_i} > 1/n_i,\]
but $\vz_i$ is reachable from $\vx_{1/n_i}$, which by definition can only reach states $\vz$ with $|\vz(Y) - f(\vx)| \leq 1/n_i$,
giving a contradiction.
\end{proof}

In some places we will talk about CRCs with multiple output species representing multi-valued functions:

\begin{cor}\label{cor:multiple-output-not-stably-compute-implies-can-reach-far-from-correct}
Suppose a CRC does not stably compute $f: \Rp^k \to \Rp^l$. Then there is an $\epsilon > 0$, an input state $\vx$ and a state $\vz$ reachable from $\vx$ so that for all $\vo$ reachable from $\vz$ there is $\vo'$ reachable from $\vo$ such that $|\vo'(Y_i) - f(\vx)_i| > \epsilon$ for some $1 \le i \le l$.  
\end{cor}

\begin{proof}
This is similar to the proof of \Cref{lem:not-stably-compute-implies-can-reach-far-from-correct}, 
but make $\Delta = \set{Y_1, \ldots, Y_l}$ the set of all output species. 
\end{proof}

\section{Fair Computation}
\label{sec:fair-computation}

In the discrete model of CRN kinetics,
if the set of states reachable from any input state is finite (i.e., the molecular counts are bounded as a function of the input state),
then stable computation as in \Cref{def:direct-computation}
(the correct output stable state is always \emph{reachable})
is equivalent to the condition that the CRC is correct under standard stochastic kinetics with probability 1
(the correct output state state is actually \emph{reached})~\cite{AngluinAE2006semilinear}.
In the continuous CRN model, however, it might seem that the idea of stable computation is not strong enough to achieve intuitively ``rate-independent'' computation.
There are at least two reasons for the concern.

First, it is possible that the output stable state is always reachable but the mass-action trajectory does not converge to it.
For example, 
consider the following CRC stably computing the identity function $f(x) = x$:
\begin{align*}
X + X &\to Y + Y \\
Y + X &\to X + X
\end{align*}
with $X$ as the input species and $Y$ as the output species.
From any reachable state we can reach the output stable state with all $X$ converted to $Y$.
However, the mass-action trajectory converges to a dynamic equilibrium with $k_2/(2 k_1 + k_2)$ fraction of $X$, where $k_1$, $k_2$ are rate constants of the two reactions.
\footnote{
The discrepancy between stable computation and correctness under mass-action kinetics shows a major difference between the discrete and continuous CRN models.
In the example above, 
with $n$ total molecules,
the discrete CRN model does a random walk on the number of $X$ that is biased toward the dynamic equilibrium point $n k_2 / (2k_1 + k_2)$.
Despite the bias upward when $X$ is below this value, there is always a positive probability to decrease $X$, so with probability 1, $X$ will eventually reach 0.
}
This shows that in general, showing that a CRC stably computes is not sufficient to claim that it computes rate-independently in any intuitive sense.

The second difficulty lies with the notion of output stable states. 
While the notion of output stable states is natural for discrete CRNs where we want the system to actually reach that state to end the computation, convergence to the output stable state in continuous CRNs will typically be only in the limit $t \to \infty$.
For example, consider the following CRC:
\begin{align*}
X &\to M + Z \\
M &\to Y \\
Z + Y &\to Z + M \\
Z &\to \emptyset
\end{align*}
The CRC stably computes $f(x) = x$ since from any reachable state we can reach the output stable state $\vo$ with $f(x)$ amount of $Y$ by
converting any remaining $X$ to $M$, converting any remaining $M$ to $Y$, and completely draining $Z$.
Note that $\vo$ is output stable since without $Z$, $Y$ cannot be converted back to $M$.
Further, under mass-action kinetics (for any choice of rate constants), the CRC converges to $\vo$ since as $Z$ drains, the rate of the third reaction converges to $0$.
However, at \emph{every} finite time $0 < t< \infty$ in the mass-action trajectory, since $Z$ is present, the state with zero amount of $Y$ is reachable, so an adversary could substantially perturb the output. 
Thus one would not call this CRC rate-independent to adversarial perturbations.

While in the previous section we argued that stable computation is \emph{necessary} for an intuitive notion of rate-independent computation,
the examples above seem to suggest that it is not \emph{sufficient} and that basing rate-independent computation entirely on stable computation could be ill-founded. 

In this section we 
develop an alternative approach to defining a very strong notion of ``rate-independent'' computation not based on stable computation, an approach we term ``fair computation''. 
The approach is based on delineating a very broad class of rate laws, possibly adversarial, that still lead to convergence to the correct output. 
Based on the previous section (\cref{lem:not-stably-compute-implies-can-reach-far-from-correct}), 
it is not surprising that CRCs that fail to stably compute, also fail to fairly compute.
What is more surprising, however, is that there is a strong connection in the other direction for a class of CRCs (feedforward)---for these CRCs stable computation \emph{implies} fair computation.
All our constructions will be in this class; thus we obtain very strong rate-independence guarantees in the positive results part of this work. 
Combined with the results of the previous section, stable computation can thus be used as an easy-to-analyze proxy for proving both positive and negative results on rate-independent computation.

Intuitively, a CRC is said to fairly compute if it converges to the correct output for a broad class of rate laws, with the class being broad enough to capture adversarial behavior.
To define the broad class of rate laws for fair computation, 
we start with the previously defined notion of valid rate schedules, 
that captures a very general class of chemical kinetics.
Nonetheless, we must add an additional condition.
In our original definition (\Cref{defn:valid-rate-schedule}), the reaction rate $\vf_\alpha$ can vary arbitrarily over time as long as $\alpha$ is applicable whenever $\vf_\alpha$ is positive.
There is no requirement the other way---that a reaction must occur with positive rate if it is applicable---allowing for a greater variety of paths (e.g., segment paths with zero flux through some reactions).
But since there is nothing to prevent an adversary from ``starving'' reactions when they are applicable, preventing convergence, we now need to impose an additional requirement that we call fairness.
We formalize this as a strictly positive lower bound $\vH$ on the reaction rate at states where the reaction is applicable.
In particular, while the reaction rate vector $\vf(t)$ is a function of the time $t$, $\vH(\vc)$ is a function of the \emph{state} $\vc$.
We allow this lower bound to be violated occasionally, so long as it holds for an infinite measure of time.
(For example, a fair rate schedule could starve applicable reactions on the unit time intervals 
$[0,1], [2,3], [4,5], \ldots$
)

\begin{definition}
\label{defn:fair-rate-schedule}
Suppose $\vf: \Rp \to \Rp^R$ is a valid rate schedule starting at $\vc$.
We say that $\vf$ is \emph{fair} if there is a continuous function $\vH: \Rp^\Lambda \to \Rp^R$ such that, 
for all reactions $\alpha \in R$, 
$\vH_\alpha(\vc) > 0$ if and only if $\alpha$ is applicable at $\vc$, 
and for some subset of times $T_\alpha \subseteq \Rp$ of infinite measure, 
$\vf_\alpha(t) \ge \vH_\alpha(\vrho(t))$ for all $t \in T_\alpha$.
\end{definition}

Requiring that the lower bound $\vH$ be continuous as a function of the state helps to ensure that if $\vH$ converges to zero then the point of convergence is a static state (no reaction is applicable):
by continuity the point of convergence must have $\vH = \vec{0}$, which implies that no reaction is applicable by \cref{defn:fair-rate-schedule}.
Note also that $T_\alpha$ can depend on $\alpha$; for example \Cref{defn:fair-rate-schedule} allows for $T_\alpha$ and $T_\beta$ to be disjoint for different reactions $\alpha,\beta$ (i.e., whenever we run one applicable reaction we starve another applicable reaction).

All typically considered rate laws such as mass-action, Michaelis-Menten, Hill function kinetics, etc, are fair. We explicitly note this for mass-action CRNs with non-divergent trajectories:

\begin{lemma}
\label{lem:mass-action-gives-fair-rate-schedule}
The valid rate schedule for mass-action CRNs as defined in Lemma~\ref{lem:mass-action-gives-valid-rate-schedule}, is fair if well-defined for all times.
\end{lemma}

\begin{proof}
In the notation of \cref{lem:mass-action-gives-valid-rate-schedule}, take $\vH = \vA$. Because we always assume that the rate constants of a mass action system are positive, $\vA_\alpha(\vc)$ is positive if and only if $\alpha$ is applicable at $\vc$. Also, Lemma~\ref{lem:mass-action-gives-valid-rate-schedule} shows that $\vf_\alpha(t) = \vA_\alpha(\vrho(t))$ for all $t \ge 0$, so certainly $\vf_\alpha(t) \ge \vA_\alpha(\vrho(t))$.
\end{proof}

CRCs satisfying the definition below converge to the correct output despite actions of a very powerful adversary.
Intuitively, the adversary is allowed to control the rates of all the reactions throughout the computation as long as applicable reactions are not entirely prevented from occurring.

\begin{defn}
\label{defn:fair-computation}
We say a CRC \emph{fairly computes} a function $f:\Rp^k \to \Rp$ if,
for every input state $\vx$,
every fair rate schedule starting at $\vx$ with trajectory $\vrho$ obeys $\lim\limits_{t\to\infty} (\vrho(t) \upharpoonright \{Y\}) = f(\vx).$
\end{defn}

There is a natural generalization of the above definition when $f: \Rp^k \to \Rp^l$ has multiple outputs: require the trajectory $\vrho$ to obey $\lim_{t \to \infty} (\vrho(t) \upharpoonright \set{Y_1, \ldots, Y_l}) = f(\vx)$, where $\set{Y_1, \ldots, Y_l}$ is the set of all output species.

We now want to establish the connection between fair computation and stable computation.
The following lemma shows that in one direction, the connection is immediate: fair computation is at least as strong as stable computation.

\begin{lem}
\label{lem:fair-computation-implies-stable-computation}
Any CRC that fairly computes a function $f$ also stably computes $f$.
\end{lem}

\begin{proof}
We will prove the contrapositive: Suppose that a CRC $\calC$ does not stably compute $f$. We want to show that $\calC$ does not fairly compute $f$, and to do this we will find an initial state $\vx$ and a fair rate schedule $\vf$ starting at $\vx$ so that the trajectory $\vrho$ does not obey $\lim_{t \to \infty}(\vrho(t) \upharpoonright \set{Y}) = f(\vx)$. 

Since we assumed that $\calC$ does not stably compute $f$, by \Cref{lem:not-stably-compute-implies-can-reach-far-from-correct} we know that there is some $\epsilon > 0$ and some input state $\vx$ and state $\vz$ so that $\vx \segto \vz$ and every $\vo$ reachable from $\vz$ there is an $\vo'$ reachable from $\vo$ with $|\vo'(Y) - f(\vx)| > \epsilon$. 
In other words, it is possible to reach a state $\vz$, after which the output can be always made incorrect by some amount $\epsilon$.
By \Cref{lem:segment-reachable-implies-valid-rate-schedule} we know that there is a valid rate schedule $\vf'$ that reaches from $\vx$ to $\vz$ in finite time $t_0$. Set $\vf = \vf'$ on the interval $[0, t_0)$. 

Now choose some positive rate constant $k_\alpha > 0$ for every reaction $\alpha$, 
and on the interval $[t_0, t_0 + 1)$,
set $\vf$ equal to the mass-action rate schedule corresponding to the rate constants $k_\alpha$ and initial state $\vz$ (this is a valid rate schedule by \Cref{lem:mass-action-gives-valid-rate-schedule}). By \Cref{thm:valid-reachable-implies-segment-reachable} we know that the state $\vo$ that is reached at time $t_0 + 1$ is segment reachable from $\vz$, so by the definition of $\vz$ there must be some $\vo'$ with $\vo \segto \vo'$, where $|\vo'(Y) - f(\vx)| > \epsilon$. Applying \Cref{lem:segment-reachable-implies-valid-rate-schedule} again we see that there is a valid rate schedule that reaches from $\vo$ to $\vo'$ by time $t_1$, so set $\vf$ equal to this rate schedule on $[t_0 + 1, t_1)$. 

Alternating a mass-action rate schedule with an adversarial rate schedule in this way, we can find a sequence of times $t_0, t_1, \ldots \in \Rp$ so that $\lim_{n \to \infty} t_n = \infty$ and a valid rate schedule $\vf$ so that the associated trajectory $\vrho$ satisfies $|\vrho(t_n)(Y) - f(\vx)| > \epsilon$ for every $n \ge 1$. This proves that $\lim_{t \to \infty} (\vrho(t) \upharpoonright \set{Y}) \neq f(\vx)$. Furthermore, by construction we know that on every interval $[t_n, t_n + 1)$ the rate schedule $\vf$ is equal to the mass-action rate schedule with fixed rate constants $k_\alpha$. Since $\bigcup_{n \in \N} [t_n, t_n + 1)$ has infinite measure, \Cref{lem:mass-action-gives-fair-rate-schedule} then shows that $\vf$ is fair. 
Since $\vf$ is a fair rate schedule failing to converge to the correct output,
$\calC$ does not fairly compute $f$. 
\end{proof}

It is natural to wonder if the converse of \cref{lem:fair-computation-implies-stable-computation} holds,
i.e., whether every CRC that stably computes $f$ also fairly computes $f$. 
This is not true in general, but the following section shows that it \emph{is} true for the CRNs we will construct.

\subsection{Feedforward CRNs}
\label{sec:feedforward}

As we saw before, for general CRCs it is possible that the output stable state is always reachable but the mass-action trajectory does not converge to it, for instance the reactions $X + X \to Y + Y$ and $Y + X \to X + X$ discussed at the start of \Cref{sec:fair-computation}.
Thus stable computation does not necessarily imply that the system will eventually produce the correct output.
Since the mass-action trajectory defines a fair rate schedule (\cref{lem:mass-action-gives-fair-rate-schedule}), the above example shows that some CRCs stably compute a function but do not fairly compute it; i.e., the converse of \cref{lem:fair-computation-implies-stable-computation} does not hold for all CRCs.

In contrast to the above example, the feedforward property defined in this section allows us to bridge the definition of stable computation, defined in terms of reachability (what \emph{could} happen), to convergence (what \emph{will} happen), defined in terms of fair rate schedules.

Recall that a reaction $\alpha=\langle  \vr,\vp \rangle$ \emph{produces} a species $S$ if $\vr(S) < \vp(S)$ and \emph{consumes} $S$ if $\vr(S) > \vp(S)$.
We say a CRN is \emph{feedforward} if the species can be ordered so that every reaction that produces a species also consumes another species earlier in the ordering.
Formally:
\begin{definition}\label{defn:feedforward}
A CRN $\calC = (\Lambda, R)$ is \emph{feedforward} if $\Lambda=\set{S_1, ..., S_n}$ and its stoichiometry matrix satisfies $\vM(i,j) > 0 \implies \exists (i' < i) \vM(i',j) < 0$.
\end{definition}

Intuitively, we want to avoid situations, as in the example above,  where the output stable state is always reachable but the trajectory does not converge to it.
This can happen if the trajectory does not converge or converges to a dynamic equilibrium where reactions balance each other.
In contrast, the total flux of reactions in a feedforward CRN must be bounded because there cannot be a complete ``cycle'' among the species that balance consumption with production.

We start with the following simple observation.
General CRNs can have reactions such as $A+B \to A+2B+C$ that do not consume any reactant, and such CRNs clearly have infinite total flux.
Luckily, by our definition of a CRN, any reaction $\alpha$ must either produce or consume some species, 
and if $\alpha$ produces a species, the feedforward condition guarantees that $\alpha$ consumes some other species.

\begin{obs}
\label{obs:feedforward-reactions-consume-species}
Every reaction in a feedforward CRN consumes some species.
\end{obs}

The following lemma will used in formalizing the idea that a feedforward CRN cannot converge to a dynamic equilibrium (like the example beginning this section).
General CRNs can have reactions that undo each other's effect (for instance, $A \to B$ and $B \to A$). 
For such CRNs, we cannot bound total flux as a function of the change in species concentration---indeed, concentrations might remain constant but the two reactions canceling each other can have arbitrarily large flux---allowing for a dynamic equilibrium.
In contrast, for feedforward CRNs the following lemma shows that total flux can be bounded by the change in state.

\begin{lemma}
\label{lem:feedforward-implies-no-cancellation}
For a feedforward CRN, for each reaction $\alpha$, there is a constant $K_\alpha$ independent of $\vu$ such that $\|\vM \vu\| < \epsilon$ implies that $\vu_\alpha < K_\alpha \epsilon$.
\end{lemma}

\begin{proof}
We show by induction that we can find such a $K_\beta$ for every reaction $\beta$ that consumes a species. First suppose $\beta$ is a reaction that consumes the first species $S_1$ in the feedforward ordering.
By the feedforward property, no reaction can produce $S_1$, so for each reaction $\alpha$, we have $\vM(S_1,\alpha) \leq 0$. 
Thus 
\[|(\vM \vu)(S_1)| = \sum_{\alpha \in R} |\vM(S_1, \alpha)| \vu_\alpha.\]
Since the left hand side of this equation is less than $\epsilon$, and the right hand side is a sum of non-negative terms, each term on the right hand side must be less than $\epsilon$. Since $\beta$ consumes $S_1$, $|\vM(S_1, \beta)| \ge 1$, so $\vu_\beta < \epsilon$ and we can take $K_\beta = 1$. 
This establishes the base case for $S_1$.

Now assume inductively that we've found an appropriate constant $K_\alpha$ for every reaction $\alpha$ that consumes a species $S_i$ for $i < n$, and suppose that $\beta$ consumes $S_n$. Then because $|(\vM\vu)(S_n)| < \epsilon$ we know that
\[\sum_{\substack{\alpha \in R \\ \vM(S_n, \alpha) < 0}} |\vM(S_n, \alpha)| \vu_\alpha 
< \epsilon + \sum_{\substack{\alpha \in R \\ \vM(S_n, \alpha) > 0}} |\vM(S_n, \alpha)| \vu_\alpha 
< \epsilon \cdot \left[ 1 + \sum_{\substack{\alpha \in R \\ \vM(S_n, \alpha) > 0}} |\vM(S_n, \alpha)| K_\alpha \right]\]
where for the second inequality we have used the feedforward condition to conclude that every reaction producing $S_n$ must consume $S_i$ for some $i < n$, and therefore have flux bounded by $K_\alpha \epsilon$ by inductive assumption. Like before, the leftmost term in the inequality is a sum of non-negative terms, so if we take 
\[K_\beta = 1 + \sum_{\substack{\alpha \in R \\ \vM(\alpha, S_n) > 0}} |\vM(S_n, \alpha)| K_\alpha\]
then $\vu_\beta < K_\beta \epsilon$. This shows that an appropriate $K_\beta$ exists for every reaction $\beta$ that consumes a species. But by \Cref{obs:feedforward-reactions-consume-species}, every reaction consumes a species, so we're done.
\end{proof}

Finally we are ready to prove the main lemma about feedforward CRNs, a consequence of which (\Cref{lem:stable-computation-implies-fair-computation-feedforward}) establishes the connection between stable computation and convergence to the output stable state for feedforward CRNs.

\begin{lemma}
\label{lem:fair-rate-schedule-converges}
Consider a feedforward CRN, and 
suppose $\vf$ is a valid rate schedule.
Then the corresponding trajectory $\vrho$ converges to a state 
$\vy$ in the limit time $t \to \infty$.
Further, if $\vf$ is fair, then $\vy$ is static, 
i.e., no reaction is applicable in $\vy$.
\end{lemma}

\begin{proof}
First, define
\[K = 1 + \max_{\alpha \in R} \sum_{S_i \in \Lambda} \max(0, \vM(S_i,\alpha)).\]
In particular, $K$ is larger than the sums of the positive entries in the columns of $\vM$. Intuitively, we will assign a ``mass'' to each species, and $K$ will be the ratio of the masses assigned to consecutive species in the feedforward ordering. By making $K$ sufficiently large, we can guarantee that running any reaction has the effect of decreasing the total mass $V(\vx)$ of any state $\vx$.
Formally, define $V: \Rp^\Lambda \to \Rp$ as
\[V(\vx) = \sum_{i =1}^n \frac{\vx(S_i)}{K^i}.\]

Recall $\vf_\alpha(t)$ is the rate of reaction $\alpha$ at time $t$. 
Note that whenever $\vrho(t)$ is differentiable\footnote{
By Lebesgue's fundamental theorem of calculus \cite[Theorem~6.11, Theorem~6.14]{royden1988real}, applied to $\vf_\alpha(t)$, we know that $\vrho(t)$ is locally absolutely continuous, and almost everywhere differentiable with derivative $\vM \cdot \vf(t)$.  
}
\begin{align*}
\frac{d}{dt} V(\vrho(t)) &= \sum_{i =1}^n \frac{d}{dt}\frac{(\vrho(t))(S_i)}{K^i} \\
&= \sum_{i =1}^n \frac{1}{K^i} \sum_{\alpha \in R} \vM(S_i,\alpha) \vf_\alpha(t) \\
&= \sum_{\alpha \in R} \sum_{i =1}^n \frac{1}{K^i}  \vM(S_i,\alpha) \vf_\alpha(t) \\
&= \sum_{\alpha \in R} C_\alpha \vf_\alpha(t),
\end{align*}
where $C_\alpha = \sum_{i =1}^n \frac{1}{K^i}  \vM(S_i,\alpha)$. For any fixed $\alpha$, let $i_0$ be the smallest $i$ such that $\vM(S_i, \alpha) \neq 0$,
i.e., $S_{i_0}$ is the first species in the feedforward ordering that is produced or consumed by $\alpha$. 
By the feedforward condition, $\vM(S_{i_0}, \alpha) \leq -1$, i.e., $S_{i_0}$ is consumed. 
As a result,
\begin{align*}
C_\alpha &= \frac{1}{K^{i_0}}\left(\vM(S_{i_0}, \alpha) + \frac{1}{K} \sum_{i = i_0 + 1}^n \frac{\vM(S_i,\alpha)}{K^{i - i_0 - 1}}\right) \\
&< \frac{1}{K^{i_0}}\left(\vM(S_{i_0}, \alpha) + 1\right) \\
&\le 0.
\end{align*}

Now we show that the total flux through all of the reactions is finite. To fix notation, let's write $\vF_\alpha$ for the total flux through reaction $\alpha$ as $t \to \infty$, i.e. 
\[\vF_\alpha = \int_0^\infty \vf_\alpha(t) dt.\]
Since every $C_\alpha$ is negative, for any fixed reaction $\alpha_0$, 
\begin{align*}
\frac{d}{dt} V(\vrho(t)) &= \sum_{\alpha \in R} C_\alpha \vf_\alpha(t) \le C_{\alpha_0} \vf_{\alpha_0}(t)
\end{align*}
so by integrating\footnote{
Note that $\int_0^\infty \frac{d}{dt} V(\vrho(t)) dt = \lim_{t \to \infty} V(\vrho(t)) - V(\vc)$ by the fundamental theorem of calculus \cite[Theorem~6.10]{royden1988real}.
To use the fundamental theorem of calculus we must ensure that $V(\vrho(t))$ is locally absolutely continuous---this follows from the fact that the trajectory $\vrho(t)$ is defined in terms of a (Lebesgue) integral.
} both sides
\begin{align*}
\lim_{t \to \infty} V(\vrho(t)) - V(\vc) &\le C_{\alpha_0} \int_0^\infty \vf_{\alpha_0}(t) dt = C_{\alpha_0}\vF_{\alpha_0}.
\end{align*}
If $\vF_{\alpha_0}$ were infinite then because $C_{\alpha_0} < 0$ we see that $V(\vrho(t))$ would be unbounded below. 
Since $\vrho(t)$ always remains in $\Rp^\Lambda$ where 
$V(\vrho(t))$
$\ge 0$, we conclude that $\vF_{\alpha_0}$ must be finite.

If we write $\vecv_\alpha \in \R^\Lambda$ for the $\alpha$-column of $\vM$ (i.e. $\vecv_\alpha(S) = \vM(S, \alpha)$), then 
\begin{align*}
\int_0^\infty \left\|\frac{d\vrho(t)}{dt}\right\| dt &= \int_0^\infty \left\|\sum_{\alpha \in R} \vecv_\alpha \vf_\alpha(t)\right\| dt \\
&\le \int_0^\infty \sum_{\alpha \in R} \vf_\alpha(t)|\vecv_\alpha| dt \\
&= \sum_{\alpha \in R}|\vecv_\alpha| \int_0^\infty  \vf_\alpha(t) dt  = \sum_{\alpha \in R}|\vec{v}_\alpha| \vF_\alpha\\
&< \infty.
\end{align*}
In other words, $\vrho$ has finite length. 

{
Let's now show that $\vrho$ having finite length implies that it converges. 
For any states $\vx,\vz$ define $d(\vx,\vz) = \| \vx - \vz \|$ as the Euclidean distance from $\vx$ to $\vz$.
It suffices to show 
}
\begin{equation}\label{eqn:convergence-statement}
\mbox{For all $\epsilon > 0$ there is some $M \in \Rp$ such that for all $t, s > M$, $d(\vrho(t), \vrho(s)) < \epsilon$.}\tag{$*$}
\end{equation}
Indeed, taking this as given, let $\vx_k = \vrho(k)$. Then the sequence $\seq{\vx_k}_{k = 1}^\infty$ is a Cauchy sequence in $\R^\Lambda$, so it must converge to some state $\vy$. Moreover, we actually know that $\vrho(t)$ converges to 
some state
$\vy$ as $t \to \infty$: for any $\epsilon > 0$ there is some $N$ such that for all $n > N$, we know $d(\vx_n, \vy) < \epsilon/2$ and there is some $M$ such that for all $t, s > M$ we know $d(\vrho(t), \vrho(s)) < \epsilon/2$. In particular, taking $m$ to be an integer larger than $N$ and $M$ we see that 
\[d(\vrho(t), \vy) \le d(\vrho(t), \vx_m) + d(\vx_m, \vy) < \epsilon\]
for any $t > M$. This shows that, taking (\ref{eqn:convergence-statement}) for granted, $\vrho$ must converge to $\vy$. 

Let's now prove (\ref{eqn:convergence-statement}). Proceed by contradiction, so
suppose that there is some $\epsilon > 0$ such that for all $M \in \mathbb{R}_{\ge 0}$, there exists some $t > s > M$ such that $d(\vrho(t), \vrho(s)) \ge \epsilon$. Then take $M_1 = 0$, and label the points we get by this assumption $t_1, s_1$. Then for any $n > 1$, take $M_n = t_{n-1}$ and label the next pair of points $t_n, s_n$. Then 
\begin{align*}
\int_0^\infty \left\|\frac{d\vrho(t)}{dt}\right\| dt &\ge \sum_{i = 1}^\infty \int_{s_i}^{t_i}\left\|\frac{d\vrho(t)}{dt}\right\| dt \\
&\ge \sum_{i = 1}^\infty d(\vrho(t_i), \vrho(s_i))  \\
&\ge \sum_{i = 1}^\infty \epsilon = \infty,   
\end{align*}
where the second inequality uses the absolute continuity of $\vrho$. This gives a contradiction, establishing (\ref{eqn:convergence-statement}),
which proves that $\vrho(t)$ converges to some state $\vy$ as $t \to \infty$.

Now let's assume that $\vf$ is fair and establish that $\vy$ is a static state.
If not then there is some reaction $\alpha$ applicable at $\vy$, so $\vH_\alpha(\vy) = C > 0$.
By the continuity of $\vH_\alpha$, there is some $\delta$ so that $\vH_\alpha(\vx) > C/2$ for all $\vx \in \Rp^\Lambda$ with $\| \vx - \vy \| < \delta$. Since $\vrho$ converges to $\vy$, this implies that $\vH_\alpha(\vrho(t)) > C/2$ for all $t$ greater than some $t_0$. 

For any measurable $U \subseteq \Rp$,
let $\mu(U) \in \Rp \cup \{\infty\}$ denote the measure of $U$.
Letting $T_\alpha$ be the subset of $\Rp$ as in Definition~\ref{defn:fair-rate-schedule} where $\vf_\alpha(t) \geq \vH_\alpha(\vrho(t))$ for $t \in T_\alpha$ and $\mu(T_\alpha) = \infty$, 
we see that $\vf_\alpha(t) > C/2$ for every $t \in T_\alpha \cap [t_0, \infty)$.
Also note that $\mu(T_\alpha \cap [t_0, \infty)) = \infty$.
By the contrapositive direction of \Cref{lem:feedforward-implies-no-cancellation} we see that this implies that $\left\|\frac{d\vrho(t)}{dt}\right\| > C/2K_\alpha$ for every $t \in T_\alpha \cap [t_0, \infty)$. 
As a result,
\[\int_0^\infty \left\|\frac{d\vrho(t)}{dt}\right\| dt \ge \int_{T_\alpha \cap [t_0, \infty)} \left\|\frac{d\vrho(t)}{dt}\right\| dt \ge \mu(T_\alpha \cap [t_0, \infty))\frac{C}{2K_\alpha} = \infty,\]
contradicting the finite length of $\vrho$. 
\end{proof}

By applying \Cref{lem:fair-rate-schedule-converges} to a feedforward CRN that also stably computes a function,
we obtain the following result, 
which states that such a CRN will reach the correct output under any fair rate schedule, from any state reachable from $\vx$.
This is almost a converse to \cref{lem:fair-computation-implies-stable-computation}; however note that unlike in \cref{lem:fair-computation-implies-stable-computation}, the CRC is required to be feedforward.

\begin{lem}
\label{lem:stable-computation-implies-fair-computation-feedforward}
Any feedforward CRC that stably computes a function $f$ also fairly computes $f$.
\end{lem}

\begin{proof}
By~\Cref{lem:fair-rate-schedule-converges},
under any fair rate schedule, the CRC converges from the initial state $\vx$ to a static state $\vy$. 
By \Cref{thm:valid-reachable-implies-segment-reachable},
since $\vy$ is reachable from $\vx$ under a valid rate schedule, 
we know $\vx \segto \vy$.
Since $\vec{y}$ is static, it is output stable.
Then because the CRC stably computes $f$ and $\vx \segto \vy$, we must have $\vy(Y) = f(\vx)$ or else it would have stabilized to an incorrect output.
So the CRC fairly computes $f$.
\end{proof}

We point out that the feedforward property is not a \emph{necessary} condition for stable computation to coincide with fair computation.
For example, consider the non-feedforward CRC with reactions $X + X \rxn R + Y$ and $R + R \rxn X$.
This CRC stably computes $f(x) = (2/3) x$~\cite{vasic2022programming}, 
and it can also be shown that it fairly computes $f$. 
The key property it shares with feedforward CRCs is \Cref{lem:feedforward-implies-no-cancellation}: 
a large flux through its reactions implies a large change in state.

\begin{lemma}
\label{lem:feedforward-mass-action-defined-all-times}
If the CRC is feedforward, then the mass-action trajectory $\vrho:\Rp\to\Rp^\Lambda$ as defined in Lemma~\ref{lem:mass-action-gives-valid-rate-schedule} is defined at all times.
\end{lemma}

\begin{proof}
By the Escape Lemma \cite[Lemma 9.19]{lee2013smooth} it suffices to show that $\vrho$ remains in a compact subset of $\Rp^\Lambda$ for all times when it is defined. 
Let $V(\vx)$ be as in the proof of \Cref{lem:fair-rate-schedule-converges}. We know by the argument of \Cref{lem:fair-rate-schedule-converges} that $\frac{d}{dt} V(\vrho(t)) \le 0$, so for all $t$ we know that $\vrho(t)$ lies in the subset of $\Rp^\Lambda$ where $V(\vx) \le V(\rho(0))$. Since $V$ is a linear function with strictly positive coefficients, this is a compact subset of $\Rp^\Lambda$. 
\end{proof}

Since mass-action yields a valid rate schedule (\Cref{lem:mass-action-gives-valid-rate-schedule}) 
that is fair 
(\Cref{lem:mass-action-gives-fair-rate-schedule}) 
and defined at all times for feedforward CRNs (\Cref{lem:feedforward-mass-action-defined-all-times}), the following corollary is immediate.
The corollary states that for a CRC stably computing a function $f$, the CRC will also converge to the correct output under mass-action kinetics, no matter the positive rate constants, and even if an adversary can first ``steer'' the CRC to some reachable state before letting mass-action kinetics take over.

\begin{cor} 
\label{cor:mass-action-convergence-stable-computation}
Consider a feedforward CRC stably computing a function $f$.
Then for any input state $\vx$, for any state $\vz$ reachable from $\vx$, 
for any choice of reaction rate constants,
the mass-action trajectory starting at $\vz$ is defined for all times and converges to an output stable state $\vy$ in the limit $t \to \infty$ such that $\vy(Y) = f(\vx)$.
\end{cor}

\section{The Computational Power of Stable Computation}
\label{sec:results-statement}

This section presents the main results of our paper, delineating the computational power of stable computation.
As justified in \Cref{subsec-stable-crns,sec:fair-computation}, we use stable computation as our primary notion of rate-independence.
In \cref{subsec:dual-rail-defn}, we discuss the input and output representation of negative quantities and the composition of CRN modules via the ``dual-rail representation''.
In \cref{sec:main-results}, we summarize our results on the computational power of stable computation for direct and dual-rail representations.
These results are proven in subsequent sections, with positive (\cref{negative-piecewise-linear-implies-computable} for dual-rail, \cref{sec:positive-continuous-piecewise-rational-linear-computable} for direct) and negative directions (\cref{sec:positive-continuous-piecewise-rational-linear-computable} for direct, \cref{subsec:negative-dual-rail} for dual-rail) separately.

\subsection{Dual-rail Representations}
\label{subsec:dual-rail-defn}

The \emph{direct} concentration-to-value mapping  articulated in \cref{def:direct-computation}
is a straightforward way to represent non-negative input and output values.
However, there are two reasons why 
an alternative, albeit more complex,
encoding may be preferred. 
First, since concentrations cannot be negative, 
computable functions are restricted to non-negative domain and range.
Second, as explained below, the direct output encoding frustrates the composition of smaller CRC modules into larger CRCs.
The \emph{dual-rail} representation we introduce in this section is a natural way to solve both problems.

A natural way to represent a (possibly negative) real value in chemistry is to encode it as the difference of two concentrations.
Formally, let $f:\R^k \to \R$ be a function.
A function $\hf:\Rp^{2k} \to \Rp^2$ is a \emph{dual-rail representation} of $f$ if, for all $\vx^+,\vx^- \in \Rp^k$, if $(y^+,y^-) = \hf(\vx^+,\vx^-)$, then $f(\vx^+ - \vx^-) = y^+ - y^-$.
In other words, $\hf$ represents $f$ as the difference of its two outputs $y^+$ and $y^-$, and it works for any input pair $(\vx^+,\vx^-)$ whose difference is the input value to $f$.
We can define a CRC to stably compute such a function in the same manner as in 
\Cref{subsec-stable-crns}, but having $2k$ input species $\Sigma = \{X_1^+,X_1^-,X_2^+,X_2^-,\ldots,X_k^+,X_k^-\}$ and two output species $\Gamma = \{Y^+,Y^-\}$.

\begin{definition}
\label{def:dual-rail-stable-computation}
We say that a CRC \emph{stably dual-computes} $f:\R^k \to \R$ if it stably computes a dual-rail representation $\hf:\Rp^k \times \Rp^k \to \Rp \times \Rp$ of $f$.
\end{definition}

This definition implies that, for all $\vx = (\vx^+, \vx^-) \in \Rp^{2k}$, for all $\vc$ such that $\vx \segto \vc$, there exists an output stable state $\vo$ such that $\vc \segto \vo$ and 
$\vo(Y^+) - \vo(Y^-) = f(\vx^+ - \vx^-)$.
Note that a single function has an infinite number of dual-rail representations;
we require only that a CRC exists to compute one of them to say that the function is stably dual-computable by a CRC.

Besides making negative values chemically representable, 
we will see that the dual-rail representation plays a key role in allowing the composition of smaller CRC modules into a larger CRC. 
A key concept to enable such composition is \emph{output-obliviousness}:

\begin{defn}
\label{def:output-oblivious}
A CRC $\calC = (\Lambda,R,\Sigma,\Gamma)$ is \emph{output-oblivious} if none of its output species $\Gamma$ is a reactant in any reaction.
In other words, for every $\alpha=\langle  \vr,\vp \rangle \in R$ and $Y \in \Gamma$, $\vr(Y) = 0$.
\end{defn}

To recognize the problem with composition of the direct representation (\cref{def:direct-computation}), 
define the \emph{composition} of two CRCs as the CRC which has the union of their reactions, relabeling the output species of the upstream CRN to be the input species of the downstream one~\cite[Definition 16]{chalk2021composable}.
Then with the direct output representation, output-obliviousness is necessary for composability but provably restricts computational power~\cite{chalk2021composable}:
(1) Composing two stably computing CRCs stably computes the function composition if and only if the upstream CRC is output-oblivious (except in trivial cases).
Intuitively, the downstream CRC can interfere with the upstream computation by 
prematurely consuming the output.
(2) An output-oblivious CRC can only stably compute ``superadditive'' functions.
For example, any CRC stably computing the function $f(x_1,x_2) = x_1 - x_2$ must necessarily be able to consume its output species since 
a state with more than the desired  amount of output is reachable by additivity (e.g., state with $x_1$ amount of $Y$).

Our results imply that the dual-rail representation allows composition without sacrificing computational power. 
In particular, our dual-rail constructions are all output-oblivious and thus composable by concatenation, and our negative results apply to dual-rail CRCs whether or not they are output-oblivious.
To see roughly why the dual-rail representation helps with composition, 
consider the $f(x_1,x_2) = x_1 - x_2$ function above.
We can now compute this function with an output-oblivious (and therefore composable) CRC by \emph{producing} $Y^-$ to decrease the value of the output, without consuming any output species.

Recall fair computation (\Cref{defn:fair-computation}).
We define dual-rail fair computation analogously.

\begin{defn}
\label{defn:fair-computation-dual-rail}
We say a CRC \emph{fairly dual-computes} a function $f:\R^k \to \R$ if it fairly computes a dual-rail representation of $f$.
\end{defn}

Since dual-rail computation is defined by a CRC that directly computes a dual-rail representation function,
\Cref{lem:fair-computation-implies-stable-computation,lem:stable-computation-implies-fair-computation-feedforward} imply the analogous results for dual-rail:

\begin{lem}\label{lem:dual-rail-fair-implies-stable}
If a CRC fairly dual-computes a function $f: \R^k \to \R$ then it stably dual-computes $f$. 
\end{lem}

\begin{proof}
The proof is similar to the proof of \Cref{lem:fair-computation-implies-stable-computation} except one uses \Cref{cor:multiple-output-not-stably-compute-implies-can-reach-far-from-correct} in the place of \Cref{lem:not-stably-compute-implies-can-reach-far-from-correct}.
\end{proof}

\begin{lem}\label{lem:dual-rail-feedforward-stable-implies-fair}
If a feedforward CRC stably dual-computes a function $f: \R^k \to \R$ then it fairly dual-computes $f$. 
\end{lem}

\begin{proof}
The proof is the same as \Cref{lem:stable-computation-implies-fair-computation-feedforward} since a static state is output stable regardless of how many output species there are. 
\end{proof}

\subsection{Statement of Main Results}
\label{sec:main-results}
Below we summarize our results about the computational power of stable computation, which we prove in subsequent subsections.
First, we formally define the relevant classes of functions that will correspond to direct and dual-rail stable computation:

\begin{defn}
\label{def:rational-linear}
A function $f: \R^k \to \R$ is \emph{rational linear} if there exist $a_1,\ldots,a_k \in \Q$ such that $f(\vx) = \sum_{i=1}^k a_i \vx(i).$
A function $f: \R^k \to \R$ is \emph{rational affine} if 
there exist $a_1,\ldots,a_k, c \in \Q$ such that $f(\vx) = \sum_{i=1}^k a_i \vx(i) + c$,
i.e., $f$ is a rational constant $c$ plus a rational linear function.
\end{defn}

We note that rational linearity has the equivalent characterization that $f$ is linear and maps rational inputs $\vec{x} \in \Q^n$ to rational outputs.

\begin{defn}\label{def:piecewise-rational-linear}
A function $f: \R^k \to \R$ is \emph{piecewise rational linear (affine)} if there is a finite set of partial rational linear (affine) functions $f_1,\ldots,f_p:\R^k \dashrightarrow \R$, with
$\bigcup_{j=1}^p \dom f_j = \R^k$, such that, for all $j \in \{1,\ldots,p\}$ and all $\vx \in \dom f_j$, $f(\vx) = f_j(\vx)$.
In this case, we say that $f_1,\ldots,f_p$ are the \emph{components} of $f$.
\end{defn}

\begin{defn}
\label{def:positive-continuous}
A function $f: \Rp^k \to \Rp$ is  \emph{positive-continuous} if, for all $U \subseteq \{1,\ldots,k\}$, $f$ is continuous on the domain
\[
D_U = \setl{\vx \in \Rp^k}{(\forall i \in \{1,\ldots,k\})\ \vx(i) > 0 \iff i \in U}.
\]
\end{defn}
In other words, $f$ is continuous on any subset $D \subset \Rp^k$ that does not have any coordinate $i \in \{1,\ldots,k\}$ that takes both zero and positive values in $D$.

The following theorems are the main results of this paper,
exactly characterizing the functions stably computable with direct and dual-rail representation of inputs and outputs,
and showing the equivalence between fair computation and stable computation.
Furthermore, although non-feedforward CRCs do exist to compute functions in this class,
the theorem shows that feedforward CRCs suffice to compute all such functions.

\begin{thm}
\label{thm:computable-characterization}
For a function $f: \Rp^k \to \Rp$, the following are equivalent:
\begin{enumerate}
    \item $f$ is fairly computable by a CRC.
    \item $f$ is stably computable by a CRC.
    \item $f$ is positive-continuous piecewise rational linear.
    \item $f$ is stably computable by a feedforward CRC.
    \item $f$ is fairly computable by a feedforward CRC.
\end{enumerate}
\end{thm}
\begin{proof}
(1) implies (2) is \Cref{lem:fair-computation-implies-stable-computation}.
(2) implies (3) is \Cref{lem:positive-continuous-piecewise-rational-linear}.
(3) implies (4) is \Cref{lem-nonnegative-piecewise-linear-implies-computable}.
(4) implies (5) is \Cref{lem:stable-computation-implies-fair-computation-feedforward}.
(5) implies (1) is obvious.
\end{proof}

The following is the dual-rail analog of \Cref{thm:computable-characterization}.

\begin{thm}
\label{thm:computable-characterization-dual-rail}
For a function $f: \R^k \to \R$, the following are equivalent:
\begin{enumerate}
    \item $f$ is fairly dual-computable by a CRC.
    \item $f$ is stably dual-computable by a CRC.
    \item $f$ is continuous piecewise rational linear.
    \item $f$ is stably dual-computable by a feedforward, output-oblivious CRC.
    \item $f$ is fairly dual-computable by a feedforward, output-oblivious CRC.
\end{enumerate}
\end{thm}

\begin{proof}
(1) implies (2) is \Cref{lem:dual-rail-fair-implies-stable}.
(2) implies (3) is \Cref{lem-computable-implies-piecewise-linear}.
(3) implies (4) is \Cref{lem-piecewise-linear-implies-computable}. 
(4) implies (5) is \Cref{lem:dual-rail-feedforward-stable-implies-fair}.
(5) implies (1) is obvious.
\end{proof}

\subsection{Positive Result: Continuous Piecewise Rational Linear Functions are Dual-Rail Computable}
\label{negative-piecewise-linear-implies-computable}

Definition~\ref{def:piecewise-rational-linear} does not stipulate how complex the ``boundaries'' between the linear pieces of a piecewise rational linear function can be. The boundaries can even be irrational in some sense, e.g., the function $f(x_1,x_2) = 0$ if $x_1 > \sqrt{2} \cdot x_2$ and $f(x_1,x_2) = x_1+x_2$ otherwise. However, if we additionally require that $f$ be \emph{continuous}, then the following theorem of Ovchinnikov~\cite[Theorem 2.1]{ovchinnikov2002max}
states that $f$ has a particularly clean form,
conducive to computation by CRCs.

\newcommand{\minMaxThmText}{
  Let $D \subseteq \R^k$ be convex.
  For every continuous piecewise affine function $f:D \to \R$ with components $g_1,\ldots,g_p$, there exists a family $S_1,\ldots,S_q \subseteq \{1,\ldots,p\}$ such that, for all $\vx \in D$,
  $
    f(\vx) = \max\limits_{i \in \{1,\ldots,q\}} \min\limits_{j \in S_i} g_j(\vx).
  $
}

\newtheorem*{minMaxThm}{\Cref{thm-min-max-representation}}

\begin{thm}[\cite{ovchinnikov2002max}, Theorem 2.1]
\label{thm-min-max-representation}
\minMaxThmText
\end{thm}

Note that as a special case, the above result applies when $f$ is continuous piecewise rational linear. 
The above theorem as stated slightly generalizes the result due to Ovchinnikov~\cite{ovchinnikov2002max} (by not requiring $D$ to be closed), although the proof technique is essentially the same.
For completeness, we provide the proof in \Cref{app:A}.

We use the theorem above to dual-compute continuous piecewise rational linear functions by composing CRC modules for rational linear functions, min, and max. 
These modules are developed in the following three lemmas.

\begin{lem}\label{lem:dual-rail-rational-linear-computable}
Rational linear functions are 
stably dual-computable by a feedforward, output-oblivious CRC.
\end{lem}

\begin{proof}
Let $g:\R^k \to \R$ be a rational linear function $g(\vx) = \sum_{i=1}^k a_i \vx(i)$.
  By clearing denominators, there exist $n_1,\ldots,n_k \in \Z$ and $d \in \Z^+$ such that
  $g(\vx) = \frac{1}{d} \sum_{i=1}^k n_i \vx(i).$
  The following reactions compute a dual-rail representation of $g$ with input species $X_1^+,\ldots,X_k^+, X_1^-,\ldots,X_k^-$ and output species $Y^+,Y^-$.
  For each $i$ such that $n_i > 0$, add the reactions
    \begin{eqnarray*}
    X_i^+ &\to& n_i W^+
    \\ X_i^- &\to& n_i W^-
    \end{eqnarray*}
  For each $i$ such that $n_i < 0$, add the reactions
    \begin{eqnarray*}
    X_i^+ &\to& |n_i| W^-
    \\ X_i^- &\to& |n_i| W^+
    \end{eqnarray*}
  To divide the values of $W^-$ and $W^+$ by $d$, add the reactions
    \begin{eqnarray*}
    d W^+ &\to& Y^+
    \\
    d W^- &\to& Y^-
    \end{eqnarray*}
    
In particular, these reactions compute the dual-rail representation $\hat{g}: \Rp^{2k} \to \Rp^2$ where 
\[\hat{g}(x_1^+,\ldots, x_k^+, x_1^-, \ldots x_k^-) = (y^+, y^-) = \left(\frac{1}{d}\left[\sum_{n_i > 0} n_ix_i^+ + \sum_{n_i < 0}|n_i|x_i^-\right], \frac{1}{d}\left[\sum_{n_i > 0} n_ix_i^- + \sum_{n_i < 0}|n_i|x_i^+\right]\right).\]
It is straightforward to verify that $\hat{g}$ really is a dual-rail representation of $g$. To see that the above CRC stably computes $\hat{g}$, define the functions $p, q: \Rp^\Lambda \to \R$ so that
\begin{align*}
p(\vc) &= \vc(Y^+) + \frac{1}{d}\vc(W^+) + \frac{1}{d}\sum_{n_i > 0}n_i\vc(X_i^+) + \frac{1}{d} \sum_{n_i < 0}|n_i|\vc(X_i^-) \\
q(\vc) &= \vc(Y^-) + \frac{1}{d}\vc(W^-) + \frac{1}{d}\sum_{n_i > 0}n_i\vc(X_i^-) + \frac{1}{d} \sum_{n_i < 0}|n_i|\vc(X_i^+)
\end{align*}

It is also straightforward to verify that both $p$ and $q$ are preserved by all of the reactions in the above CRC. This shows that for all $\vc$ and $\vd$ with $\vc \segto \vd$ we have $p(\vd) = p(\vc)$ and $q(\vd) = q(\vc)$. Now observe that from any state it is always possible to reach a state $\vo$ that only has positive concentrations of $Y^+$ and $Y^-$ by executing the reactions above to completion in the order in which they are listed. Such a state is evidently output stable so for any input state $\vx$ and any $\vc$ reachable from $\vx$, there is an output stable state $\vo$ reachable from $\vc$. Furthermore, since $\vo$ only has positive concentrations of $Y^+$ and $Y^-$ we know that $p(\vo) = \vo(Y^+)$ and $q(\vo) = \vo(Y^-)$. Since $\vo$ is reachable from $\vx$ we see that
\begin{align*}
\vo(Y^+) &= p(\vo) = p(\vx) = \frac{1}{d}\left[\sum_{n_i > 0} n_ix_i^+ + \sum_{n_i < 0}|n_i|x_i^-\right] \\
\vo(Y^-) &= q(\vo) = q(\vx) = \frac{1}{d}\left[\sum_{n_i > 0} n_ix_i^- + \sum_{n_i < 0}|n_i|x_i^+\right]
\end{align*}
This shows that a dual-rail representation of any rational linear function can be stably dual-computed by a CRC. Moreover, the CRC above is clearly output-oblivious, and it is feedforward under the ordering $X_1^+ < \ldots < X_n^+ < X_1^- < \ldots X_n^- < W^+ < W^- < Y^+ < Y^-$. 
\end{proof}

\begin{lem}\label{lem:dual-rail-min-computable}
Min is 
stably dual-computable by a feedforward, output-oblivious CRC.
\end{lem}

\begin{proof}
The following reactions stably compute a dual-rail representation of $\min$ with input species $X_1^+$, $X_2^+$, $X_1^-$, $X_2^-$ and output species $Y^+$, $Y^-$.
  \begin{eqnarray}
    X_1^+ + X_2^+ &\to& Y^+ \label{rxn-min-1}
    \\
    X_1^- &\to& X_2^+ + Y^- \label{rxn-min-2}
    \\
    X_2^- &\to& X_1^+ + Y^- \label{rxn-min-3}
  \end{eqnarray}

In particular this CRC computes the dual rail representation $\hat{f}: \Rp^4 \to \Rp^2$ where
\[\hat{f}(x_1^+, x_2^+, x_1^-, x_2^-) = (y^+, y^-) = (\min(x_1^+ + x_2^-, x_2^+ + x_1^-), x_1^- + x_2^-)\]
It is straightforward to verify that $\hat{f}$ is really a dual-rail representation of $\min$. To see that the above CRC stably computes $\hat{f}$ 
  \footnote{
    This analysis of the CRC for the min function here is directly based on the definition of stable computation.
    Recently a powerful framework has been developed~\cite{vasic2022programming},
    based on a wide class of so-called \emph{noncompetitive} CRCs in which no species consumed in a reaction is a reactant in another reaction (not even as a non-consumed catalyst).
    For such CRCs, the task of proving correctness of stable computation is greatly simplified.
    Since the CRC computing min is noncompetitive, that framework could be applied here to yield a simpler proof of correctness.
    We use our direct proof here for the sake of making the current paper self-contained.
  }, define the functions $p, q, \delta: \Rp^\Lambda \to \R$ so that
\begin{align*}
p(\vc) &= 2\vc(Y^+) + \vc(X_1^+) + \vc(X_2^+) + \vc(X_1^-) + \vc(X_2^-) \\
q(\vc) &= \vc(Y^-) + \vc(X_1^-) + \vc(X_2^-) \\
\delta(\vc) &= \vc(X_1^+) - \vc(X_1^-) - \vc(X_2^+) + \vc(X_2^-)
\end{align*}

It is also straightforward to verify that the above three functions are preserved by all of the reactions in the given CRC. Note that by running \ref{rxn-min-3}, then \ref{rxn-min-2}, then \ref{rxn-min-1} to completion it is always possible to reach a state $\vo$ with $\vo(X_1^-) = \vo(X_2^-) = 0$ and also $\vo(X_1^+) = 0$ or $\vo(X_2^+) = 0$. Such a state is evidently output stable, so for any input state $\vx$ and any $\vc$ reachable from $\vx$ there is an output stable state $\vo$ reachable from $\vc$. Since $\vo$ is reachable from $\vx$ we know that 
\[\vo(Y^-) = q(\vo) = q(\vx) = x_1^- + x_2^-.\]

Now suppose without loss of generality that $x_1^+ + x_2^- \le x_2^+ + x_1^-$ (the analysis of the other case is similar). Then $\delta(\vx) \le 0$, so $\vo(X_1^+) - \vo(X_2^+) = \delta(\vo) \le 0$. If $\vo(X_1^+)$ were positive, then by the definition of $\vo$ we would know that $\vo(X_2^+) = 0$, contradicting the fact that $\delta(\vo) \le 0$. Thus $\vo(X_1^+) = 0$ and
\[\vo(X_2^+) = -\delta(\vo) = -\delta(\vx) = x_1^- + x_2^+ - x_1^+ - x_2^-.\]
Finally, note that 
\[p(\vx) = x_1^+ + x_2^+ + x_1^- + x_2^- = p(\vo) = 2\vo(Y^+) + \vo(X_2^+) = 2\vo(Y^+) + x_1^- + x_2^+ - x_1^+ - x_2^-.\]
Solving the above equation for $\vo(Y^+)$ shows that $\vo(Y^+) = x_1^+ + x_2^- = \min(x_1^+ + x_2^-, x_2^+ + x_1^-)$. This shows that the above CRC stably computes $\hat{f}$. The CRN is clearly output-oblivious and it is feedforward with the ordering $X_1^- < X_2^- < X_1^+ < X_2^+ < Y^- < Y^+$. 
\end{proof}

\begin{cor}\label{cor:dual-rail-max-computable}
Max is 
stably dual-computable by a feedforward, output-oblivious CRC.
\end{cor}

\begin{proof}
To stably compute a dual-rail representation of max, observe that it is equivalent to computing the min function with the roles of the ``plus'' and ``minus'' species reversed (which negates the value represented in dual-rail), 
  because $\max(x_1,x_2) = - \min(-x_1,-x_2)$.
  In other words, use the reactions
  \begin{eqnarray*}
    X_1^- + X_2^- &\to& Y^-
    \\
    X_1^+ &\to& X_2^- + Y^+
    \\
    X_2^+ &\to& X_1^- + Y^+
  \end{eqnarray*}
\end{proof}

\begin{lem}
\label{lem-piecewise-linear-implies-computable}
  Let $D \subseteq \R^k$ be convex, and let $f:D \to \R$ be a continuous piecewise rational linear function.
  Then $f$ is stably dual-computed by a feedforward, output-oblivious CRC.
\end{lem}

\begin{proof}
By Theorem~\ref{thm-min-max-representation}, we know that any such continuous piecewise rational linear function can be represented as a composition of $\max$, $\min$, and rational linear functions. Moreover, the $\min$ and $\max$ functions with two arguments can be composed in a tree of depth $\log l$ to compute the minimum or maximum functions with input arity $l$. 
Since multiple rational linear functions may use the same inputs, 
we also need a fan-out module (likewise feedforward, output-oblivious) which copies a single input to multiple outputs: 
\begin{align*}
X_1^+ &\to Y_{1}^+ +\dots + Y_{p}^+ \\
X_1^- &\to Y_{1}^- +\dots + Y_{p}^-,
\end{align*}
where $p$ is the number of linear components of the piecewise linear function $f$.
Because output-oblivious CRCs are composable, and because the composition of feedforward, output-oblivious CRCs is again feedforward and output-oblivious, we can compose the CRCs from \Cref{lem:dual-rail-rational-linear-computable}, \Cref{lem:dual-rail-min-computable}, and \Cref{cor:dual-rail-max-computable} to produce a feedforward, output-oblivious CRC that computes $f$.
\end{proof}

Note that \Cref{lem-piecewise-linear-implies-computable} applies for arbitrary convex domains $D \subseteq \R^k$, which will be useful in proving \Cref{lem-nonnegative-piecewise-linear-implies-computable}, where we take $D$ to be a strict convex subset of $\Rp^k$ in which no coordinate takes both 0 and positive values in $D$.
Since $\R^k$ is itself convex, it also establishes the ``(3) implies (4)'' implication in the proof of \Cref{thm:computable-characterization-dual-rail}.

\subsection{Positive Result: Positive-Continuous Piecewise Rational Linear Functions are Directly Computable}\label{sec:positive-continuous-piecewise-rational-linear-computable}

The following lemma is the \emph{direct} stable computation analog of \Cref{lem-piecewise-linear-implies-computable}.
Intuitively, it is proven by using \Cref{lem-piecewise-linear-implies-computable} to stably dual-compute $2^k$ different continuous piecewise rational linear functions in parallel,
one for each possible choice of which input species $X_1,\ldots,X_k$ are 0.
A separate computation determines which inputs are positive and selects the appropriate output.
Note that positive inputs may be discovered ``piecemeal'', so the system must be robust to a continual updating of the decision.
However, a there is a monotonicity to this process that will make consistent updating possible:
once an input species is discovered to be present, we know for sure it is.
(Whereas a species that appears to be absent simply may have not yet reacted.)

\begin{lem}\label{lem-nonnegative-piecewise-linear-implies-computable}
  Every positive-continuous piecewise rational linear function $f:\Rp^k \to \Rp$ is stably computable by a feedforward CRC.
\end{lem}

\begin{proof}
  The CRC will have input species $X_1,\ldots,X_k$ and output species $Y^+$.
  (While it will be helpful to think of a $Y^+$ and $Y^-$ species, and during the computation the output will be encoded in their difference, the output of the CRC is only the $Y^+$ species as per direct computability.)

  Let $f:\Rp^k \to \Rp$ be a positive-continuous piecewise linear function.
  Since it is positive-continuous, there exist $2^k$ domains
  \[D_U = \{ \vx \in \Rp^k \ |\ (\forall i \in \{1,\ldots,k\})\ \vx(i) > 0 \iff i \in U \},\]
  one for each subset $U \subseteq \{1,\ldots,k\}$, such that $f \upharpoonright D_U$ is continuous.
  Define $f_U = f \upharpoonright D_U$. 
  Since $D_U$ is convex, by Lemma~\ref{lem-piecewise-linear-implies-computable} there is a CRC $\calC_U$ computing a dual-rail representation $\hf_U:\Rp^k\times\Rp^k\to\R\times\R$ of $f_U$.
  By letting the initial concentration of the ``minus'' version of the $i$'th input species $X_i^-$ be 0, we convert $\calC_U$ into a CRC that directly computes an output dual-rail representation of $f_U$.

  The intuition of the proof is as follows.
The case $U=\emptyset$ is trivial, as we will have no reactions of the form $\emptyset \to A$ for any species $A$, so if no species are initially present, no species (including $Y^+$) will ever be produced;
this is correct since any linear function $f$ obeys $f(\vec{0})=0$.
For each non-empty $U$, we compute  $f_U$ independently in parallel by CRC $\calC_{U}$, modifying each reaction producing $Y^+$ to produce an equivalent amount of species $Y_U$, which is specific to $U$.
For each such $U$ there are inactive and active ``indicator'' species $J_U$ and $I_U$.
  In parallel, there are reactions that will activate indicator species $I_U$ (i.e.~convert $J_U$ to $I_U$) if and only if all species $X_i$ are present initially for each $i \in U$.
  These $I_U$ species will then counteract the effect of any CRC computing $f_{U'}$ for $U' \subsetneq U$ by catalytically converting 
  all $Y_{U'}^+$ to $Y^-$ and all $Y_{U'}^-$ to $Y^+$.
  If $U$ is the complete set of indices of non-zero inputs, then only CRCs computing $f_{U'}$ for subsets $U' \subsetneq U$ have produced any amount of $Y^+$, so eventually all of these will be counteracted by $I_U$.

  Formally, construct the CRC as follows.
  Let $l = 2^{k-1}-1$.
  For each $i \in \{1,\ldots,k\}$, add the reaction
    $X_i \to I_{\{i\}} +
    J_{U_1} + X_i^{U_1} + J_{U_2} + X_i^{U_2} + \ldots + J_{U_{l}} + X_i^{U_{l}},$
  where $U_1,U_2,\ldots,U_{l}$ are all subsets of $\{1,\ldots,k\}$ that are strict supersets of $\{ i \}$.
  The extra versions $X_i^{U_1}, \ldots, X_i^{U_l}$ of $X_i$ are used as inputs to the parallel computation of each $f_U$.
  We generate the inactive indicator species from the input species in this manner, because the CRC is not allowed to start with anything other than the input.

  The indicator species are activated as follows. For each nonempty $U,U' \subseteq \{1,\ldots,k\}$ such that $U \neq U'$, add the reaction
  $
    I_U + I_{U'} + J_{U \cup U'} \to I_U + I_{U'} + I_{U \cup U'}.
  $

  For each nonempty $U \subseteq \{1,\ldots,k\}$, let $\calC_U$ be the CRC computing an output dual-rail representation of $f_U$ (i.e.~dual rail on the output). 
  Modify $\calC_U$ as follows.
  Rename the output species of $\calC_U$ to $Y^+$ and $Y^-$, i.e., all parallel CRCs share the same output species.
  For each reaction producing the output species $Y^+$, add the product $Y^+_U$ (which is a species specific to $\calC_U$) with the same net stoichimetry.
  Similarly, for each reaction producing the output species $Y^-$, add the product $Y^-_U$ with the same net stoichimetry.
  For instance, replace the reaction $A + B \to Y^+$ by the reaction $A + B \to Y^+ + Y^+_U$, and replace the reaction $A + Y^+ \to B + 4 Y^+$ by the reaction $A + Y^+ \to B + 4 Y^+ + 3 Y^+_U$.
  Therefore the eventual amount of $Y_U^+$ is equal to the total amount of $Y^+$ produced by $\calC_U$, and similarly for $Y_U^-$ and $Y^-$.
  For each $U' \subset U$, add the reactions
    $I_U + Y^+_{U'} \to I_U + Y^-$,
    $I_U + Y^-_{U'} \to I_U + Y^+.$
  Also, for each reaction in $\calC_U$, add $I_U$ as a catalyst.
  This ensures that $\calC_U$ cannot execute any reactions (and therefore cannot produce any amount of $Y^+$ or $Y^-$) unless all species $X_i$ for $i \in U$ are present.

We observe that the dual-rail CRC described above (with output species $Y^+$ and $Y^-$) is output-oblivious as it involves the output species only as products of reactions. 
Further, the output value is non-negative for any input.
Thus, we can convert the dual-rail representation to the direct one (\cref{def:direct-computation}) with a single output species $Y^+$, by adding the reaction $Y^+ + Y^- \to \emptyset$.

  To complete the proof, since the CRC of \Cref{lem-piecewise-linear-implies-computable} is feedforward, we can confirm by inspection our modifications preserve the feedforward property as well.
  In particular, one should order species within the CRN dual-computing $f_U$ before any species for supersets of $U$ and after any species for subsets of $U$.
\end{proof}

\subsection{Negative Result: Directly Computable Functions are Positive-Continuous Piecewise Rational Linear}

\label{subsec-negative-computable-implies-piecewise-linear}

\subsubsection{Siphons and output stability}

In order to characterize stable function computation for CRCs, we will crucially rely on the notion of siphons, which we recall from~\cref{sec:mass-action-reachability}, Definition~\ref{defn:siphon}.
Lemma~\ref{lem:siphon-stable} shows the underlying relationship between output stability and siphons.

Let $\calC=(\Lambda,R,\Sigma,\{Y\})$ be a CRC.
We call a siphon $\Omega$ 
\emph{stabilizing}
if, for any state $\vd$, $\vd  \upharpoonright \Omega = \vec{0}$ implies that $\vd$ is output stable.
In other words, ``draining'' $\Omega$ (removing all of its species) causes the output to stabilize.

\begin{lem}\label{lem:stable-state-has-stable-siphon}
  If $\vc$ is an output stable state, then $\Omega_\vc = \Lambda \setminus \producible(\vc)$ is a {stabilizing} siphon.
\end{lem}

\begin{proof}
    By Lemma~\ref{lem:not-producible-is-siphon},
    $\Omega_\vc$ is a siphon for any state $\vc$. 
    We show the contrapositive that if $\Omega_\vc$ is not 
    stabilizing, then $\vc$ is not output stable.
    Suppose there is some particular state $\vd$ with $\vd \upharpoonright \Omega_\vc = \vec{0}$ that is not output stable, i.e., for some $\ve$ such that $\vd \segto \ve$ we have $\vd(Y) \neq \ve(Y)$. Because of this there must be some reaction $\alpha$ applicable at a state $\vd'$ reachable from $\vd$ such that $\alpha$ changes the amount of $Y$ (in other words $\vM_{Y, \alpha} \neq 0$). 
    Since $\Omega_\vc$ is a siphon absent in $\vd$, 
    by Lemma~\ref{lem-seg-reachable-forward-invariant-siphon}, $\vd' \upharpoonright \Omega_{\vc} = \vec{0}$, so all reactants of $\alpha$ must be contained in $\producible(\vc)$.
    By Lemma~\ref{lem-all-species-present} we can find a state $\vc'$ such that $\vc \segto \vc'$ and $[\vc'] = \producible(\vc)$. Since $\alpha$ is applicable at $\vc'$ we see that $\vc$ is not output stable.
\end{proof}

The next lemma shows the key property of siphons that we will use to reason about stably computing CRC's:
that they \emph{characterize} the output stable states,
i.e., the \emph{only} way for the output to stabilize is to drain some stabilizing siphon.

\begin{lem}  \label{lem:siphon-stable}
  There is a set of stabilizing siphons $\calS$ such that a state $\vc$ is output stable if and only if $\exists \Omega \in \calS$ such that\ $\vc \upharpoonright \Omega = \vec{0}$.
\end{lem}

\begin{proof}
Take $\calS$ to be $\set{\Omega_\vc \subseteq \Lambda \ |\ \vc \text{ is output stable}}$, 
where $\Omega_\vc = \Lambda \setminus \producible(\vc)$.
By Lemma~\ref{lem:stable-state-has-stable-siphon}, each $\Omega \in \calS$ is a stabilizing siphon, so if $\exists \Omega \in \calS$ with $\vd \upharpoonright \Omega = \vec{0}$ then $\vd$ is output stable. 
On the other hand, if $\vd$ is some output stable state, 
then by the definition of $\Omega_\vd$
(the set of species that cannot be produced from $\vd$),
we have $\vd \upharpoonright \Omega_{\vd} = \vec{0}$.
By our construction of $\calS$ we know that $\Omega_{\vd} \in \calS$. 
\end{proof}

\subsubsection{Linearity restricted to inputs draining a siphon}

This section aims to prove that the function computed by $\calC$, when restricted to inputs that can drain a particular stabilizing siphon, is linear.
This is the first step establishing the ``piecewise linear'' portion of the ``positive-continuous piecewise rational linear'' claims in 
 \Cref{thm:computable-characterization}.

Recall that $\prepathsset$, defined in Definition~\ref{def:prepaths}, 
is the space of all prepaths
(i.e., the vector space in which paths live),
and $\Gamma_\infty$, defined in Definition~\ref{def:paths}, is the space of all paths (allowing infinitely-many line segments). 
We now define a map $\vec{o}$ which intuitively sends a path to the final state that it reaches.

\begin{lem}
\label{lem:o}
The map $\vec{o}: \prepathsset \to \R^\Lambda$ sending
\[\vgamma \mapsto \lim_{n \to \infty} \vec{x}_n(\vgamma)\]
is linear 
(recall $\vx_n(\vgamma)$, defined in Definition~\ref{def:prepaths}, is the state reached after traversing $n$ segments along $\vgamma$).
\end{lem}

\begin{proof}
To check that $\vec{o}$ is a linear function, note that for any $\vgamma_0, \vgamma_1 \in \prepathsset$ and $\lambda \in \R$ we have
\begin{align*}
    \vec{o}(\vgamma_0 + \lambda \vgamma_1) &= \lim_{n \to \infty} \vec{x}_n( \vgamma_0 + \lambda \vgamma_1) \\
    &= \lim_{n \to \infty} \vec{x}_n(\vgamma_0) + \lambda\vec{x}_n(\vgamma_1) \\
    &= \lim_{n \to \infty} \vec{x}_n(\vgamma_0) + \lambda\lim_{n \to \infty}\vec{x}_n(\vgamma_1) \\
    &= \vec{o}(\vgamma_0) + \lambda \vec{o}(\vgamma_1). \qedhere
\end{align*}
\end{proof}

\begin{defn}
\label{defn:paths-draining-siphon}
Let $\Omega$ be a stabilizing siphon. Define $\Gamma(\Omega)$ to consist of paths $\vgamma \in \Gamma_\infty$ such that $\vec{o}(\vgamma)\upharpoonright \Omega = \vec{0}$.
\end{defn}

These are the paths that converge to a state where a given stabilizing siphon is drained. 

\begin{lem}
$\Gamma(\Omega)$ is convex for each stabilizing siphon $\Omega$.
\end{lem}

\begin{proof}
Suppose that $\vgamma_0$ and $\vgamma_1$ are in $\Gamma(\Omega)$ and let $\vgamma_\lambda = (1-\lambda)\vgamma_0 + \lambda \vgamma_1$. By Lemma~\ref{lem:Pathspace-convex} we know that $\vgamma_\lambda$ is in $\Gamma_\infty$. Moreover 
\[\vec{o}(\vgamma_\lambda) = (1 - \lambda)\vec{o}(\vgamma_0) + \lambda \vec{o}(\vgamma_1).\]
Since $\Omega$ is drained at both $\vec{o}(\vgamma_0)$ and $\vec{o}(\vgamma_1)$, we conclude that it must also be drained at $\vec{o}(\vgamma_\lambda)$, so $\vgamma_\lambda \in \Gamma(\Omega)$. 
\end{proof}

\begin{defn}
Let 
\[\Sigma(\Omega) = \{ \left. \vx \in \Rp^\Lambda \ \right|\ [\vx] \subseteq \Sigma,\ (\exists \vo)\ \vx \segto \vo, \text{ and } \vo \upharpoonright \Omega = \vec{0} \}\]
denote those input states from which the siphon $\Omega$ is drainable.
\end{defn}

\begin{lem}\label{lem-siphon-implies-linear-nondual}
  Let $f:\Rp^k \to \Rp$ be stably computed by a CRC $\calC=(\Lambda,R,\Sigma,\{Y\})$.
  Let $\Omega$ be a stabilizing siphon.
  Then $f$ restricted to $\Sigma(\Omega)$ is a linear function.
\end{lem}

\begin{proof}
Recall $\Gamma(\Omega)$ from \Cref{defn:paths-draining-siphon}, the map $\vx_0$ from \Cref{def:prepaths}, and the map $\vo$ from \Cref{lem:o}.
First project 
$\Gamma(\Omega)$ to $\Rp^\Lambda \times \Rp^\Lambda$ by the map 
$\vgamma \mapsto (\vec{x}_0(\vgamma), \vec{o}(\vgamma))$. Let $G \subseteq \R^{k+1}$ be the further projection to the $(k+1)$-dimensional subspace corresponding only to the input species $X_1,\ldots,X_k$ and output species $Y$.
$G$ is the graph of the function $y = f(\vx)$ restricted to inputs $\vx \in \Sigma(\Omega)$. Since $G$ is the image of a convex set under a linear transformation, it is also convex. We claim that $G$ must be a subset of a $k$-dimensional hyperplane. 

For the sake of contradiction, suppose not. Then there are $k+1$ non-coplanar points in $G$. Since $G$ is convex, it contains the entire $(k+1)$-dimensional convex hull $H$ of these points. Since $H$ is a $(k+1)$-dimensional convex polytope, it contains two different values of $y$ corresponding to the same value of $\vx$, contradicting the fact that only a single $y$ value exists in all output-stable states reachable from $\vx$. This establishes the claim that $G$ must be a subset of a $k$-dimensional hyperplane.
  
Since the graph of $f$ is a subset of a $k$-dimensional hyperplane, $f$ is an affine function. Since there are no reactions of the form $\emptyset \to \ldots$, $Y$ cannot be produced from the initial state $\vx=\vec{0}$ (nor can any other species), so $f(\vec{0})=0$. Therefore this hyperplane passes through the origin, so it defines a linear function.
\end{proof}

In Lemma~\ref{lem-siphon-implies-linear-nondual}, the reason that we restrict attention to a single output siphon $\Omega$ is that if different output siphons are drained, then different linear functions may be computed by the CRC. 
For example, $X_1+X_2\to Y$ computes $f(x_1,x_2) = x_1$ if siphon $\{X_2\}$ is drained and $f(x_1,x_2) = x_2$ if siphon $\{X_1\}$ is drained.
If a CRC stably computes then an output stable state is reachable from any input state.
Thus by \Cref{lem:siphon-stable}, from every input state some stabilizing siphon is drainable, and the following corollary is immediate:
\begin{corollary}
    \label{cor:picewise-linear}
  Let $f:\Rp^k \to \Rp$ be stably computed by a CRC.
  Then $f$ is piecewise linear.    
\end{corollary}

\subsubsection{Positive-Continuity}

Ideally, in order to prove that the function stably computed by a CRC is positive-continuous,
we would like to prove the following:
for any {stabilizing} siphon $\Omega$, the set $\Sigma(\Omega)$ of input states that can drain $\Omega$ is closed relative to the positive orthant.
If that were true, then we could use a fundamental topological result that if a function is piecewise continuous with finitely many pieces 
(e.g., piecewise linear),
and if the domain defining each piece is closed
(with agreement between pieces on intersecting domains),
then the whole function is continuous.
However, the above statement is not true in general. 
Consider the following counterexample:
\begin{align*}
    X_1 &\to C
    \\
    X_1 + X_2 + C &\to C + Y
\end{align*}
If initially $\vi(X_1) > \vi(X_2)$, then the {stabilizing} siphon $\{X_2\}$ is drainable, by producing $(\vi(X_1) - \vi(X_2)) / 2$ of $C$ via the first reaction (leaving an excess of $X_1$ over $X_2$ still),
then running the second reaction until $X_2$ is gone to produce $Y$.
Because $X_2$ can only be consumed if $C$ is produced,
which requires consuming a positive amount of $X_1$,
the set of inputs from which $\{X_2\}$ can be drained is the non-closed set $\{ \vi \mid \vi(X_1) > \vi(X_2)\}$.
Note that the above CRC does not stably compute anything because,
starting from a state with $\vi(X_1) > \vi(X_2)$, 
the first reaction could run until $X_2$ exceeds $X_1$ before starting the second reaction,
which would imply that the amount of $Y$ produced depends on how much reaction 1 happens.
It is still unclear whether such counterexamples exist for CRCs stably computing some function that is not identically $0$.

Instead of relying on $\Sigma(\Omega)$ being closed, we must make a more careful argument.
In lieu of working directly with the sets $\Sigma(\Omega)$ we consider ``shifted'' sets $\Tilde{\Sigma}_{(\vy, \vz)}(\Omega)$ (see Definition~\ref{defn:shifted_set} below). Each $\Tilde{\Sigma}_{(\vy, \vz)}(\Omega)$ is (possibly strictly) contained in the original $\Sigma(\Omega)$, but they still cover the set of inputs. Crucially, we are able to show that the shifted sets $\Tilde{\Sigma}_{(\vy, \vz)}(\Omega)$ are closed, allowing us to apply the argument at the start of this section to prove that every function stably computed by a CRC is positive-continuous.

\begin{defn}
Let 
\[X(\Omega) = \{ \left. \vx \in \Rp^\Lambda\ \right|\ (\exists \vo)\ \vx \to^1 \vo \text{ and } \vo \upharpoonright \Omega = \vec{0} \}\]
denote those states from which siphon $\Omega$ is drainable via a single straight line segment.
\end{defn}

\begin{lem}\label{lem-siphon-inputs-closed}
  Let $\Omega$ be a siphon.
  Let $\va_1,\va_2,\ldots \in X(\Omega)$ be a convergent sequence of states, where $\va=\lim_{i\to\infty} \va_i$.
  Suppose $[\va] = \Lambda$.
  Then $\va \in X(\Omega).$
\end{lem}

\begin{proof}
Consider the set $P$ of $\vgamma = (\vec{x}_0, \vec{u}) \in \R^\Lambda_{\ge 0} \times \R^R_{\ge 0}$ such that $\vec{o}(\vgamma) \in \R^\Lambda_{\ge 0}$ and $\vec{o}(\vgamma) \upharpoonright \Omega = \vec{0}$. 
Note that reactions occurring with positive flux in $\vec{u}$ might not be applicable at $\vec{x}_0$.
$P$ is cut out by a system of non-strict linear inequalities (in other words, it is a polyhedron). By \cite{ziegler1995polytopes}, $\vec{x}_0(P)$ is also a polyhedron, and is in particular closed. 

Note that $X(\Omega) \subseteq \vec{x}_0(P)$ and $\vec{x}_0(P) \cap \R_{> 0}^\Lambda = X(\Omega) \cap \R_{> 0}^\Lambda$. The first relation follows since if $\vec{x} \in X(\Omega)$, then the straight-line path draining $\Omega$ produces a $\vgamma \in P$ such that $\vec{x}_0(\vgamma) = \vec{x}$. The second relation holds since if every species is present in $\vec{x}_0$ then every reaction is applicable at $\vec{x}_0$, so any of the points $\vgamma \in P$ that project to $\vec{x}_0$ are valid paths that drain $\Omega$.

Since $\va = \lim_{i\to\infty} \va_i$, we have that for all but finitely many $i$, $\va_i \in X(\Omega) \cap \R_{> 0}^\Lambda$. As a result, these $\va_i$ are in $\vec{x}_0(P)$, so $\va$ in in $\vec{x}_0(P)$, too, since the set is closed.
Since $\va \in \vec{x}_0(P) \cap \R_{>0}^\Lambda$, we conclude that $\va \in X(\Omega).$
\end{proof}

Note that the hypothesis $[\va]=\Lambda$ is necessary.
Otherwise, consider the reactions $A\to C$, $A+B\to\emptyset$, and $F+C\to C$, with $\va_i(C)=0$, $\va_i(F)=1$, $\va_i(B)=1$, and $\va_i(A)$ approaching 1 from above as $i \to \infty$ (whence $C \not \in [\va]$).
Then the siphon $\Omega=\{A,B,F\}$ is drainable from each $\va_i$ by running $A \to C$ until $A$ and $B$ have the same concentration, then running the other two reactions to completion.
However, $\va(A)=\va(B)$, so running any amount of reaction $A \to C$ prevents reaction $A+B\to\emptyset$ from draining $B$.
Therefore $\va \not\in X(\Omega)$ but $\va_i \in X(\Omega)$ for all $i$.

\begin{defn}
A pair $(\vy, \vz) \in \Rp^\Lambda \times \Rp^\Lambda$ of states is a 
\emph{full input pair}
if $[\vy] = \Sigma$, $[\vz] = \Lambda$, and $\vy \segto \vz$. 
\end{defn}

\begin{defn}\label{defn:shifted_set}
If $(\vy, \vz)$ is a 
full input pair and $\Omega$ is an stabilizing siphon, define
\begin{align*}
\Tilde{\Sigma}_{(\vy, \vz)}(\Omega) = \{\vx \in \Rp^\Lambda \ |\ [\vx] = \Sigma \text{ and } \forall \lambda > 0 \text{ such that } \lambda\vy < \vx, \\
\text{it is the case that }\vx - \lambda\vy + \lambda\vz \in X(\Omega)\}    
\end{align*}
\end{defn}

Intuitively, in the CRC at the beginning of the section, the obstacle to the set $\Sigma(\Omega)$ being closed for the siphon $\Omega = \set{X_2}$ was that not all species were present initially: to drain $\Omega$, you first need to produce some (arbitrarily small) amount of $C$, which leads to the requirement $\vi(X_1) > \vi(X_2)$. To fix this problem, $\tilde{\Sigma}_{(\vy, \vz)}(\Omega)$ considers only the states that can drain $\Omega$ after having been ``perturbed" into a state where all species are present, where the full input pair $(\vy, \vz)$ specifies how to perform this perturbation. 

To see how this works in the example CRC, let $(\vy, \vz)$ be the full input pair where $\vy = \set{1X_1, 1X_2, 0C, 0Y}$ and $\vz = \set{.25X_1, .75X_2, .5C, .25Y}$. Then for an input state $\vx = \set{aX_1, bX_2, 0C, 0Y}$, $\lambda\vy < \vx$ when $\lambda < \min(a,b)$. As a result, $\vx \in \tilde{\Sigma}_{(\vy, \vz)}(\Omega)$ if $\vx - \lambda\vy + \lambda\vz = \set{(a - .75\lambda)X_1, (b - .25\lambda)X_2, .5\lambda C, .25\lambda Y}$ can drain $X_2$ for all $0 < \lambda < \min(a,b)$. This happens when $b - .25\lambda \le a - .75\lambda$, so $b + .5\lambda \le a$. Taking the limit as $\lambda \to 0$ we see that $b \le a$, so $\min(a,b) = b$, and taking the limit as $\lambda \to \min(a,b) = b$, we see that $1.5 b \le a$. Since $1.5b \le a$ is also a sufficient condition for $\vx$ to be in $\tilde{\Sigma}_{(\vy, \vz)}(\Omega)$ we see that $\tilde{\Sigma}_{(\vy, \vz)}(\Omega)$ is closed and contained in $\Sigma(\Omega)$. The next two lemmas show that these properties hold in general.

\begin{lem}\label{lem-shifted-closed}
For any full input pair $(\vy, \vz)$ and any {stabilizing} siphon $\Omega$, $\Tilde{\Sigma}_{(\vy, \vz)}(\Omega)$ is closed relative to $\R_{>0}^\Sigma$. 
\end{lem}

\begin{proof}
Let $\vx$ be a state such that $[\vx] = \Sigma$ and let $\{\vx_i\}$ be a sequence such that $\vx = \lim_{i \to \infty} \vx_i$ and $\vx_i \in \Tilde{\Sigma}_{(\vy, \vz)}(\Omega)$. For any $\lambda > 0$ such that $\lambda\vy < \vx$, by throwing out finitely many terms in the sequence $\{\vx_i\}$, we can guarantee that $\lambda \vy < \vx_i$ for all $i$, too. Since $\vx_i \in \Tilde{\Sigma}_{(\vy, \vz)}(\Omega)$, we know that $\vx_i - \lambda \vy + \lambda\vz \in X(\Omega)$ for all $i$. Since $X(\Omega)$ is closed, we see that 
\[\vx -\lambda \vy + \lambda\vz = \left(\lim_{i \to \infty} \vx_i\right) -\lambda \vy + \lambda\vz = \lim_{i \to \infty} (\vx_i -\lambda \vy + \lambda\vz)\]
is also in $X(\Omega)$. Since this is true for every $\lambda$ such that $\lambda\vy < \vx$, we conclude that $\vx \in \Tilde{\Sigma}_{\vy, \vz}(\Omega)$. Since this is true for any $\vx$ in $\R_{>0}^\Sigma$, we conclude that $\Tilde{\Sigma}_{\vy, \vz}(\Omega)$ is closed relative to $\R_{>0}^\Sigma$.
\end{proof}

The following lemma is almost immediate from the definition. The only possible concern one might have is that an input state $\vx$ is contained in $\Tilde{\Sigma}_{(\vy, \vz)}(\Omega)$ ``vacuously"---in other words, that there simply does not exist a $\lambda > 0$ such that $\lambda\vy < \vx$.  

\begin{lem}
\label{lem-shifted-in-unshifted}
For any full input pair $(\vy, \vz)$ and any {stabilizing} siphon $\Omega$, $\Tilde{\Sigma}_{(\vy, \vz)}(\Omega) \subseteq \Sigma(\Omega)$.
\end{lem}

\begin{proof}
Let $\vx$ be in $\Tilde{\Sigma}_{(\vy, \vz)}$. By definition, $[\vy] = \Sigma$, so the following is a well-defined real number
\[\lambda_0 = \min_{S \in \Sigma}
\left\{\frac{\vx(S)}{\vy(S)}\right\}.\]
In other words, $\lambda_0$ is the number so that $\lambda\vy < \vx$ if and only if $\lambda < \lambda_0$. 
Because $[\vx] = \Sigma$ we know that $\lambda_0 > 0$. 
Also $(\lambda_0/2) \vy \segto (\lambda_0/2) \vz$ because $(\vy, \vz)$ is a full input pair.
Then by additivity of $\segto$,
$\vx \segto \vx - (\lambda_0/2) \vy + (\lambda_0/2) \vz \to^1 \vo$, where $\Omega$ is drained at $\vo$. 
Because $\segto$ is transitive (Corollary~\ref{cor:segto-is-transitive}) we conclude that $\vx \segto \vo$, so $\vx \in \Sigma(\Omega)$. 
\end{proof}

The above two lemmas show that $\tilde{\Sigma}_{(\vy, \vz)}(\Omega)$ is topologically better behaved than $\Sigma(\Omega)$, although possibly smaller. However, in order for these sets to be useful for analyzing the behavior of a CRC we need to show that they are not ``too small''---in particular the next lemma shows that under the assumption that a CRC stably computes a function then for any fixed choice of a full input pair $(\vy, \vz)$, the sets $\tilde{\Sigma}_{(\vy, \vz)}(\Omega)$ are still big enough to cover all of $\R_{> 0}^\Sigma$.

\begin{lem}
\label{lem-shifted-sets-cover}
If $\calC$ is a stably computing CRC then for any fixed full input pair $(\vy, \vz)$, as $\Omega$ varies among all of the {stabilizing} siphons, the sets $\Tilde{\Sigma}_{(\vy, \vz)}(\Omega)$ cover $\R_{>0}^\Sigma$. 
\end{lem}

\begin{proof}
Let $\vx$ be an input state such that $[\vx] = \Sigma$ and let $\lambda_0$ be as in the proof of \Cref{lem-shifted-in-unshifted}. We know that 
$\lambda_0\vy \segto \lambda_0 \vz$ because $(\vy, \vz)$ is a full input pair. 
By additivity of $\segto$,
\[\vx = (\vx - \lambda_0 \vy) + \lambda_0 \vy \segto (\vx - \lambda_0 \vy) + \lambda_0 \vz \]
so since $\calC$ is a stably-computing CRC there must be 
some output-stable state $\vo$ such that $\vx - \lambda_0 \vy + \lambda_0 \vz \segto \vo$. By Lemma~\ref{lem:siphon-stable} there is an {stabilizing} siphon $\Omega$ so that $\Omega$ is drained at $\vo$. 

For any $\lambda > 0$ such that $\lambda\vy < \vx$, we know that $\lambda < \lambda_0$, so 
\begin{align*}
&\vx - \lambda\vy + \lambda\vz = (\vx - \lambda_0\vy + \lambda\vz) + (\lambda_0 - \lambda)\vy \\
\segto &(\vx - \lambda_0\vy + \lambda\vz) + (\lambda_0 - \lambda)\vz = \vx - \lambda_0 \vy + \lambda_0 \vz
\end{align*}
Since $\vx - \lambda_0 \vy + \lambda_0 \vz \segto \vo$ we see that $\vx - \lambda \vy + \lambda \vz \segto \vo$, and since $[\vz] = \Lambda$, by Lemma~\ref{lem-all-producible-present} we conclude that $\vx - \lambda\vy + \lambda\vz \to^1 \vo$. Since this is true for any $\lambda > 0$ such that $\lambda\vy < \vx$, we conclude that $\vx \in \Tilde{\Sigma}_{(\vy, \vz)}(\Omega)$. This shows that every $\vx$ with $[\vx] = \Sigma$ is in $\Tilde{\Sigma}_{(\vy, \vz)}(\Omega)$ for some $\Omega$, as desired. 
\end{proof}

We now use the above technical machinery to prove the following result, which is \emph{almost} the full negative result for direct computation, but leaves out the constraint that $f$ is \emph{rational} linear.
Rationality is shown in \Cref{subsec:rationality} below. 
Recall the 
\emph{positive-continuous} functions  
from \Cref{def:positive-continuous}.

\begin{lem}\label{lem-nonnegative-computable-implies-positive-continuous}
  Let $f:\Rp^k \to \Rp$ be stably computed by a CRC.
  Then $f$ is positive-continuous and piecewise linear.
\end{lem}

\begin{proof}
Piecewise linearity follows from \Cref{cor:picewise-linear}.
For positive continuity we proceed as follows.
Let $U \subseteq \{1,\ldots,k\}$, let $\vx \in D_U$
(where $D_U$ is as defined in \Cref{def:positive-continuous}), and let $\vx_1$, $\vx_2$, $\ldots \in D_U$ be an infinite sequence of points such that $\lim_{i\to\infty} \vx_i = \vx$. It suffices to show that $\lim_{i\to\infty} f(\vx_i) = f(\vx)$ --- i.e.~that $f$ is continuous on $D_U$. We take $\vx_i$ and $\vx$ equivalently to represent an initial state of the CRC giving the concentrations of species in $\Sigma=\{X_1,\ldots,X_k\}$.

In analyzing the behavior of the CRC on states in $D_U$, it will help us to consider the functionally equivalent CRC in which we remove species that are not producible from any state in $D_U$. For the purposes of this proof we consider this reduced CRC, and let $\Lambda$ be the corresponding reduced set of species.

Let $(\vy, \vz)$ be some full input pair. Then as $\Omega$ varies among the {stabilizing} siphons, $\Tilde{\Sigma}_{(\vy, \vz)}(\Omega)$ gives a finite collection of closed sets covering $\R_{>0}^\Sigma$ by Lemmas~\ref{lem-shifted-sets-cover} and~\ref{lem-shifted-closed}. Since $\Tilde{\Sigma}_{(\vy, \vz)}(\Omega) \subseteq \Sigma(\Omega)$ by Lemma~\ref{lem-shifted-in-unshifted} and since $f$ is linear (and therefore continuous) on $\Sigma(\Omega)$ by Lemma~\ref{lem-siphon-implies-linear-nondual}, we see that $f$ is continuous on each of the closed sets in this covering. By \cite{munkres2000topology}, if a topological space is a union of finitely many closed sets and $f_i$ are continuous function on each closed set that agree on overlaps, then they combine to give a continuous function. From this result, we conclude that $f$ is continuous on $D_U$, as desired. 
\end{proof}

\subsubsection{Rationality}
\label{subsec:rationality}

Recall Definition~\ref{def:rational-linear} defining rational linear functions and
Definition~\ref{def:piecewise-rational-linear} defining piecewise rational linear functions.

The main ideas of this section are as follows: to show that a function is piecewise rational linear, we need to show that it is rational linear on some finite set of domains that cover the input space.

A linear function $f: \R^n \to \R$ which sends $\Q^n$ to $\Q$ is necessarily rational linear. 
Since a linear function on $\R^n$ is completely determined by its behavior on any open ball, we can check this condition ``locally" on any domain that contains an open ball.
(Since $f$ is continuous and all of the points of domains that don't contain an open ball are limit points of the other domains,
we can ignore domains that don't contain open balls.)
The fact that the function sends $\Q^n$ to $\Q$ on such a domain is ultimately a consequence of the fact that the stoichiometry matrix of a CRN has only integer coefficients, so it preserves rationality.

Recall that $\prepathsset$, defined in Definition~\ref{def:prepaths}, is the space of all prepaths, and $\Gamma_\infty$, defined in Definition~\ref{def:paths}, is the space of all paths.

\begin{defn}
A path $\vgamma \in \Gamma_\infty$ is a \emph{rational path} if it has rational initial concentrations and all of its segments have rational fluxes. In other words, $\vx_0(\vgamma) \in \Q^\Lambda$ and $\vu_i(\vgamma) \in \Q^R$ for all $i$. 
\end{defn}

Note that since the stoichiometry matrix is an integer-valued matrix, it is automatically the case that $\vo(\vgamma)$ and every $\vx_i(\vgamma)$ is in $\Q^\Lambda$ for any rational path $\vgamma$.

\begin{defn}
We say that two prepaths $\vgamma, \vgamma' \in \prepathsset$ \emph{have the same sign} if for all species $S$, reactions $\alpha$, and $i \in \N$, it is the case that $\sgn \vx_i(\vgamma)_S = \sgn \vx_i(\vgamma')_S$ and $\sgn \vu_i(\vgamma)_\alpha = \sgn \vu_i(\vgamma')_\alpha$.  
\end{defn}

\begin{lem}\label{lem-rational-path}
Let $\vgamma \in \Gamma_\infty$ be a finite piecewise linear path. Then for any $\varepsilon > 0$, there is a rational path $\vgamma' \in \Gamma_\infty$ such that $\vgamma'$ has the same sign as $\vgamma$ and $||\vgamma' - \vgamma|| < \varepsilon$. If $\vgamma$ already has rational initial concentrations, then $\vgamma'$ can be chosen with the same initial concentrations.
\end{lem}

\begin{proof}
First, let $N$ be the largest natural number such that $\vu_N(\vgamma) \neq 0$.
(Such an $N$ exists since $\vgamma$ is finite.)
For any reaction $\alpha \in R$ and any $n \in \N_{>0}$ such that $\vec{u}_n(\vgamma)_\alpha = 0$, set $\vec{u}_n(\vgamma)_\alpha = 0$. Now for each species $S \in \Lambda$ such that $\vx_k(\vgamma)_S = 0$ and $k \le N$, consider the following linear equation:
\[\vx_0(\vgamma')_S + \sum_{\substack{1 \le n \le k \\ \alpha \in R \\ \vu_n(\vgamma)_\alpha \neq 0}} M_{S\alpha} \vu_n(\vgamma')_\alpha = 0.\]
Aggregating these equations for all $0 \le k \le N$ and all $S \in \Lambda$ such that $\vx_k(\vgamma)_S = 0$ gives a system of equations, linear in $\vx_0(\vgamma')_S$ and $\vu_n(\vgamma')_\alpha$, with rational coefficients. This equation has a real-valued solution, namely $\vx_0(\vgamma') = \vx_0(\vgamma)$ and $\vu_n(\vgamma') = \vu(\vgamma)$, so by Lemma~\ref{lem-rational-sols} (proven in \Cref{app:C}), it must have a solution with rational coefficients that is $\delta$-close for any $\delta > 0$. By taking $\delta$ small enough, we can of course make $\delta < \varepsilon$, but we can also guarantee that $\vx_0(\vgamma')_S$ is positive whenever $\vx_0(\vgamma)_S$ is positive and similarly for $\vu_n(\vgamma')_\alpha$. We have therefore specified a $\vgamma' \in \prepathsset$ with the same sign as $\vgamma$. Since $\vgamma'$ has the same sign as $\vgamma$, and since $\vgamma$ is a valid path, we conclude that $\vgamma'$ is also a valid path, so $\vgamma' \in \Gamma_\infty$. 

If $\vx_0(\vgamma)$ is already in $\Q^\Lambda$, then the same argument applies, with the modification that you fix $\vx_0(\vgamma') = \vx_0(\vgamma)$, and instead solve the inhomogeneous system of equations 
\[\sum_{\substack{1 \le n \le k \\ \alpha \in R \\ \vu_n(\vgamma)_\alpha \neq 0}} M_{S\alpha} \vu_n(\vgamma')_\alpha = -\vx_0(\vgamma)_S\]
when $\vx_k(\vgamma)_S = 0$.
\end{proof}

\begin{lem}\label{lem:open-ball-implies-rational}
Let $\Omega$ be 
a stabilizing siphon. If $\Sigma(\Omega)$ contains an open ball, then $f$ is rational linear when restricted to inputs in $\Sigma(\Omega)$.
\end{lem}

\begin{proof}
Let $B$ be the open ball contained in $\Sigma(\Omega)$ and let $\vx$ be in $\Q^\Lambda \cap B$. We know that there is a piecewise linear path $\vgamma$ starting at $\vx$ such that $f(\vx) = \vo(\vgamma)_Y$. 
By Theorem~\ref{thm-reachable-segment-bound} we may assume without loss of generality that $\vgamma$ is finite.
By Lemma~\ref{lem-rational-path}, there is a rational path $\vgamma'$ with the same sign as $\vgamma$ such that $\vx_0(\vgamma') = \vx_0(\vgamma) = \vx$. Because $\vo(\vgamma)$ is an output-stable state, some siphon $\Omega$ is drained at $\vo(\vgamma)$. Since $\vgamma'$ has the same sign as $\vgamma$, we know that $\Omega$ is also drained at $\vo(\vgamma')$, so $\vo(\vgamma')$ is also output stable. We must then have that $f(\vx) = \vo(\vgamma')_Y$, but by the construction of $\vgamma'$ we know that $\vo(\vgamma')_Y \in \Q$. Since $B \subseteq \Sigma(\Omega)$, we know that $f|_B$ is linear by Lemma~\ref{lem-siphon-implies-linear-nondual}. Since $B$ is an open ball we know that $\Q^\Lambda \cap B$ contains a basis for $\Q^\Lambda$, so $f|_B$ is a linear function that maps $\Q^\Lambda$ to $\Q$. Since every linear function $\R^\Lambda \to \R$ that sends $\Q^\Lambda$ to $\Q$ is rational linear, we are done.
\end{proof}

Recall that a \emph{closed domain} is the closure of an open set.

\begin{lem}
\label{lem:replace-domains-of-defn-with-closed-domains}
Let $f: X \to Y$ be a continuous function defined piecewise on closed sets, so $X$ is covered by finitely many closed sets $D_1 \ldots D_k$ and there are continuous functions $g_i: X\to Y$ such that $f|_{D_i} = g_i|_{D_i}$. Then there are (possibly empty) closed domains $E_1 \ldots E_k$ that cover $X$ such that $f|_{E_i} = g_i|_{E_i}$. 
\end{lem}

\begin{proof}
We show how to convert each $D_i$ that is not a closed domain to a corresponding $E_i$ that is.
Let $D_i$ be some set that isn't a closed domain. Let 
\[D_+ = \bigcup_{j\neq i} D_j.\]
Note that $D_1, \ldots D_{i - 1}, X \setminus D_+, D_{i + 1}, \ldots D_k$ cover $X$.
Let $E_i$ be the closure of $X \setminus D_+$.
Clearly $E_i$ is a closed domain. Since our original sets $D_1 \ldots D_k$ cover $X$ we know that $X \setminus D_+ \subseteq D_i$, and since $D_i$ is closed this implies that $E_i \subseteq D_i$. Because of this we also have that $f|_{E_i} = g_i|_{E_i}$. Finally, because $X \setminus D_+ \subseteq E_i$, we know that the sets 
$D_1, \ldots D_{i - 1}, E_i, D_{i + 1}, \ldots D_k$ cover $X$. 
\end{proof}

The following is the main result of \Cref{subsec-negative-computable-implies-piecewise-linear},
showing a limitation on the computational power of CRCs stably computing functions in the direct sense.

\begin{lem}
\label{lem:positive-continuous-piecewise-rational-linear}
  Let $f:\Rp^k \to \Rp$ be stably computed by a CRC $\calC=(\Lambda,R,\Sigma,\{Y\})$. Then $f$ is positive-continuous and piecewise rational linear. 
\end{lem}

\begin{proof}
By Lemma~\ref{lem-nonnegative-computable-implies-positive-continuous}, we know that $f$ is positive-continuous and piecewise linear. By a general topological argument one could show that any function with these properties has domains of definition that are closed relative to $\R_{>0}^\Sigma$, but since by Lemma~\ref{lem-shifted-closed} the domains we constructed earlier already have this property, we won't give the general proof here. By Lemma~\ref{lem:replace-domains-of-defn-with-closed-domains} we can replace the closed sets that give the domains of definition of $f$ by closed domains. If some of the domains produced by Lemma~\ref{lem:replace-domains-of-defn-with-closed-domains} are empty we can simply ignore them in what follows. Since all of the nonempty domains are the closures of nonempty open sets, they must each contain some open ball. By Lemma~\ref{lem:open-ball-implies-rational}, $f$ is a rational linear function when restricted to each of these closed domains, so $f$ is piecewise rational linear on $\R_{>0}^\Sigma$. 

For any proper subset $U$ of the input species, one can apply the above argument to the reduced CRN that  discards all species not producible from the given inputs to show that $f$ is continuous and piecewise rational linear on $D_U$. This shows that $f$ is a positive-continuous piecewise rational linear function on all of $\Rp^\Sigma$.
\end{proof}

\subsection{Negative Result: Dual-Rail Computable Functions are Continuous Piecewise Rational Linear}
\label{subsec:negative-dual-rail}

The following result,
a dual-rail analog of Lemma~\ref{lem:open-ball-implies-rational},
is not necessary for the proof of the main result of this section (Lemma~\ref{lem-computable-implies-piecewise-linear}), but it may be of independent interest. 

\begin{proposition}\label{prop-siphon-implies-linear}
  Let $f:\R^k \to \R$ be stably dual computed by a CRC.
  Let $\Omega$ be 
  a stabilizing siphon. 
  Then $f$ restricted to inputs that have a dual rail representation in $\Sigma(\Omega)$ is linear.
\end{proposition}

\begin{proof}
A dual-rail computing CRC can be thought to \emph{directly} compute two separate functions $\hat{f}^+, \hat{f}^-: \Rp^{2k} \to \Rp$ such that $\hat{f} = \hat{f}^+ - \hat{f}^-$ where $\hat{f}$ is a dual rail representation of $f$.
By Lemma~\ref{lem-siphon-implies-linear-nondual} we know that $\hat{f}^+$ and  $\hat{f}^-$ are rational linear when restricted to $\Sigma(\Omega)$.
The proposition follows because linearity is closed under subtraction.
\end{proof}

The following is our main negative result for dual-rail CRC's, a dual-rail analog of \Cref{lem:positive-continuous-piecewise-rational-linear}.

\begin{lem}\label{lem-computable-implies-piecewise-linear}
  Let $f:\R^k \to \R$ be stably dual-computable by a CRC.
  Then $f$ is continuous and piecewise rational linear.
\end{lem}

\begin{proof}
  Let $\calC$ be the CRC stably computing a dual-rail representation $\hf$ of $f$, with input species
  $X_1^+$, $\ldots$, $X_k^+$, $X_1^-$, $\ldots$, $X_k^-$
  and output species $Y^+,Y^-$.

  Similarly to the proof of Proposition~\ref{prop-siphon-implies-linear}, a dual-rail computing CRC can be thought to \emph{directly} compute two separate functions $\hat{f}^+, \hat{f}^-: \Rp^{2k} \to \Rp$ such that $\hat{f} = \hat{f}^+ - \hat{f}^-$ where $\hat{f}$ is a dual rail representation of $f$. Since $\hat{f}^+$ and $\hat{f}^-$ are stably computed by a CRC, Lemma~\ref{lem:positive-continuous-piecewise-rational-linear} implies that they are both piecewise rational linear, and since piecewise rational linear functions are closed under subtraction, this implies that $f$ is also piecewise rational linear. It remains to show that $f$ is continuous.

  For \emph{any} input $\vx' \in \R^k$ to $f$, there is an initial state $\vx \in \R^\Sigma_{> 0}$ representing $\vx'$,
  with strictly positive concentrations of all input species.
  (For example, if $\vx'(1) = 5$, we can choose $\vx(X^+_1) = 6$ and $\vx(X^-_1) = 1$.)
  Let $\vx'_1,\vx'_2,\ldots \in \R^k$ be any sequence of inputs  to $f$  such that $\lim_{i\to\infty} \vx'_i = \vx'$.
  Let us represent each input $\vx'_i$ in the sequence by an initial state 
  $\vx_i \in \R^\Sigma_{> 0}$ such that $\lim_{i\to\infty} \vx_i = \vx$.
  Then the sequence of initial states
  has the property that all 
  of the $\vx_i$'s obey $[\vx_i]=\Sigma$.
  By Lemma~\ref{lem-nonnegative-computable-implies-positive-continuous}, $f$ is continuous on the domain in which all input species are positive, which includes the input represented by $\vx$ and the inputs represented by all 
  of the $\vx_i$'s.
  Therefore, $f(\vx') = \lim_{i\to\infty} f(\vx'_i)$, so $f$ is continuous. 
\end{proof}

\section{Extensions}

\subsection{A Game-Theoretic Formulation of Rate-Independent Computation}
\label{sec:game}
In this section, we use our notion of a valid rate schedule (\Cref{defn:valid-rate-schedule}) to propose a framework for studying rate-independent computation. Intuitively, we consider a model where the CRC experiences intermittent ``shocks" where the system behaves erratically and the user of the CRC loses some amount of control of the rates of its reactions. We then say that the CRC rate-independently computes a function if, despite these shocks, the CRC still converges to the correct output.  

In principle, one could consider different versions of the above model, depending on how much control the user of the CRC is expected to have over the kinetics of the system and how severe the shocks to the system are expected to be. Our goal in this section is to introduce a framework that is general enough to accommodate these various situations. 

We will formalize the situation presented above by describing it as an infinite game, played by two players, which we will call the Demon and the Chemist. Here the Chemist represents the user of the CRC, who is trying to perform some rate-independent computation, and the Demon represents (per tradition) the forces of nature, human failing, etc, which create the intermittent shocks to the system. In this game, the Demon and the Chemist take turns building a valid rate schedule. The Chemist wins the game if the amount of the output species converges to the correct value as time goes to infinity. If the Chemist has a winning strategy for the game associated with a given CRC, then we say that the CRC rate-independently computes the desired function.

In order to model the different levels of control that the Chemist has over the system, and the different levels of severity of the shocks that the Demon can induce in the system, we can consider games where the Demon and the Chemist are only allowed to play moves from some restricted class of valid rate schedules. In this paper, we have been interested in functions that are rate-independently computable in a very strong sense, i.e. where the shocks are arbitrarily severe---so the Demon is allowed to play any valid rate schedule (Strong Demon). 

Once we have developed the game-theoretic model of rate-independent computation more explicitly we will be able to use our main result \Cref{thm:computable-characterization} to deduce the following claim: the class of functions that one can compute rate-independently with a Strong Demon is effectively insensitive to the amount of control that the Chemist has over the system. In particular, the class of functions that are rate-independently computable by Strong Demon, Strong Chemist games, where both the Demon and the Chemist can play any valid rate law, is \emph{the same} as any Strong Demon complexity class, where the Chemist is restricted to playing only a subset of valid rate laws. 

To formalize the above discussion, we can use the notion of an infinite game~\cite{mycielski1992games}. 

\begin{defn}
An \emph{infinite game} is a pair of a set $X$ of possible \emph{moves} and a set $A \subseteq X^\N$ called the \emph{payoff set}.
\end{defn}

Intuitively one should imagine that two players, player I and player II, are taking turns playing moves from the set $X$. Together they form a sequence $(x_1, x_2, \ldots) \in X^\N$ and player I wins if this sequence is in $A$. In most games of interest, the set of moves that a player can make is restricted by the state of the game. The definition above captures this by using a payoff set such that a player will lose instantly if they play a move outside of some ``legal" set of moves.
For any previously played sequence we can thus let \emph{moveset} $M(x_1, \ldots, x_n)$ be the legal set of moves for the two players (depending on whether $n$ is odd or even) and assume that the payoff set $A$ is defined relative to these sets of allowed moves accordingly (i.e., the first time that an illegal move is made, the other party wins). 
An infinite sequence of moves is consistent with the moveset if every move by both players is legal.

\begin{defn}
For an infinite game $(X, A)$ a \emph{winning strategy for player II} is a collection of functions $f_{2k}: X^k \to X$ so that the sequence
\[\seq{x_1, f_2(x_1), x_3, f_4(x_1, x_3), x_5, f_6(x_1, x_3, x_5), \ldots}\]
is never in $A$ for every choice of $x_1, x_3, \ldots$. A \emph{winning strategy for player I} is defined similarly. 
\end{defn}

For our context, 
fix a CRC $\calC$ and an initial state $\vx \in \Rp^\Sigma$.
The set of possible moves $X$ is the set of all rate schedules that are zero at times $t > 1$.
Such moves (i.e., rate schedules) can be naturally concatenated into a longer rate schedule: the concatenation $\vf_1 \circ \ldots \circ \vf_n$ is the rate schedule which follows $\vf_i$ for time $i-1 \leq t < i$ and zero elsewhere.
Our game captures the idea that the Demon wins by making the constructed rate schedule not converge to the desired output, where the Demon is player I and the Chemist is player II.
We allow the Demon to play any valid rate schedule, but the Chemist may have varying power as we discuss below.

The most unrestricted moveset contains all possible valid rate schedules. 
We call this case the Strong Demon Strong Chemist game.
\begin{defn}   \label{def:strong-demon-strong-chemist}
For a CRC $\calC$ and function $f: \Rp^k \to \Rp$, an initial state $\vx \in \Rp^\Sigma$, the \emph{Strong Demon Strong Chemist game} is the game where 
the moveset $M(\vf_1, \ldots, \vf_n)$ includes all the moves  $\vf_{n+1}$ such that $\vf_1 \circ \ldots \circ \vf_{n+1}$ is a valid rate schedule starting at $\vx$,
and a sequence of moves $\seq{\vf_1, \vf_2, \ldots}$ consistent with the moveset is in the payoff set iff  $\lim_{t \to \infty} \vrho(t) \upharpoonright \set{Y} \neq f(\vx)$ for trajectory $\vrho$ corresponding to $\vf_1 \circ \vf_2 \circ \ldots$ starting at $\vx$.
\end{defn}

\begin{lem}\label{lem:SDSC-equals-stable}
The following are equivalent for a CRC $\calC$ and function $f: \Rp^k \to \Rp$:
\begin{enumerate}
    \item The Chemist has a winning strategy for the Strong Demon Strong Chemist game associated with $\calC$ and $f$ for every initial state $\vx \in \Rp^\Sigma$.
    \item $\calC$ stably computes $f$.
\end{enumerate}
\end{lem}

\begin{proof}
First suppose that $\calC$ stably computes. Then by \Cref{thm:valid-reachable-implies-segment-reachable} we know that for any first move $\vf_1$ the Demon plays, the state $\vrho(1)$ will be reachable from the initial state $\vx$. By the definition of stable computation there must be a state $\vz$ reachable from $\vx$ so that $\vz$ is output stable with the correct amount of output. By \Cref{lem:segment-reachable-implies-valid-rate-schedule} there is a valid rate schedule $\vf_2$ that goes from $\vx$ to $\vz$, and by ``rescaling" the rate schedule by increasing the rate of every reaction we can guarantee that $\vf_2$ is completed within a single unit of time. Since $\vz$ is output stable, every further move $\vf_4, \vf_6, \ldots$ that the Chemist makes can just be taken to be the zero rate schedule (where no reaction occurs). This specifies a winning strategy for the Chemist for the Strong Demon Strong Chemist game associated with $\calC$, $f$, and $\vx$.  

On the other hand, suppose that $\calC$ doesn't stably compute. Then let $\epsilon$, $\vx$, and $\vz$ be the ones given by \Cref{lem:not-stably-compute-implies-can-reach-far-from-correct}. The Demon can play a rate schedule which goes from $\vx$ to $\vz$. After this, every time it's the Demon's turn the CRC will be in a state $\vo$ which is reachable from $\vz$, so the Demon will always be able to play a valid rate schedule that takes the CRN to $\vo'$ as in \Cref{lem:not-stably-compute-implies-can-reach-far-from-correct} with $|\vo'(Y) - f(\vx)| > \epsilon$. This shows that the Demon has a winning strategy for the Strong Demon Strong Chemist game associated with $\calC$, $f$ and $\vx$, so certainly the Chemist doesn't have a winning strategy. 
\end{proof}

We can also consider games where the players are restricted to only playing a subset of the valid rate schedules.
We are particularly interested in games where the Demon is allowed to play any valid rate schedule as above, but the Chemist is restricted to only playing rate schedules from a particular restricted moveset $M$. 
It turns out that as long as the chemist has \emph{some fair} rate schedule to play, the computational power is the same as with a Strong Chemist.

We first define the validity and fairness restrictions on movesets. 
Then given a restricted moveset for the Chemist we define the corresponding game similar to \cref{def:strong-demon-strong-chemist}.

\begin{defn}  \label{def:valid-fair-movesets}
A moveset $M(\vf_1, \ldots, \vf_n)$ is valid if it includes only moves  $\vf_{n+1}$ such that $\vf_1 \circ \ldots \circ \vf_{n+1}$ is a valid rate schedule starting at $\vx$.
Further, we say that a valid moveset $M(\vf_1, \ldots, \vf_n)$ is Chemist-fair if there is some $\vH: \Rp^\Lambda \to \Rp^R$ (with  the same properties as in \Cref{defn:fair-rate-schedule}) such that for any odd $n$,
the moveset contains at least one $\vf_{n+1}$ such that the trajectory $\vrho$ corresponding to
$\vf = \vf_1 \circ \ldots \circ \vf_{n+1}$ starting at $\vx$
satisfies 
$\vf_\alpha(t) \ge \vH_\alpha(\vrho(t))$ for every reaction $\alpha$ and for all $t \in [n,n+1)$.
\end{defn}

\begin{defn}
Consider a CRC $\calC$, function $f: \Rp^k \to \Rp$, and an initial state $\vx \in \Rp^\Sigma$.
Let $M(\vf_1, \ldots, \vf_n)$ be any valid moveset.
We say the  \emph{Strong Demon $M$-Chemist} is the game where 
a sequence of moves $\seq{\vf_1, \vf_2, \ldots}$ consistent with the moveset is in the payoff set iff  $\lim_{t \to \infty} \vrho(t) \upharpoonright \set{Y} \neq f(\vx)$ for trajectory $\vrho$ corresponding to $\vf_1 \circ \vf_2 \circ \ldots$ starting at $\vx$.
\end{defn}

It may seem that depending on what moves the Chemist is allowed to play, different classes of functions are computable (i.e., the Chemist has a winning strategy). 
However, the following theorem shows that the class of functions is invariant to the class of movesets that the Chemist has available, as long as there is always at least one fair move.

\begin{definition}
A \emph{Strong Demon complexity class} $\mathsf{SD}$ is a set of functions $f: \Rp^k \to \Rp$ for all $k \in \N$, that is characterized in the following way: there is a (non-empty) set of Chemist-fair movesets $\mathcal{M}$ so that $f: \Rp^k \to \Rp$ is in $\mathsf{SD}$ if and only if there is a CRC $\calC$ such that for each $M \in \mathcal{M}$ and initial state $\vx$ the Chemist has a winning strategy for the associated Strong Demon $M$-Chemist game. 
\end{definition}

\begin{thm}
Every Strong Demon complexity class is the same as the Strong Demon Strong Chemist class. 
\end{thm}

\begin{proof}
If a Chemist has a winning strategy in a Strong Demon $M$-Chemist game, then clearly the Chemist has a winning strategy in the Strong Demon Strong Chemist game. 
Thus every Strong Demon complexity class is a subset of the Strong Demon Strong Chemist class.
For the other direction,
let $\mathsf{SD}$ be some Strong Demon complexity class, and let $\calM$ be the nonempty set of Chemist-fair movesets associated with $\mathsf{SD}$.
Suppose that $f: \Rp^k \to \Rp$ is in the Strong Demon Strong Chemist class. Then by \Cref{lem:SDSC-equals-stable} and \Cref{thm:computable-characterization} we know that $f$ is fairly computable by a CRC $\calC$. For any moveset $M \in \calM$, let us show that the Chemist has a winning strategy for the Strong Demon $M$-Chemist game associated with $\calC$, and $f$. Since $M$ is Chemist-fair we know that there is some function $\vH: \Rp^\Lambda \to \Rp^R$ as in \Cref{def:valid-fair-movesets}. 
On their turn, the Chemist simply plays any move in the moveset $M$ that satisfies the condition on $f_{n+1}$ in \cref{def:valid-fair-movesets}.
This strategy is winning for the Chemist because 
the total rate schedule $\vf$ obtained by concatenating the $\vf_i$'s is fair, since on the set $T = \bigcup_{n = 1}^\infty [2n-1, 2n)$ we know that $\vf_\alpha(t) \ge \vH_\alpha(\vrho(t))$. 
Since $\calC$ was chosen to fairly compute $f$ this proves that $\lim_{t \to \infty} (\vrho(t)\upharpoonright \set{Y}) = f(\vx)$. 
This shows that $f \in \mathsf{SD}$, so any Strong Demon $M$-Chemist complexity class is the same as the Strong Demon Strong Chemist complexity class. 
\end{proof}

\subsection{Non-zero initial context}
\label{sec:initial-context}
    Throughout this paper we have assumed 
    that the only species allowed to be present at the start of the computation are the input species. Instead, one could consider a model where certain non-input species $Z_1 \ldots Z_n$, called the \emph{initial context}~\cite{CheDotSolNaCo14}, 
    have a fixed, nonzero rational concentration at the start of the computation. In this setting, we can clearly compute more functions than in the setting without initial context: for instance, we can easily compute $f(x_1, \ldots, x_k) = C$ for some nonzero constant $C$, which is impossible without initial context because $f$ is affine  but not linear. 
    
    In fact, we can dual-rail compute any continuous piecewise rational affine function $f: \R^k \to \R$,
    i.e., any function that is a rational linear function plus a rational constant: $f(\vec{x}) = \vec{a} \cdot \vec{x} + b$. 
    To see this, first note that we can compute any rational affine function by using the initial context to offset the value at $f(\vec{0}) = b$.
    In fact, we can simply let output species $Y^+$ and $Y^-$ be the initial context, with $\vec{x}(Y^+) = b$ initially if $b>0$ and $\vec{x}(Y^-) = -b$ otherwise.
    Similar machinery to the proof of Lemma~\ref{lem-piecewise-linear-implies-computable} can be used to extend this to continuous \emph{piecewise} rational affine functions:
    By Theorem~\ref{thm-min-max-representation}, any continuous piecewise rational affine function can be represented in max-min form, and then our construction from Section~\ref{negative-piecewise-linear-implies-computable} shows that we can compute our given function $f$. In the direct computation setting, by a construction like the one in Section~\ref{sec:positive-continuous-piecewise-rational-linear-computable} we can compute any positive-continuous piecewise rational affine function $f: \R^k_{\ge 0} \to \R_{\ge 0}$. 
    
    It also turns out that, even with initial context, we can't compute any more functions than these. To see this, note that without loss of generality we can assume that there is only one initial context species $Z$ with initial concentration 1, since we can modify any CRC with initial context to include
    reactions
    that convert $Z$ into $Z_1,\ldots,Z_n$ with appropriate concentrations.\footnote{
    For example, to simulate initial context 
    $\{3/2\ Z_1, 1/2\ Z_2 \}$ from $\{1\ Z\}$, add the reactions
    $2Z \to Z'$
    and
    $Z' \to 3 Z_1 +  Z_2$.
    }
    Now let $g(x_1, \ldots, x_k, z)$ be the value that the CRC computes when the input species have initial values $x_1,\ldots, x_k$ and the initial context species has value $z$. 
    A priori, $g$ is only well-defined when $z = 1$, but because paths
    remain valid after scaling we know that 
    \[
    g(x_1, \ldots, x_k, z) = 
    z \cdot g(x_1/z, \ldots, x_k/z, 1) = 
    z \cdot f(x_1/z, \ldots, x_k/z)\]
    for any value of $z > 0$. This shows that $g$ is well-defined on 
    $D_U$ (recall Definition~\ref{def:positive-continuous}) for every $U$ that contains $Z$, so we can apply the results of \Cref{subsec-negative-computable-implies-piecewise-linear,subsec:negative-dual-rail} to characterize $g$ on these domains. Using the fact that $f(x_1, \ldots, x_k) = g(x_1, \ldots, x_k, 1)$ gives us the desired result. 
    
    Since every continuous function on a compact domain is uniformly continuous,
    it can be uniformly approximated by continuous piecewise rational affine functions. 
    This shows that we can use rate-independent CRNs to approximate continuous functions. Note that for the negative argument above to work, it was important that all of the initial concentrations of the initial context species were rational. For practical purposes, this assumption is not at all restrictive, but it might be of theoretical interest to know what other functions can be computed if the initial concentrations are allowed to be arbitrary real numbers.

\section{Conclusion and Open Questions}
\label{sec-conclusion}

We characterized the class of functions computable in a manner that is absolutely robust to reaction rates in the continuous model of chemical kinetics.
Such rate-independent computation must rely solely on reaction stoichiometry---which reactants, and how many of each, become which products, and how many of each?
We considered two methods of encoding inputs and outputs: direct and dual-rail.
The dual-rail encoding permits easier composition of modules and can represent negative values;
we characterized its computational power as continuous, piecewise rational linear. 
The direct encoding, however, allows computing functions that are discontinuous at the faces of the nonnegative orthant.
For both encodings, we 
showed matching 
negative results (showing that nothing more can be computed) and positive results (describing CRNs computing any function in the class).

Since rate-independent computation does not require  difficult-to-achieve tuning of parameters or reaction conditions, it may be significantly more ``engineerable'' than rate-dependent computation.
More generally, our work also helps uncover the multifaceted sources of chemical computational power by disentangling the control of stoichiometry from reaction rates.

We now describe some natural open questions.

\paragraph{Reaction complexity of stably computable functions}

    An interesting question regards the description complexity of functions stably computable by CRNs.
    Some piecewise linear functions have a number of pieces exponential in the number of inputs;
    for example, $f(x_1,\ldots,x_{2k}) = \min(x_1,x_2) + \min(x_3,x_4) + \ldots + \min(x_{2k-1}, x_{2k})$ has $2^k$ linear pieces.
    If we express this function in $\max \min g_i$ form of Theorem~\ref{thm-min-max-representation},
    we need $2^k$ different linear $g_i$,
    and thus the construction in the proof of Lemma~\ref{lem-piecewise-linear-implies-computable} would require exponentially many species and reactions.
    However, this particular $f$ has a more succinct CRN that stably computes it, namely the reactions 
    \begin{eqnarray*}
      X_1 + X_2 &\to& Y
      \\
      X_3 + X_4 &\to& Y
      \\
      \vdots
      \\
      X_{2k-1} + X_{2k} &\to& Y.
    \end{eqnarray*}
    Given a positive-continuous, piecewise rational linear function $f$, how can we tell whether it has a more compact CRN stably computing it than our construction? 
    If it does, how can we arrive at it?

\paragraph{Requiring ``always'' fair rate schedules}
\Cref{lem:stable-computation-implies-fair-computation-feedforward} shows that for a feedforward CRC that stably computes a function $f$, any fair rate schedule converges to the correct output of $f$, where fair essentially means that applicable reactions must have positive rate for an infinite subset of times.
In other words, the adversary is allowed for some times outside this subset for the rate schedule to be unfair:
to ``starve'' some applicable reactions by keeping their rates at 0 despite all reactants being present.
Since such rate schedules seem physically implausible, it is natural to consider a modified definition of computation, one that requires every rate schedule that is \emph{always-fair} (i.e., for \emph{all} time, applicable reactions must have positive rate) to converge to the correct output.
The following CRC shows that these two requirements can result in different behaviors for a given CRC:
\begin{eqnarray*}
X &\to& X+C
\\
X &\to& X'
\\
C+X' &\to& C+Y
\end{eqnarray*}
With initial concentration $x$ of $X$,
every always-fair rate schedule converges to concentration $x$ of $Y$, 
so it computes the identity function $f(x)=x$ under this modified definition.
(Note that always-fairness requires the first two reactions to have positive rate at time 0, so $C$ immediately becomes present, at which point it is inevitable to convert all $X$ to $X'$ and all $X'$ to $Y$.)
However, the CRC does not stably compute $f$:
an initially unfair adversary can execute the second reaction to completion,
starving the first reaction until it becomes inapplicable.
This schedule never produces any $C$, so $Y$ stays at 0,
yet the schedule is fair by the original definition since,
once all reactions become inapplicable,
subsequent rate 0 of all reactions for all time vacuously satisfies the definition of fair.
It is an interesting question is whether,
under the modified definition of always-fair,
some CRC can compute a function that is not stably computable.

\paragraph{Arbitrary but fixed rate constants}
    A related notion of rate-independence is one where the form of rate-law cannot vary,
    but the constant parameters can,
    e.g., mass-action rates, where an adversary picks the rate constants 
    (possibly depending on the initial input).
    For example, consider the following CRN with input species $A,B,C,X$ and output species $Y$.
    \begin{eqnarray*}
      A + X &\to& A + Y
      \\
      B + Y &\to& B + X
      \\
      A + B &\to& C
      \\
      C + Y &\to& C + X
      \\
      3C &\to& \emptyset.
    \end{eqnarray*}
    Let $a, b, c, x$ denote the initial concentrations of species $A$, $B$, $C$, $X$.
    This system does not stably compute any function in the model defined in this paper because on input $a=b$, 
    it can stabilize to any value of output $y$ between $0$ and $x$.

    In contrast, consider the above system under mass-action kinetics,
    where an adversary picks the rate constants, but they remain constant over time.
    Because $A$ and $B$ are required to produce $C$, and at least one of them goes to 0 by the third reaction, the concentration of $C$ approaches $0$ as time goes to infinity by the final reaction, no matter the rate constants.
    If $a>b$, then also the concentration of $B$ approaches $0$ but the concentration of $A$ remains bounded away from $0$. Therefore, the output $y$ converges to $x$, regardless of what the rate constants are. Similarly, if $b>a$, the output $y$ approaches $0$.
    When $a=b$, the concentrations of $A, B, C$ approach $0$ at different rates. The concentrations of $A, B$ are $\Theta(\frac{1}{t})$ at time $t$ (rate of bimolecular decay), 
    and the concentration of $C$ is $\Theta(\frac{1}{\sqrt{t}})$ (rate of trimolecular decay). As a result, the effective rate of conversion of $Y$ to $X$ via the channel $C + Y \to C + X$ is $\Theta(\frac{1}{\sqrt{t}})$. 
    Since this is $\omega(\frac{1}{t})$ the output $y$ always converges to $0$ regardless of the rate constants (in our particular case the concentration of $Y$ is $e^{-\Theta(\sqrt{t})}$). From the above, this CRN computes $f(a, b, c, x) = x$ when $a > b$ and $f(a, b, c, x) = 0$ when $a \leq b$, no matter what the rate constants are.
    This function is discontinuous at points where $a=b$,
    so it is not positive-continuous, 
    thus not stably computable by any CRN under our model of rate-independence.
    It remains open to classify what functions can be computed by
    mass-action CRNs in which rate constants are chosen adversarially.

\paragraph{Axiomatic derivation of reachability}
An important contribution of this paper is to develop segment-reachability as the ``correct'' notion of rate-independent reachability.
While we justify our definition of segment-reachability by its equivalence to valid rate schedules (\Cref{thm:valid-reachable-implies-segment-reachable}),
one could imagine taking an axiomatic approach:
Any notion of rate-independent reachability ought to satisfy certain constraints. 
For instance, it should be transitive.
Also, if $\vd$ is straight-line reachable from $\vc$, then it is evidently allowed by stoichiometry for the CRN to evolve from $\vc$ to $\vd$, so $\vd$ should be reachable from $\vc$. 
Of course, there are many relations that satisfy these two constraints: for instance the relation where every state is reachable from every other state satisfies these two constraints---but this is evidently too permissive. 
Indeed, our notion of segment-reachability is the transitive closure of the relation of straight-line reachability (Corollary~\ref{cor:segto-is-transitive}), so it is minimal among all of the relations with these two properties.
However, since our goal is to develop a notion of reachability that is as unrestricted as possible, 
a more compelling axiomatic derivation would follow from simple set of natural conditions for which segment-reachability is \emph{maximal}.

\paragraph{Absolute inhibition}
In biology it is not uncommon to have rate laws with explicit inhibitors; the higher the concentration of the inhibitor, the slower the reaction. 
Formally, general rate laws have been described in which reactants are partitioned into consumed species, species that increase the reaction rate (catalysts), and species that decrease the rate (inhibitors)~\cite{fages2015inferring}.
In our model, while inhibitors can be modeled mechanistically as sequestering reacting species in a non-reactive form (e.g., $A + B \to C$ is inhibited by $I$ via the reaction $A + I \revrxn AI$), such inhibition just  slows down the reaction rate but does not prevent the reaction entirely.
In contrast one can imagine a definition of reachability in which a reaction is applicable 
only if all of its reactants are present \emph{and} all of its inhibitors are absent.
It remains an open question whether the computational power of rate-independent computation changes if such absolute inhibition is allowed.
We note that this change drastically changes the notion of reachability because it is no longer additive: it might be that $\vx \segto \vy$, but $\vx + \vc \not\segto \vy + \vc$ if $\vc$ contains some inhibitors of reactions occurring in the first path.
Further, for discrete CRNs, absolute inhibition dramatically expands computational power of stable computation to 
Turing universality~\cite{CooSolWinBru09,LiptonVAS} (i.e., the ability to compute any function computable by any algorithm).
thus it seems reasonable to conjecture that it also expands the computational power in the continuous setting.

\newcommand{\yes}{\mathsf{Y}}
\newcommand{\no}{\mathsf{N}}

\paragraph{Decision problems}
In the discrete CRN model, particularly the subset of it known as population protocols~\cite{angluin2006passivelymobile},
a major focus of research is on \emph{decision} problems with a yes/no output, a.k.a. predicates.
The typical output convention partitions the set $\Lambda$ of species into two disjoint subsets $\Lambda = \Lambda_\yes \cup \Lambda_\no$, the ``yes voters'' and ``no voters''.
The goal of stable computation in this setting is to reach a configuration with a unanimous, correct vote
(e.g., if the correct answer is yes, only species in $\Lambda_\yes$ are present),
that is also stable in the sense that no incorrect voter is producible.
In this setting it has been shown that exactly the \emph{semilinear} predicates can be stably decided~\cite{AngluinAE2006semilinear, angluin2006passivelymobile} by discrete CRNs.

The concept can be similarly defined with continuous CRNs using our notion of segment-reachability.
Say that a CRN with voting species defined as above \emph{stably computes} a predicate $\phi: \Rp^k \to \{\yes,\no\}$ if,
for any initial configuration $\vx \in \Rp^k$,
for any configuration $\vc$ such that $\vx \segto \vc$,
there is a configuration $\vy$ such that $\vc \segto \vy$,
and $\vy$ is \emph{stably correct},
meaning that for all $\vy'$ such that $\vy \segto \vy'$,
we have $\emptyset \neq [\vy'] \subseteq \Lambda_{\phi(\vx)}$. 
In other words, stable computation leads to a nonempty configuration with only ``correct'' votes (according to $\phi$),
and this is also true of every configuration reachable from there.
We also say in this case that the CRN \emph{stably decides} the set $\phi^{-1}(\yes) \subseteq \Rp^k$ of inputs that map to output $\yes$.
What is the class of sets $S \subseteq \Rp^k$ that are stably decidable by this definition?

For discrete CRNs, the question of ``how long'' a system takes to stably compute a function has received much interest~\cite{LeaderElectionDIST, alistarh2017time,alistarh2018recent,gkasieniec2018fast}.
In general, asking questions of time-complexity of continuous computation seems more difficult than in the discrete setting. 
For general polynomial ODEs, the breakthrough work of Bournez, Graça, and Pouly~\cite{bournez2017odes} established a surprisingly tight connection between the length of the trajectory and the Turing machine computation time.
Closer to our domain of interest, prior work has studied asymptotic convergence speed for the composition of simple CRN motifs as a function of the number of (feedforward) layers~\cite{seelig2009time}.
Although not explicitly stated in terms of rate-independent computation, 
these modules compute in the rate-independent manner as studied here.
It is interesting to ask whether these techniques could be adopted to our constructions to articulate and help resolve questions of computation time.

\section{Acknowledgements}
We thank Manoj Gopalkrishnan, Elisa Franco, Damien Woods, and the organizers and participants of the American Mathematical Institute workshop on Mathematical Problems Arising from Biochemical Reaction Networks for insightful discussions.
We are grateful to anonymous reviewers for insightful comments and suggestions that have greatly improved this paper. 
DD was supported by NSF grants 2211793, 1900931, and 1844976. 
DS was supported by NSF grants CCF-1901025, CCF-1652824, and a Sloan Foundation Research Fellowship.
WR was supported by the National Science Foundation Graduate Research Fellowship under Grant No. DGE1745303. 
HC was supported by MOST (Taiwan) grants 107-2221-E-002-031-MY3 and 110-2223-E-002-006-MY3.

\bibliographystyle{abbrv}
\bibliography{tam}

\begin{thebibliography}{10}

\bibitem{alistarh2017time}
D.~Alistarh, J.~Aspnes, D.~Eisenstat, R.~Gelashvili, and R.~L. Rivest.
\newblock Time-space trade-offs in population protocols.
\newblock In {\em SODA 2017: Proceedings of the Twenty-Eighth Annual ACM-SIAM
  Symposium on Discrete Algorithms}, pages 2560--2579. SIAM, 2017.

\bibitem{alistarh2018space}
D.~Alistarh, J.~Aspnes, and R.~Gelashvili.
\newblock Space-optimal majority in population protocols.
\newblock In {\em SODA 2018: Proceedings of the Twenty-Ninth Annual ACM-SIAM
  Symposium on Discrete Algorithms}, pages 2221--2239, 2018.

\bibitem{alistarh2018recent}
D.~Alistarh and R.~Gelashvili.
\newblock Recent algorithmic advances in population protocols.
\newblock {\em ACM SIGACT News}, 49(3):63--73, 2018.

\bibitem{angeli2006structural}
D.~Angeli, P.~De~Leenheer, and E.~D. Sontag.
\newblock On the structural monotonicity of chemical reaction networks.
\newblock In {\em 45th {IEEE} Conference on Decision and Control}, pages 7--12.
  {IEEE}, 2006.

\bibitem{angeli2007petri}
D.~Angeli, P.~De~Leenheer, and E.~D. Sontag.
\newblock A {P}etri net approach to the study of persistence in chemical
  reaction networks.
\newblock {\em Mathematical Biosciences}, 210(2):598--618, 2007.

\bibitem{angluin2006passivelymobile}
D.~Angluin, J.~Aspnes, Z.~Diamadi, M.~Fischer, and R.~Peralta.
\newblock Computation in networks of passively mobile finite-state sensors.
\newblock {\em Distributed Computing}, 18:235--253, 2006.
\newblock Preliminary version appeared in PODC 2004.

\bibitem{AngluinAE2006semilinear}
D.~Angluin, J.~Aspnes, and D.~Eisenstat.
\newblock Stably computable predicates are semilinear.
\newblock In {\em PODC 2006: Proceedings of the twenty-fifth annual ACM
  symposium on Principles of distributed computing}, pages 292--299, New York,
  NY, USA, 2006. ACM Press.

\bibitem{angluin2006fast}
D.~Angluin, J.~Aspnes, and D.~Eisenstat.
\newblock Fast computation by population protocols with a leader.
\newblock {\em Distributed Computing}, 21(3):183--199, Sept. 2008.
\newblock Preliminary version appeared in DISC 2006.

\bibitem{aspnes2007introduction}
J.~Aspnes and E.~Ruppert.
\newblock An introduction to population protocols.
\newblock {\em Bulletin of the European Association for Theoretical Computer
  Science}, 93:98--117, 2007.

\bibitem{barkal1997robustness}
N.~Barkai and S.~Leibler.
\newblock Robustness in simple biochemical networks.
\newblock {\em Nature}, 387(6636):913--917, 1997.

\bibitem{belleville2017hardness}
A.~Belleville, D.~Doty, and D.~Soloveichik.
\newblock Hardness of computing and approximating predicates and functions with
  leaderless population protocols.
\newblock In {\em ICALP 2017: 44th International Colloquium on Automata,
  Languages, and Programming}, volume~80 of {\em Leibniz International
  Proceedings in Informatics (LIPIcs)}, pages 141:1--141:14, 2017.

\bibitem{bournez2017odes}
O.~Bournez, D.~S. Gra\c{c}a, and A.~Pouly.
\newblock Polynomial time corresponds to solutions of polynomial ordinary
  differential equations of polynomial length.
\newblock {\em J. ACM}, 64(6), oct 2017.

\bibitem{brijder2019computing}
R.~Brijder.
\newblock Computing with chemical reaction networks: a tutorial.
\newblock {\em Natural Computing}, 18(1):119--137, 2019.

\bibitem{cardelli2011strand}
L.~Cardelli.
\newblock Strand algebras for {DNA} computing.
\newblock {\em Natural Computing}, 10(1):407--428, 2011.

\bibitem{cardelli2012cell}
L.~Cardelli and A.~Csik{\'a}sz-Nagy.
\newblock The cell cycle switch computes approximate majority.
\newblock {\em Scientific Reports}, 2, 2012.

\bibitem{case2018reachability}
A.~Case, J.~H. Lutz, and D.~M. Stull.
\newblock Reachability problems for continuous chemical reaction networks.
\newblock {\em Natural Computing}, 17(2):223--230, 2018.

\bibitem{chalk2021composable}
C.~Chalk, N.~Kornerup, W.~Reeves, and D.~Soloveichik.
\newblock Composable rate-independent computation in continuous chemical
  reaction networks.
\newblock {\em {IEEE/ACM} Transactions on Computational Biology and
  Bioinformatics}, 18(1):250--260, 2021.

\bibitem{CheDotSolNaCo14}
H.-L. Chen, D.~Doty, and D.~Soloveichik.
\newblock Deterministic function computation with chemical reaction networks.
\newblock {\em Natural Computing}, 13(4):517--534, 2014.
\newblock Special issue of invited papers from DNA 2012.

\bibitem{CheDotSol14}
H.-L. Chen, D.~Doty, and D.~Soloveichik.
\newblock Rate-independent computation in continuous chemical reaction
  networks.
\newblock In {\em ITCS 2014: Proceedings of the 5th Conference on Innovations
  in Theoretical Computer Science}, 2014.

\bibitem{chen2013programmable}
Y.-J. Chen, N.~Dalchau, N.~Srinivas, A.~Phillips, L.~Cardelli, D.~Soloveichik,
  and G.~Seelig.
\newblock Programmable chemical controllers made from {DNA}.
\newblock {\em Nature Nanotechnology}, 8(10):755--762, 2013.

\bibitem{CooSolWinBru09}
M.~Cook, D.~Soloveichik, E.~Winfree, and J.~Bruck.
\newblock Programmability of chemical reaction networks.
\newblock In A.~Condon, D.~Harel, J.~N. Kok, A.~Salomaa, and E.~Winfree,
  editors, {\em Algorithmic Bioprocesses}, pages 543--584. Springer Berlin
  Heidelberg, 2009.

\bibitem{czerwinski2021reachability}
W.~Czerwi{\'n}ski and {\L}.~Orlikowski.
\newblock Reachability in vector addition systems is {A}ckermann-complete.
\newblock In {\em 2021 IEEE 62nd Annual Symposium on Foundations of Computer
  Science (FOCS)}, pages 1229--1240, 2022.

\bibitem{degrand2020graphical}
{\'E}.~Degrand, F.~Fages, and S.~Soliman.
\newblock Graphical conditions for rate independence in chemical reaction
  networks.
\newblock In A.~Abate, T.~Petrov, and V.~Wolf, editors, {\em Computational
  Methods in Systems Biology}, pages 61--78, Cham, 2020. Springer International
  Publishing.

\bibitem{ppsim}
D.~Doty and E.~Severson.
\newblock ppsim: A software package for efficiently simulating and visualizing
  population protocols.
\newblock In {\em CMSB 2021: Proceedings of the 19th International Conference
  on Computational Methods in Systems Biology}, pages 245--253, 2021.

\bibitem{LeaderElectionDIST}
D.~Doty and D.~Soloveichik.
\newblock Stable leader election in population protocols requires linear time.
\newblock {\em Distributed Computing}, 31(4):257--271, 2018.
\newblock Special issue of invited papers from DISC 2015.

\bibitem{fages2015inferring}
F.~Fages, S.~Gay, and S.~Soliman.
\newblock Inferring reaction systems from ordinary differential equations.
\newblock {\em Theoretical Computer Science}, 599:64--78, 2015.
\newblock Advances in Computational Methods in Systems Biology.

\bibitem{fages2017strong}
F.~Fages, G.~Le~Guludec, O.~Bournez, and A.~Pouly.
\newblock Strong {T}uring completeness of continuous chemical reaction networks
  and compilation of mixed analog-digital programs.
\newblock In {\em International conference on computational methods in systems
  biology}, pages 108--127. Springer, 2017.

\bibitem{gkasieniec2018fast}
L.~G{\k{a}}sieniec and G.~Staehowiak.
\newblock Fast space optimal leader election in population protocols.
\newblock In {\em Proceedings of the Twenty-Ninth Annual {ACM-SIAM} Symposium
  on Discrete Algorithms}, pages 2653--2667. SIAM, 2018.

\bibitem{Gillespie77}
D.~T. Gillespie.
\newblock Exact stochastic simulation of coupled chemical reactions.
\newblock {\em Journal of Physical Chemistry}, 81(25):2340--2361, 1977.

\bibitem{gopalkrishnan2013projection}
M.~Gopalkrishnan, E.~Miller, and A.~Shiu.
\newblock A projection argument for differential inclusions, with applications
  to persistence of mass-action kinetics.
\newblock {\em Symmetry, Integrability and Geometry: Methods and Applications},
  9(0):25--25, 2013.

\bibitem{hashemi2020composable}
H.~Hashemi, B.~Chugg, and A.~Condon.
\newblock Composable computation in leaderless, discrete chemical reaction
  networks.
\newblock In {\em 26th International Conference on {DNA} Computing and
  Molecular Programming ({DNA} 26)}. Schloss Dagstuhl-Leibniz-Zentrum f{\"u}r
  Informatik, 2020.

\bibitem{kondo2010reaction}
S.~Kondo and T.~Miura.
\newblock Reaction-diffusion model as a framework for understanding biological
  pattern formation.
\newblock {\em Science}, 329(5999):1616--1620, 2010.

\bibitem{kreyszig1991introductory}
E.~Kreyszig.
\newblock {\em Introductory functional analysis with applications}, volume~17.
\newblock John Wiley \& Sons, 1991.

\bibitem{kurtz1972relationship}
T.~G. Kurtz.
\newblock The relationship between stochastic and deterministic models for
  chemical reactions.
\newblock {\em The Journal of Chemical Physics}, 57(7):2976--2978, 1972.

\bibitem{lathrop2020population}
J.~I. Lathrop, J.~H. Lutz, R.~R. Lutz, H.~D. Potter, and M.~R. Riley.
\newblock Population-induced phase transitions and the verification of chemical
  reaction networks.
\newblock In C.~Geary and M.~J. Patitz, editors, {\em DNA 26: 26th
  International Conference on DNA Computing and Molecular Programming}, volume
  174 of {\em Leibniz International Proceedings in Informatics (LIPIcs)}, pages
  5:1--5:17, Dagstuhl, Germany, 2020. Schloss Dagstuhl--Leibniz-Zentrum f{\"u}r
  Informatik.

\bibitem{lee2013smooth}
J.~M. Lee.
\newblock {\em Introduction to Smooth Manifolds}.
\newblock Springer, 2nd edition, 2013.

\bibitem{leroux2021reachability}
J.~Leroux.
\newblock The reachability problem for {P}etri nets is not primitive recursive.
\newblock In {\em 2021 IEEE 62nd Annual Symposium on Foundations of Computer
  Science (FOCS)}, pages 1241--1252, 2022.

\bibitem{LiptonVAS}
R.~J. Lipton.
\newblock The reachability problem requires exponential space.
\newblock Technical report, Yale University, 1976.

\bibitem{mayr1984algorithm}
E.~W. Mayr.
\newblock An algorithm for the general {P}etri net reachability problem.
\newblock {\em SIAM Journal on computing}, 13(3):441--460, 1984.

\bibitem{munkres2000topology}
J.~R. Munkres.
\newblock {\em Topology}, chapter~2, pages 108--109.
\newblock Prentice Hall Upper Saddle River, NJ, 2 edition, 2000.

\bibitem{mycielski1992games}
J.~Mycielski.
\newblock Games with perfect information.
\newblock {\em Handbook of game theory with economic applications}, 1:41--70,
  1992.

\bibitem{ovchinnikov2002max}
S.~Ovchinnikov.
\newblock Max-min representation of piecewise linear functions.
\newblock {\em Contributions to Algebra and Geometry}, 43(1):297--302, 2002.

\bibitem{royden1988real}
H.~L. Royden and P.~Fitzpatrick.
\newblock {\em Real analysis}, volume~32.
\newblock Macmillan New York, 4th edition, 1988.

\bibitem{salehi2017chemical}
S.~A. Salehi, K.~K. Parhi, and M.~D. Riedel.
\newblock Chemical reaction networks for computing polynomials.
\newblock {\em ACS Synthetic Biology}, 6(1):76--83, 2017.

\bibitem{samoilov2006deviant}
M.~S. Samoilov and A.~P. Arkin.
\newblock Deviant effects in molecular reaction pathways.
\newblock {\em Nature biotechnology}, 24(10):1235--1240, 2006.

\bibitem{seelig2009time}
G.~Seelig and D.~Soloveichik.
\newblock Time-complexity of multilayered {DNA} strand displacement circuits.
\newblock In {\em International Workshop on {DNA}-Based Computers}, pages
  144--153. Springer, 2009.

\bibitem{senum2011rate}
P.~Senum and M.~Riedel.
\newblock Rate-independent constructs for chemical computation.
\newblock {\em {PloS} one}, 6(6):e21414, 2011.

\bibitem{severson2021composable}
E.~E. Severson, D.~Haley, and D.~Doty.
\newblock Composable computation in discrete chemical reaction networks.
\newblock {\em Distributed Computing}, 34(6):437--461, 2021.
\newblock special issue of invited papers from PODC 2019.

\bibitem{SolCooWinBru08}
D.~Soloveichik, M.~Cook, E.~Winfree, and J.~Bruck.
\newblock Computation with finite stochastic chemical reaction networks.
\newblock {\em Natural Computing}, 7(4):615--633, 2008.

\bibitem{SolSeeWin10}
D.~Soloveichik, G.~Seelig, and E.~Winfree.
\newblock {D}{N}{A} as a universal substrate for chemical kinetics.
\newblock {\em Proceedings of the National Academy of Sciences}, 107(12):5393,
  2010.
\newblock Preliminary version appeared in DNA 2008.

\bibitem{srinivas2017enzyme}
N.~Srinivas, J.~Parkin, G.~Seelig, E.~Winfree, and D.~Soloveichik.
\newblock Enzyme-free nucleic acid dynamical systems.
\newblock {\em Science}, 358(6369), 2017.

\bibitem{vasic2022programming}
M.~Vasi{\'c}, C.~Chalk, A.~Luchsinger, S.~Khurshid, and D.~Soloveichik.
\newblock Programming and training rate-independent chemical reaction networks.
\newblock {\em Proceedings of the National Academy of Sciences},
  119(24):e2111552119, 2022.

\bibitem{ziegler1995polytopes}
G.~M. Ziegler.
\newblock {\em Lectures on Polytopes}, chapter~1, page~30.
\newblock Springer-Verlag New York, 1995.

\end{thebibliography}

\newpage
\appendix
\section{Forward-Invariance of Absent Siphons in Mass-Action Systems}
\label{app:siphons-absence-forward-invariant}

This section gives an alternate proof of the following result used in \Cref{sec:mass-action-reachability}, originally due to Angeli, De Leenheer, and Sontag~\cite{angeli2007petri}.

\begin{siphonsLem}[\cite{angeli2007petri}, Proposition 5.5]
\siphonsLemText
\end{siphonsLem}

\begin{proof}
Let $\calC = (\Lambda,R)$ be the CRN,
with positive mass-action rate constants assigned.

To see the forward direction,
let $\Omega$ be a siphon,
and let $\vc$ be a state such that $\Omega \cap [\vc] = \emptyset$.
Consider the reduced CRN $\calC_\reduce = (\Lambda_\reduce, R_\reduce)$ where we remove all species in $\Omega$ (i.e., $\Lambda_\reduce = \Lambda \setminus \Omega$) and all reactions referencing them
(i.e., $R_\reduce = \{ \langle \vr,\vp \rangle \in R \mid  [\vr] \cap \Omega = \emptyset \text { and } [\vp] \cap \Omega = \emptyset\}$).
Let $\vc_\reduce = \vc \upharpoonright \Lambda_\reduce$.
Let $\vrho':\Rp \to \Rp^{\Lambda_\reduce}$ be the mass-action trajectory of $\calC_\reduce$ starting at $\vc_\reduce$.
Define the trajectory $\vrho: \Rp \to \Rp^\Lambda$ of $\calC$ by $\vrho_S(t) = \vrho'_S(t)$ if $S \in \Lambda_\reduce$ and $\vrho_S(t) = 0$ otherwise, i.e.,
$\vrho$ keeps all of $\Omega$ at 0 and otherwise follows $\vrho'$.

We claim that $\vrho$ is a solution to the mass-action ODEs of $\calC$.
Since all polynomials,
such as those defining mass-action ODEs,
are locally Lipschitz 
(have bounded derivatives in some open set around every point),
the Picard-Lindel\"of Theorem implies that the mass-action ODEs have a unique solution.
Hence $\vrho$ is the \emph{only} solution to the mass-action ODEs of $\calC$ starting at $\vc$.
Since $\vrho_S(t) = 0$ for all $S \in \Omega$ and $t \geq 0$,
this implies that any $\vd$ mass-action reachable from $\vc$ obeys $\Omega \cap [\vd] = \emptyset$.

To show that $\vrho$ is a valid solution to the mass-action ODEs of $\calC$, we need to check that
\begin{equation}
\label{eq:odesfull}
\frac{d\vrho}{dt} = \vM \cdot \vA(\vrho(t))
\end{equation}
for all times $t \ge 0$,
where 
$\vA(\vrho(t))$ is the vector of reaction rates at time $t \ge 0$
(\Cref{sec:mass-action-reachability})
and
$\vM$ is the stoichiometry matrix 
(\Cref{subsec-reachability}) converting reaction rates to species derivatives.
First, let us show that $\vA_\alpha(\vrho(t)) = 0$ for all $t \geq 0$ and $\alpha \not\in R_\reduce$. 
If $\alpha \notin R_\reduce$, then there is some species in $\Omega$ that is either a reactant or product of $\alpha$, and by the fact that $\Omega$ is a siphon we know that there must necessarily be a reactant $S$ of $\alpha$ that is in $\Omega$. 
Because $\vrho_S(t) = 0$ for all $t \geq 0$ by definition and because for mass action ODEs $\vA_\alpha(\vrho(t)) = 0$ if $\vrho_S(t) = 0$ for any reactant $S$ of $\alpha$, this shows that $\vA_\alpha(\vrho(t)) = 0$ for all $t \geq 0$ and $\alpha \not\in R_\reduce$. 
Thus the right hand side of \eqref{eq:odesfull} consists entirely of the contributions of reactions in $\calC_\reduce$,  
\[\vM \cdot \vA(\vrho(t)) = \sum_{\alpha \in R_\reduce} \vM_\alpha \vA_\alpha(\vrho(t)).\]
From this and the fact that $\vrho'$ is a solution of the mass-action ODEs of $\calC_\reduce$ it follows that $\vrho$ is a solution of the mass-action ODEs for $\calC$.

To see the reverse direction,
let $\alpha = \langle \vr,\vp \rangle$ be a reaction with a product $S \in \Omega$;
it suffices to show that $\alpha$ has a reactant in $\Omega$.
If $S$ itself is a reactant in $\alpha$, we are done,
so assume otherwise;
then $\alpha$ produces $S$.

Let $\vc$ be a state with $[\vc] = \Lambda_\reduce$, i.e., exactly species not in $\Omega$ are present.
In particular $\vc(S) = 0$.
We claim that $\vA_\alpha(\vc) = 0$, i.e. $\alpha$ has rate 0 in $\vc$.
To see why, for the sake of contradiction suppose $\vA_\alpha(\vc) > 0$.
To have $\vd(S) = 0$ for all $\vd$ mass-action reachable from $\vc$,
$S$ as a function of time is the constant $0$, 
so $dS/dt = 0$ in $\vc$.
Then in order to balance $\alpha$'s production of $S$ in $\vc$ to maintain $dS/dt = 0$,
there must be some other reaction $\beta$ with $\vM(S, \beta) < 0$ (so $\beta$ consumes $S$) and $\vA_\beta(\vc) > 0$.
Since $\beta$ consumes $S$ we know that $S$ is a reactant in $\beta$.
But since $\vc(S) = 0$,
the rate of any reaction consuming $S$ is 0 in $\vc$, a contradiction.
Thus $\vA_\alpha(\vc) = 0$.

Since $[\vc] = \Lambda_\reduce$
(all species outside of $\Omega$ are present),
to have mass-action rate 0 in $\vc$,
$\alpha$ must have at least one reactant in $\Omega$ (recall reaction rate constants in a mass-action system are strictly positive).
So $\Omega$ is a siphon.
\end{proof}

The definition of a valid rate schedule (\Cref{defn:valid-rate-schedule}, part \eqref{defn:valid-rate-schedule:reactants-present}) requires that if some reactant is 0, then the reaction rate is 0.
However, the converse implication
(if a reaction rate is 0, then some reactant must be 0) 
holds for mass-action but not more general rate schedules such as segment-reachability,
which are allowed to ``starve'' applicable reactions by holding their rates at 0.
The reverse direction of the proof of \Cref{lem-mass-action-forward-invariant-siphon} uses this converse implication,
but the forward direction uses only the more general implication of \Cref{defn:valid-rate-schedule}, part \eqref{defn:valid-rate-schedule:reactants-present}.
For the forward direction,
the key property used from mass-action is its determinism: 
it has unique solutions,
so to show that the siphon remains absent in all possible reachable states
it suffices to show that there is just one solution where the siphon remains absent.

\section{Max-min representation of continuous piecewise linear functions}
\label{app:A}

Here we prove a slight generalization of Ovchinnikov's theorem~\cite{ovchinnikov2002max}. 
In Ovchinnikov's original paper, he only considers piecewise affine functions 
(in Ovchinnikov's terminology, piecewise ``linear'' functions) that are defined on closed domains (that is, closures of open subsets of $\R^n$). 
However, the key proof techniques of~\cite{ovchinnikov2002max} did not crucially use this fact.
In fact, we apply \cref{thm-min-max-representation} on non-closed domains such as the sets $D_U$ in the proof of Lemma~\ref{lem-nonnegative-piecewise-linear-implies-computable}.
For completeness we prove the variant of the theorem not requiring $D$ to be closed.

\begin{minMaxThm}
[\cite{ovchinnikov2002max}, Theorem 2.1]
\minMaxThmText
\end{minMaxThm}

In order to prove the theorem, we  first prove three lemmas.
The first technical lemma is implicit in~\cite{ovchinnikov2002max}.
The second and third lemmas correspond to Lemmas 2.1 and 2.2 of~\cite{ovchinnikov2002max}. The proofs we give of the second and third lemmas are almost identical in content to the proofs of the corresponding lemmas in~\cite{ovchinnikov2002max}, except for the fact that we consider piecewise affine functions defined over more general subsets of $\R^n$. The same is true for our proof of \Cref{thm-min-max-representation}, which is again almost identical to the proof of Theorem 2.1 in~\cite{ovchinnikov2002max}. 

\begin{lemma}\label[lemma]{lem:piecewise-linear-functions-are-nice}
If $f: [a,b] \to \R$ is a continuous piecewise affine function with components $\set{g_1, \ldots, g_n}$, then there are finitely many numbers $a = x_0 < x_1 < \ldots < x_m = b$ such that $f$ is affine on $[x_k, x_{k + 1}]$ for all $k$.
\end{lemma}

\begin{proof}
Without loss of generality we can assume that all of the component functions $g_i$ are distinct affine functions. For each $i$ between $1$ and $n$, let $D_i$ be the set of $x \in [a,b]$ such that $g_i(x) = f(x)$. For each $i$, both $g_i$ and $f$ are continuous, so $D_i$ is closed. Let $S$ be the subset of $[a,b]$ consisting of points $x \in [a,b]$ where $x$ is a member of more than one $D_i$. For each pair $i \neq j$, we know that $g_i$ and $g_j$ are distinct affine functions, so there can be at most one $x \in [a,b]$ such that $g_i(x) = g_j(x)$. This implies that $D_i$ and $D_j$ can intersect in at most one point, so $S$ must be a finite set. Let $x_0\ldots x_m$ be the elements of $S \cup \set{a,b}$. 

Now for a given $k$ write $I$ for the interval $(x_k, x_{k + 1})$ and consider the restriction of $f$ to $I$. Pick a random point $c \in I$ and suppose $f(c) = g_l(c)$. Then clearly $D_l \cap I$ is nonempty. Because $D_l$ is closed in $[a,b]$, by definition $D_l \cap I$ is closed relative to $I$. Because $S \cap I = \emptyset$ we also know that 
\[D_l \cap I = I \setminus \left(\bigcup_{i \neq l} D_i \right) \]
so $D_l \cap I$ is open relative to $I$. But the only subset of an interval that is both open and closed is the whole interval, so $I \subseteq D_l$. Therefore $f = g_l$ when restricted to $I$, and by continuity we see that $f$ is affine on the closure of $I$ as well.
\end{proof}

Note that we define piecewise affine functions to have only finitely many components---without this assumption the above lemma is false.

\begin{lem}\label[lemma]{lem:piecewise-linear-on-interval}
Let $f: [a,b] \to \R$ be a continuous piecewise affine function. Let $\set{g_1, \ldots, g_n}$ be its set of components. Then there is some $k$ such that
\[g_k(a) \le f(a)\ \text{ and }\ g_k(b) \ge f(b).\]
\end{lem}

\begin{proof}
We'll first prove the result for $f(a) = f(b) = 0$ and then show how this implies the general case. Given this assumption, if one of the $g_i$ is the zero function, then we're done. If not, since all of the component functions $g_i$ are affine, each $g_i$ can have at most one zero. Since $f$ has finitely many components, this implies that $f$ has finitely many zeros. Let $c$ be the smallest number such that $c > a$ and $f(c) = 0$. 

By \Cref{lem:piecewise-linear-functions-are-nice} we know that there are some $x$ and $y$ with $a < x < y < c$ and component functions $g_k$ and $g_l$ such that $f = g_k$ on $[a, x]$ and $f = g_l$ on $[y, c]$. If the slope of either $g_k$ or $g_l$ is non-negative, then we are done:
\begin{align*}
     g_k(a) &\le g_k(c) = 0 = f(a)\\
     g_k(b) &\ge g_k(c) = 0 = f(b) 
\end{align*}
and similarly for $g_l$. 
But $g_k$ and $g_l$ can't both have negative slope, for then $f(x) = g_k(x) < 0$ and $f(y) = g_l(y) > 0$, so by the intermediate value theorem there would be some $z$ between $x$ and $y$ such that $f(z) = 0$. This contradicts our assumption that $c$ was the smallest number with $c > a$ and $f(c) = 0$. 
This concludes the proof assuming that $f(a)=f(b)=0.$

To deduce the result for a general continuous piecewise affine function from this special case, subtract the affine function
\[\ell(x) = f(a) + f(b)\frac{x - a}{b - a}\]
from $f$ and all of its components. 
\end{proof}

\begin{lem}\label[lemma]{lem:piecewise-linear-on-convex-domain}
Let $D$ be a convex subset of $\R^n$ and let $f: D \to \R$ be a continuous piecewise affine function. If the components of $f$ are $\set{g_1 \ldots g_n}$, then for every pair of vectors $\va$ and $\vb$ in $D$, there is some $k$ such that
\[g_k(\va) \le f(\va)\ \text{ and }\ g_k(\vb) \ge f(\vb)\]
\end{lem}
\begin{proof}
Because $D$ is convex, the straight-line interval between $\va$ and $\vb$ is contained in $D$. Apply \Cref{lem:piecewise-linear-on-interval} to the restriction of $f$ to this interval. 
\end{proof}

Finally, we are ready to prove \Cref{thm-min-max-representation}.

\begin{proof}[Proof of \Cref{thm-min-max-representation}]
For each $\vb \in D$, define the set $S_{\vx} \subseteq \set{1 \ldots p}$ as
\[S_\vb = \set{i\ |\ g_i(\vb) \ge f(\vb)}.\]
Let 
\[F_\vb(\vx) = \min_{i \in S_\vb} g_i(\vx).\]
Because there is always some component function $g_j$ with $g_j(\vb) = f(\vb)$, we see that $F_\vb(\vb) = f(\vb)$ for every $\vb \in D$. Also, by \Cref{lem:piecewise-linear-on-convex-domain}, we know that for every $\va \in D$, there is some component function $g_k \in S_\vb$ with $g_k(\va) \le f(\va)$, so $F_\vb(\va) \le f(\va)$ for every pair $\vb, \va \in D$. This implies that
\begin{align}\label{eq:max-min}
f(\vx) = \max_{\vb \in D}F_\vb(\vx) = \max_{\vb \in D}\min_{i \in S_\vb}g_i(\vx)    
\end{align}
Since $\set{1 \ldots p}$ is a finite set, it has only finitely many subsets, so each $S_\vb$ is equal to one of finitely many sets $S_j$. We can therefore replace the maximum over all $\vb \in D$ in \Cref{eq:max-min} with a maximum over finitely many functions. 
\end{proof}

\section{Finding Rational Solutions to Systems of Linear Equations}
\label{app:C}

It is well-known that a system of linear equations with rational coefficients has a rational solution if and only if it has a real solution.
The following result 
shows the slightly generalized claim that
rational solutions exist arbitrarily close to all real solutions
(i.e., the rational solutions are dense in the real solutions).

\begin{lem}\label{lem-rational-sols}
Let $A\vec{x} = \vec{b}$ be a system of linear equations, where $A$ is a matrix with rational coefficients and $\vec{b}$ is a vector with rational coefficients. If the equation has a solution $\vec{x}$ with real coefficients, then for any $\varepsilon > 0$, it has a solution with $\vec{x}'$ rational coefficients such that $||\vec{x}' - \vec{x}|| < \varepsilon$. 
\end{lem}

\begin{proof}
Let $n$ be the number of rows of $A$ and the length of $\vec{b}$. Let $m$ be the number of columns of $A$ and the length of $\vec{x}$. Because $A$ has rational entries, using elementary row and column operations it can be decomposed as $PNQ$ where $P$ is an $m \times m$ invertible rational matrix, $Q$ is an $n \times n$ invertible rational matrix, and 
\[N_{ij} = \begin{cases}
1 & i = j \text{ and } i \le r \\
0 & \text{otherwise}
\end{cases}\]
where $r$ is the rank of $M$. Let $\vec{y} = Q\vec{x}$ and let $\vec{c} = P^{-1} \vec{b}$, so that $N\vec{y} = \vec{c}$. If $r = n$, then all of the entries of $\vec{y}$ must be rational, since $\vec{y}_i = \vec{c}_i$ for all $i$ and all of the entries of $\vec{c}_i$ are rational. Then the entries of $\vec{x}$ must also all be rational, since $\vec{x} = Q^{-1}\vec{y}$ and $Q$ has all rational entries. As a result, if $r = n$, we can just take $\vec{x}' = \vec{x}$. 

On the other hand, if $r < n$, then let 
\[\delta = \frac{\varepsilon}{\sqrt{n - r} ||Q^{-1}||}\]
where 
\[||Q^{-1}|| = \sup_{\vec{v} \neq 0} \frac{||Q^{-1}\vec{v}||}{||\vec{v}||}\]
is the operator norm of $||Q^{-1}||$. Now let $\vec{y}'$ be a vector such that $\vec{y}'_i = \vec{y}_i$ for $i \le r$ and $\vec{y}'_i$ is a rational number such that $|\vec{y}'_i - \vec{y}_i| < \delta$ for $i > r$. All of the components of $\vec{y}'$ are rational: $\vec{y}'_i$ is rational by construction for $i > r$, and $\vec{y}'_i = \vec{y}_i = \vec{c}_i$ is rational for $i \le r$. Moreover, the fact that $\vec{y}'_i = \vec{c}_i$ for $i \le r$ shows that $N\vec{y}' = \vec{c}$. 

If we take $\vec{x}' = Q^{-1}\vec{y}'$, then all of the components of $\vec{x}'$ are rational, and $M\vec{x}' = \vec{b}$, since
\begin{align*}
M\vec{x}' = PNQ\vec{x}' = PN\vec{y}' = P\vec{c} = \vec{b}.
\end{align*}
Finally, we know that $||\vec{x}' - \vec{x}|| < \varepsilon$, since
\begin{align*}
||\vec{x}' - \vec{x}|| &= ||Q^{-1}(\vec{y}' - \vec{y})|| \\
&\le ||Q^{-1}|| \cdot ||\vec{y}' - \vec{y}|| \\
&= ||Q^{-1}|| \sqrt{\sum_{i = 1}^n (\vec{y}'_i - \vec{y}_i)^2} \\
&< ||Q^{-1}|| \sqrt{(n - r)\delta^2} \\
&= \varepsilon.
\end{align*}
This shows that $\vec{x}'$ is our desired solution. 
\end{proof}

\section{Bounding reaction fluxes in straight-line reachability}
\label{app:D}

This section is devoted to proving Lemma~\ref{lem:bound-reaction-fluxes-straight-line-reachability}.
Intuitively,
it shows that if a CRN can reach from state $\vc$ to state $\vd$ by a straight line (of length $\| \vd - \vc \|$),
then the reaction fluxes required can be bounded by $O(\| \vd - \vc \|)$.
This is nontrivial since one can have reactions that cancel, e.g.,
$X \to Y$ and $Y \to X$.
The same straight line from $\vc$ to $\vd$ could result from arbitrarily large but equal fluxes of each reaction (plus some other reactions).
Lemma~\ref{lem:bound-reaction-fluxes-straight-line-reachability} states that we never \emph{need} arbitrarily large fluxes to get from $\vc$ to $\vd$.

\begin{definition}
A \emph{convex polyhedral cone} $C$ is a subset of a vector space $V$ such that there exist vectors $v_1, \ldots, v_k \in V$ so that $C$ is exactly the set of $x \in V$ that can be written as $x = \sum_i \lambda_iv_i$ with all $\lambda_i \ge 0$. Such a set $S = \set{v_1 \ldots v_k}$ is called a \emph{spanning set} for $C$ and we say that $C$ is \emph{spanned by} $S$. Given a set $S = \set{v_1 \ldots v_k}$ of vectors in $V$ we write $C_S$ for the convex polyhedral cone spanned by $S$. 
\end{definition}

\begin{lemma}\label{lem:cones-spanned-by-linearly-independent-frames}
Let $C$ be a convex cone with spanning set $S = \set{v_1 \ldots v_k}$. Let $\mathcal{I}$ be the collection of all linearly independent subsets of $S$. Then
\[C = \bigcup_{S' \in \mathcal{I}} C_{S'}\]
\end{lemma}

\begin{proof}
Suppose to the contrary that there was some $x \in C$ not contained in $C_{S'}$ for any $S' \in \mathcal{I}$. Let $T$ be a minimal subset of $S$ such that $x \in C_T$. Then since $x \in C_T$ we can write $x$ as 
\[x = \sum_{v_i \in T} \lambda_i v_i.\]
To produce a contradiction, let us show that we can express $x$ as 
\[x = \sum_{v_i \in T} \lambda_i' v_i\]
with some $\lambda'_i = 0$. This will imply that $x \in C_{T \setminus \set{v_i}}$, contradicting the minimality of $T$. By assumption, since $x \in T$ we know that $T \notin \mathcal{I}$, so there is some nontrivial linear relationship
\[\sum_{v_i \in T}\mu_i v_i = 0\]
among the $v_i \in T$. By negating all of the $\mu_i$ if needed, we can assume that at least one $\mu_i < 0$. Let $C$ be the constant 
\[C = \min_{\mu_i < 0} \frac{-\lambda_i}{\mu_i}.\]
Let us show that $\lambda_j + C\mu_j \ge 0$ for all $j$. Because $\lambda_i > 0$ for all $i$ we know that $C > 0$. As a result, if $\mu_j \ge 0$ then necessarily $\lambda_j + C\mu_j \ge 0$. On the other hand, if $\mu_j < 0$, then $C \le -\lambda_j/\mu_j$, so 
\[\lambda_j + C\mu_j \ge \lambda_j + \left(\frac{\lambda_j}{\mu_j}\right)\mu_j = 0\]
Moreover, for some $\mu_{i_0} < 0$ we know that $\lambda_{i_0} + C\mu_{i_0} = 0$. As a result,
\[x = \left(\sum_{v_i \in T} \lambda_i v_i\right) + C\left(\sum_{v_i \in T}\mu_i v_i\right) = \sum_{v_i \in T} (\lambda_i + C\mu_i) v_i = \sum_{v_i \in T} \lambda'_i v_i\]
where $\lambda'_{i_0} = 0$. Thus $x \in C_{T\setminus v_{i_0}}$, contradicting the minimality of $T$ and therefore the existence of $x$. 
\end{proof}

\begin{lemma}\label{lem:smallness-lemma-cones}
Let $C$ be a convex cone with spanning set $S = \set{v_1\ldots v_k}$. Then there is a constant $K$ depending only on $S$ so that for any $x \in C$, there is some representation of $x$ as 
\[x = \sum_i \lambda_i v_i\]
with $0 \le \lambda_i \le K||x||$. Additionally, if there is some subset $T \subseteq S$ such that $x \in C_T$, then the above represenation can be chosen with $\lambda_i = 0$ for any $v_i \notin T$. 
\end{lemma}

\begin{proof}
Let $\mathcal{I}$ be the collection of linearly independent subsets of $S$. If $S' \in \mathcal{I}$ and $x \in C_{S'}$ then because the vectors $v_1 \ldots v_k$ in $S'$ are linearly independent, there is a unique way to write  
\[x = \sum_{v_i \in S'} \lambda_i v_i\]
and by \cite[Lemma 2.4-1]{kreyszig1991introductory}, we know there is some constant $K_{S'}$ so that $0 \le \lambda_i \le K_{S'} \|x\|$. Take $K = \max_{S' \in \mathcal{I}} K_{S'}$. Applying Lemma~\ref{lem:cones-spanned-by-linearly-independent-frames} to $C_T$, we can find some linearly independent collection of vectors $T'$ such that $T' \subseteq T$ and $x \in C_{T'}$. Then
\[x  = \sum_{v_i \in T'}\lambda_iv_i\]
with $0 \le \lambda_i \le K_{T'}\| x\|$. Because $T' \subseteq T \subseteq S$ we know that $T' \in \mathcal{I}$, so $K_{T'} \le K$, and since $T' \subseteq T$ we know that above sum only ranges over vectors $v_i \in T$. 
\end{proof}

Finally, we have the main result of this section.

\begin{lemma}
\label{lem:bound-reaction-fluxes-straight-line-reachability}
Fix a CRN $\calC$ and suppose a flux vector $\vu \in \Rp^R$ is applicable at a state $\vc \in \Rp^\Lambda$. Let $\ell = ||M\vu||$ be the length of the straight-line segment in $\R^\Lambda$ given by $M\vu$. Then there exists a constant $K$, depending only on $\calC$ and independent of $\vu$ and $\vc$, such that there exists a flux vector $\vu'$ where $\vu'$ is also applicable at $\vc$, the length of $\vu'$ is bounded as $||\vu'|| \le K\ell$, and $M\vu = M\vu'$.
\end{lemma}

\begin{proof}

Let $\set{v_1\ldots v_m}$ be the standard basis vectors of $\R^R$. Because $\vu = \sum_i \lambda_i v_i$ with all $\lambda_i \ge 0$, we know that $M\vu$ is contained in the convex polyhedral cone spanned by $\set{Mv_1 \ldots Mv_m}$. By Lemma~\ref{lem:smallness-lemma-cones} we know that there are $\lambda'_i$ such that $M\vu = \sum_i \lambda'_i Mv_i$ and $0 \le \lambda'_i \le K||M\vu||$, and moreover $\lambda_i' = 0$ whenever $\lambda_i = 0$. Let $\vu' = \sum_i \lambda'_iv_i$. Then $\vu'$ is still applicable at $\vc$ because $\lambda'_i > 0$ implies $\lambda_i > 0$ and we assumed that $\vu$ was applicable at $\vc$. Also, note that
\[M\vu'= \sum_i \lambda'_i Mv_i = M\vu.\]
Finally, we know that 
\[||\vu'|| \le \sum_i \lambda'_i \le \sum_i K||M\vu|| \le mK\ell,\]
so $\vu'$ is our desired vector in $\Rp^R$. 
\end{proof}

\section{Partial States and Reachability}
\label{sec:partial-reachability}

In this section we define a notion of ``partial'' states and reachability,
which are used in the proof of \cref{lem:not-stably-compute-implies-can-reach-far-from-correct}.
If $\Delta \subsetneq \Lambda$,
we say $\vp \in \Rp^\Delta$ is a \emph{partial state}.
Given a state $\vc \in \Rp^\Lambda$,
recall that $\vc \upharpoonright \Delta$ is $\vc$ restricted to $\Delta$,
i.e., the partial state $\vp = \vc  \upharpoonright \Delta \in \Rp^\Delta$ such that $\vp(S) = \vc(S)$ for all $S \in \Delta$.

Let $k \in \N \cup \{\infty\}.$
Given a state $\vc \in \Rp^\Lambda$ and a partial state  $\vp \in \Rp^\Delta$,
we write  $\vc \segto^k \vp$ 
if there is a sequence of states $\vb_0, \dots, \vb_{k} \in \Rp^\Lambda$ such that $\vc  = \vb_0 \slto \vb_1 \slto \vb_2 \slto \dots  \slto \vb_{k}$,
with $\vp = \vb_k \upharpoonright \Delta$ 
if $k \in \N$,
or $\vp = \lim\limits_{i \to \infty} (\vb_i \upharpoonright \Delta)$ if $k = \infty$.
We write \emph{$\vc \segto \vp$ via $\seq{\vb_i}_{i=1}^k$} if $\vc \segto^k \vp$ for some $k \in \N \cup \{\infty\}$ with intermediate states $\vb_0,\vb_1,\dots$ as above,
or simply $\vc \segto \vp$ when the intermediate states $\vb_0,\vb_1,\dots$ are implicit.
We write \emph{$\vc \segtoio \vp$ via $\seq{\vb_i}_{i=1}^k$} if $\vc = \vb'_0 \to^1 \vb'_1 \to^1 \dots$,
where for some subsequence $\vb_1,\vb_1,\dots$ of $\vb'_1,\vb'_1,\dots$,
we have $\vp = \lim\limits_{i\to\infty} (\vb_i \upharpoonright \Delta)$,
i.e., an infinite subsequence of states converges on concentrations in $\Delta$.

Note that if there is a state $\vd$ such that $\vc \segto \vd$ and $\vd \upharpoonright \Delta = \vp$, then $\vc \segto \vp$, but it is not apparent from the definition that this is the only way for a partial state to be reachable if the number of line segments is infinite.
In particular, it could be that the sequence $\vb_0, \vb_1, \ldots$ does not converge to any state (even though the partial states $\vb_0 \upharpoonright \Delta, \vb_1 \upharpoonright \Delta, \ldots$ converge to $\vp$), 
if some concentration values outside of $\Delta$ do not converge (for instance they may oscillate or go to infinity).

Our goal now is to show that in fact, if a partial state $\vp$ is reachable (or even merely $\segtoio$ reachable), then there is a particular state reachable whose restriction to $\Delta$ is $\vp$.

\begin{thm}
\label{thm:partial-state-reachable}
Let $\Delta \subsetneq \Lambda$,
let $\vc \in \Rp^\Lambda$ be a state, and $\vp \in \Rp^\Delta$ be a partial state.
If $\vc \segtoio \vp$, then 
there is a state $\vd$ such that $\vc \segto \vd$ and $\vp=\vd \upharpoonright \Delta$.

Furthermore,
if $\vc \segtoio \vp$ via $\seq{\vb_i}_{i=1}^k$ for $k \in \N \cup \{\infty\}$,
then there is a partition of $\Lambda$ into $\Lambda_\bdd$ and $\Lambda_\unbdd$, 
with $\Delta \subseteq \Lambda_\bdd$, 
and a subsequence $\seq{\vr_i}_i$ of $\seq{\vb_i}_i$ so that 
$\lim_{i\to\infty} \vr_i(S) = \infty$ for all $S \in \Lambda_\unbdd$, 
and $\lim_{i\to\infty}\vr_i(S) = \vd(S)$ for all $S \in \Lambda_\bdd$.
\end{thm}

\begin{proof}
The finite case is immediate from the definition of $\segto$ for partial states (choose $\vd = \vb_k$ and $\Lambda_\unbdd = \emptyset$),
so assume $\vc \segtoio \vp$.
Then there is an infinite sequence of states 
(the converging subsequence in the definition of $\segtoio$)
$\vb_0, \vb_1, \dots \in \Rp^\Lambda$ 
such that $\vc = \vb_0$,
each 
$\vb_i \segto^{k_i} \vb_{i+1}$ for some $k_i \in \N$,
and $\vp = \lim_{i \to \infty} \vb_i \upharpoonright \Delta$.

Let $\Gamma = \Lambda \setminus \Delta$ be the species outside of $\Delta$,
which may not be converging in the subsequence $(\vb_i)_i$.
Intuitively, our goal will be to make some species in $\Gamma$ converge, and all others simultaneously diverge to infinity.
Then, by driving the concentrations of the diverging species sufficiently large,
then removing them from the system,
we obtain a reduced CRN where all species converge to a single state.
We apply \Cref{thm-reachable-segment-bound} to this reduced CRN to find a finite path of $m+1$ segments reaching this state.
Finally we argue that this path is applicable even in the original CRN,
because the species that were removed first had their concentrations driven large enough that the finite path does not have sufficient flux to send any of them to $0$.
Thus at the end of this finite path, all species converge to some concentration.

Formally, define the set $\Gamma_\unbdd$ (the ``unbounded'' species in $\Gamma$) iteratively as follows.
If there is any $S_1 \in \Gamma$ such that $\limsup_{i\to\infty} \vb_i(S_1) = \infty$, 
then put $S_1$ in $\Gamma_\unbdd$, otherwise $\Gamma_\unbdd$ is defined to be $\emptyset$.
Then pick a subsequence $\vb'_0, \vb'_1,\ldots$ of $\vb_0,\vb_1,\ldots$ such that, for all $i\in\N$, $\vb'_i(S_1) \geq i$;
such a subsequence exists since $\limsup_{i\to\infty} \vb_i(S_1) = \infty$.
Now, from \emph{that} subsequence,
if any species $S_2$ obeys $\limsup_{i\to\infty} \vb'_i(S_2) = \infty$, then place $S_2$ in $\Gamma_\unbdd$ and choose a subsequence $\vb''_0,\vb''_1,\ldots$ of $\vb'_0,\vb'_1,\ldots$ where $\vb''_i(S_2) \geq i$ for each $i \in \N$;
note that because $\vb''_i = \vb'_{i'}$ for $i' \geq i$,
we also have $\vb''_i(S_1) \geq i$.

Repeat this alternation of choosing a species $S_j$ to put in $\Gamma_\unbdd$ and picking a subsequence, until in the final subsequence $\vb'''_0,\vb'''_1,\ldots$ every remaining species $S$ obeys $\limsup_{i\to\infty} \vb'''_i(S) < \infty$.
Define the set $\Gamma_\bdd = \Gamma \setminus \Gamma_\unbdd$;
by construction each species in $\Gamma_\bdd$ has bounded concentrations in $\vb'''_0,\vb'''_1,\dots$
(Though some may have had unbounded concentrations in the original sequence $\vb_0,\vb_1,\dots$)

Now, since concentrations of species in $\Gamma_\bdd$ are in a closed, bounded (i.e., compact) set,
there is a subsequence $\vr_0,\vr_1,\dots$ of $\vb'''_0,\vb'''_1,\dots$ that converges on concentrations for species in $\Gamma_\bdd$.
Furthermore, all subsequences of $\vb_0,\vb_1,\dots$ we have taken so far can be assumed without loss of generality to contain $\vb_0$.
Thus the subsequence $\vr_0,\vr_1,\dots$ obeys
\begin{enumerate}
    \item
    \label{first-state-in-subsequence}
    $\vr_0 = \vb_0\ (= \vc)$,
    
    \item
    \label{gamma-inf-increases}
    for all $S \in \Gamma_\unbdd$ and $i\in \N$,
    $\vr_i(S) \geq i$
    (concentrations in $\Gamma_\unbdd$ increase to infinity simultaneously), 
    and

    \item
    \label{gamma-bdd-converge}
    for all $S \in \Gamma_\bdd \cup \Delta$,
    $\lim_{i \to \infty} \vr_i(S)$ exists and is finite
    (all other concentrations converge).
\end{enumerate}

\newcommand{\quo}{\mathrm{q}}

Consider the ``quotient'' CRN $\calC_\quo = (\Lambda_\quo, R_\quo)$,
where $\Lambda_\quo = \Lambda \setminus \Gamma_\unbdd = \Delta \cup \Gamma_\bdd$, 
and $R_\quo$ is defined by taking each reaction from $R$ and removing any species from it in $\Gamma_\unbdd$.
For example, if $A,B \in \Gamma_\unbdd$ and $C,D,E \not\in \Gamma_\unbdd$,
then the reaction 
$A + B + 2C \to 2A + D + E$ becomes $2C \to D+E$.
Then for each state $\vr \in \Rp^\Lambda$ of the original CRN $\calC$, 
its partial state $\vr \upharpoonright \Lambda_\quo$ is a (normal) state of $\calC_\quo.$

For each $i\in\N$, let $\vq_i = \vr_i \upharpoonright \Lambda_\quo$ be the state of $\calC_\quo$ corresponding to $\vr_i$.
Note that $\vq_0 = \vc \upharpoonright \Lambda_\quo$.
By the definition of $\Lambda_\quo = \Delta \cup \Gamma_\bdd$ and the convergence of concentrations in both $\Delta$ (by the definition of $\vc \segto \vp$) and $\Gamma_\bdd$ (shown as part~\eqref{gamma-bdd-converge} above),
the sequence $\vq_0, \vq_1,\dots$ converges to some state $\vq \in \Rp^{\Lambda_\quo}$.
This implies that $\vq_i \segtoio \vq$ for each $i \in \N$.
Note that $\vq \upharpoonright \Delta = \vp$.

Since the $\vq_i$'s converge to $\vq$,
there is some $i_0$ such that,
for all $i \geq i_0$,
$\| \vq - \vq_i \| \leq 1$.
Since each $\vq_i \segtoio \vq$, 
by \Cref{thm:io-reachable-segment-flux-bound},
$\vq_i \segto^{m+1} \vq$,
where $m = \min\{ |\Lambda_\quo|, |R_\quo| \}$,
and for some constant $K$ depending only on $\calC_\quo$,
the total reaction fluxes do not exceed $K$ along the entire path of $m+1$ segments.

Let $C = \max \{ \vr'(S) \mid S \in \Gamma_\unbdd, \langle \vr',\vp' \rangle \in R \}$ be the maximum reactant coefficient of any reactant in $\Gamma_\unbdd$ for any reaction in $R$.
Choose $i_1 = \max\{ C \cdot K, i_0 \}$.
Then $\vq_{i_1} \segto \vq$ with total reaction flux at most $K$,
since $i_1 \geq i_0$.

Consider running the same reaction fluxes on the original CRN, 
from the state $\vb_{i'}$,
choosing $i'$ such that $\vb_{i'} \upharpoonright \Lambda_\quo = \vq_{i_1}$.
By property~\eqref{gamma-inf-increases} above,
there is sufficient concentration $C \cdot K$ of each reactant in $\Gamma_\unbdd$ in state $\vb_{i'}$ for these reactions to remain applicable along the entire path.
Let $\vd \in \Rp^\Lambda$ be the state reached at the end of this path, then $\vb_{i'} \segto \vd$.
Note that $\vd \upharpoonright \Lambda_\quo = \vq$, 
and recall $\vq \upharpoonright \Delta = \vp$,
so $\vd \upharpoonright \Delta = \vp$.
Since $\vc \segto \vb_{i'}$,
by transitivity of $\segto$ we have $\vc \segto \vd$, 
proving the theorem.
\end{proof}

Note that by \Cref{thm-reachable-segment-bound},
if $\vc \segto \vp$ for a partial state $\vp$, 
then $\vc \segto^{m+1} \vp$, 
where $m = \min\{|\Lambda|,|R|\}$.

\end{document}